\documentclass[11pt,american]{article}
\usepackage[latin9]{inputenc}
\usepackage[letterpaper]{geometry}
\usepackage{babel}
\usepackage{array}
\usepackage{verbatim}
\usepackage{float}
\usepackage{booktabs}
\usepackage{textcomp}
\usepackage{mathtools}
\usepackage{amsmath}
\usepackage{amsthm}
\usepackage{amssymb}
\usepackage{graphicx}
\usepackage{setspace}
\usepackage[unicode=true,pdfusetitle,
 bookmarks=true,bookmarksnumbered=true,bookmarksopen=true,bookmarksopenlevel=3,
 breaklinks=true,pdfborder={0 0 0},pdfborderstyle={},backref=false,colorlinks=false]
 {hyperref}

\makeatletter

\newcommand*\LyXZeroWidthSpace{\hspace{0pt}}
\providecommand{\tabularnewline}{\\}
\floatstyle{ruled}
\newfloat{algorithm}{tbp}{loa}
\providecommand{\algorithmname}{Algorithm}
\floatname{algorithm}{\protect\algorithmname}

\numberwithin{figure}{section}
\theoremstyle{definition}
\newtheorem{defn}{\protect\definitionname}
\theoremstyle{plain}
\newtheorem{lem}{\protect\lemmaname}
\theoremstyle{plain}
\newtheorem{thm}{\protect\theoremname}
\theoremstyle{plain}
\newtheorem{assumption}{\protect\assumptionname}
\theoremstyle{plain}
\newtheorem{prop}{\protect\propositionname}
\theoremstyle{plain}
\newtheorem{fact}{\protect\factname}
\theoremstyle{plain}
\newtheorem*{cor*}{\protect\corollaryname}

\usepackage{amsmath}
\usepackage{listings}
\usepackage{tikz}
\usetikzlibrary {graphs,arrows,patterns,plotmarks,backgrounds,fit}

\def\gpchecktikzversion#1.#2\relax{%
\ifnum#1<2%
  \errmessage{PGF/TikZ version >= 2.0 is required!}%
\fi}
\expandafter\gpchecktikzversion\pgfversion\relax

\def\gnuplot@output@ps{ps} 

\expandafter\def\csname gnuplot@select@driver@pgfsys-dvi.def\endcsname     {ps}
\expandafter\def\csname gnuplot@select@driver@pgfsys-dvipdfm.def\endcsname {pdf}
\expandafter\def\csname gnuplot@select@driver@pgfsys-dvipdfmx.def\endcsname{pdf}
\expandafter\def\csname gnuplot@select@driver@pgfsys-dvips.def\endcsname   {ps}
\expandafter\def\csname gnuplot@select@driver@pgfsys-pdftex.def\endcsname  {pdf}
\expandafter\def\csname gnuplot@select@driver@pgfsys-luatex.def\endcsname  {pdf}
\expandafter\def\csname gnuplot@select@driver@pgfsys-tex4ht.def\endcsname  {html}
\expandafter\def\csname gnuplot@select@driver@pgfsys-textures.def\endcsname{ps}
\expandafter\def\csname gnuplot@select@driver@pgfsys-vtex.def\endcsname    {ps}
\expandafter\def\csname gnuplot@select@driver@pgfsys-xetex.def\endcsname   {pdf}

\ifcsname gnuplot@select@driver@\pgfsysdriver\endcsname
  \edef\gnuplot@output@mode{\csname gnuplot@select@driver@\pgfsysdriver\endcsname}
\else
  \errmessage{The driver \pgfsysdriver\space is not supported by gnuplot-lua-tikz}%
\fi

\ifcsname PackageWarning\endcsname
  \let\gnuplot@packagewarning\PackageWarning
\else
  \def\gnuplot@packagewarning#1#2{\immediate\write-1{Package #1 Warning: #2}}
\fi

\def\gp@rawimage@pdf#1#2#3#4#5#6{%
  \def\gp@tempa{cmyk}%
  \def\gp@tempb{#1}%
  \ifx\gp@tempa\gp@tempb%
    \def\gp@temp{/CMYK}%
  \else%
    \def\gp@temp{/RGB}%
  \fi%
  \pgf@sys@bp{#4}\pgfsysprotocol@literalbuffered{0 0}\pgf@sys@bp{#5}%
  \pgfsysprotocol@literalbuffered{0 0 cm}%
  \pgfsysprotocol@literalbuffered{BI /W #2 /H #3 /CS \gp@temp}%
  \pgfsysprotocol@literalbuffered{/BPC 8 /F /AHx ID}%
  \pgfsysprotocol@literal{#6 > EI}%
}
\def\gp@rawimage@ps#1#2#3#4#5#6{%
  \def\gp@tempa{cmyk}%
  \def\gp@tempb{#1}%
  \ifx\gp@tempa\gp@tempb%
    \def\gp@temp{4}%
  \else%
    \def\gp@temp{3}%
  \fi%
  \pgfsysprotocol@literalbuffered{0 0 translate}%
  \pgf@sys@bp{#4}\pgf@sys@bp{#5}\pgfsysprotocol@literalbuffered{scale}%
  \pgfsysprotocol@literalbuffered{#2 #3 8 [#2 0 0 -#3 0 #3]}%
  \pgfsysprotocol@literalbuffered{currentfile /ASCIIHexDecode filter}%
  \pgfsysprotocol@literalbuffered{false \gp@temp\space colorimage}%
  \pgfsysprotocol@literal{#6 >}%
}
\def\gp@rawimage@html#1#2#3#4#5#6{%
  \gnuplot@packagewarning{gnuplot-lua-tikz}{%
    \noexpand\gp@rawimage is not implemented for \gnuplot@output@mode\space output.%
  }%
}

\expandafter\let\expandafter\gp@rawimage\csname gp@rawimage@\gnuplot@output@mode\endcsname

\def\gploadimage#1#2#3#4#5{%
  \pgftext[left,bottom,x=#1cm,y=#2cm] {\pgfimage[interpolate=false,width=#3cm,height=#4cm]{#5}};%
}

\def\gp@set@size#1{%
  \def\gp@image@size{#1}%
}

\def\gp@rawimage@#1#2#3#4#5#6#7#8{%
  \tikz@scan@one@point\gp@set@size(#6,#7)\relax%
  \tikz@scan@one@point\pgftransformshift(#2,#3)\relax%
  \pgftext{%
    \pgfsys@beginpurepicture%
    \gp@image@size
    \gp@rawimage{#1}{#4}{#5}{\pgf@x}{\pgf@y}{#8}%
    \pgfsys@endpurepicture%
  }%
}

\def\gprawimage#1#2#3#4#5#6#7#8#9{%
  \ifx&#9&%
    \gp@rawimage@{#1}{#2}{#3}{#4}{#5}{#6}{#7}{#8}%
  \else
    \ifx\gnuplot@output@mode\gnuplot@output@ps
      \gp@rawimage@{#1}{#2}{#3}{#4}{#5}{#6}{#7}{#8}%
    \else
      \gploadimage{#2}{#3}{#6}{#7}{#9}%
    \fi
  \fi
}

\def\gnuplottexextension@lua{\string tex}
\def\gnuplottexextension@tikz{\string tex}

\def\gpsetvar#1#2{%
  \expandafter\xdef\csname gp@var@#1\endcsname{#2}
}

\def\gpgetvar#1{%
  \csname gp@var@#1\endcsname %
}

\def\gpsetlinewidth#1{\pgfsetlinewidth{#1\gpbaselw}}

\def\gpsetlinetype#1{\tikzset{gp path/.style={#1,#1 add}}}

\def\gpsetdashtype#1{\tikzset{gp path/.append style={#1}}}

\def\gpsetpointsize#1{\tikzset{gp point/.style={mark size=#1\gpbasems}}}

\def\gpcolor#1{\tikzset{global #1}}
\tikzset{rgb color/.code={\pgfutil@definecolor{.}{rgb}{#1}\tikzset{color=.}}}
\tikzset{global rgb color/.code={\pgfutil@definecolor{.}{rgb}{#1}\pgfutil@color{.}}}
\tikzset{global color/.code={\pgfutil@color{#1}}}

\newif\ifgpscalepoints
\tikzset{gp shift only/.style={%
  \ifgpscalepoints\else shift only\fi%
}}

\def\gp3point#1#2#3{%
  \path[solid#2] plot[only marks,gp point,mark options={gp shift only},#1] coordinates {#3};%
}

\newdimen\gpvcharsize
\newdimen\gphcharsize
\def\gpcalccharsize{%
  \pgfinterruptboundingbox%
  \pgfsys@begininvisible%
  \node at (0,0) {%
    \global\gphcharsize=1.05\fontcharwd\font`0%
    \global\gpvcharsize=1.05\fontcharht\font`0%
    \global\advance\gpvcharsize by 1.05\fontchardp\font`g%
  };%
  \pgfsys@endinvisible%
  \endpgfinterruptboundingbox%
}

\def\gpdefrectangularnode#1#2#3{%
  \expandafter\gdef\csname pgf@sh@ns@#1\endcsname{rectangle}%
  \expandafter\gdef\csname pgf@sh@np@#1\endcsname{%
    \def\southwest{#2}%
    \def\northeast{#3}%
  }%
  \pgfgettransform\pgf@temp%
  \expandafter\xdef\csname pgf@sh@nt@#1\endcsname{\pgf@temp}%
  \expandafter\xdef\csname pgf@sh@pi@#1\endcsname{\pgfpictureid}%
}

\tikzset{gnuplot/.style={%
  >=stealth',%
  line cap=round,%
  line join=round,%
}}

\tikzset{gp node left/.style={anchor=mid west,yshift=-.12ex,line width=0pt}}
\tikzset{gp node center/.style={anchor=mid,yshift=-.12ex,line width=0pt}}
\tikzset{gp node right/.style={anchor=mid east,yshift=-.12ex,line width=0pt}}

\newdimen\gpbasems
\newdimen\gpbaselw
\colorlet{gpbgfillcolor}{white}

\def\gpsetbgcolor#1{%
  \pgfutil@definecolor{gpbgfillcolor}{rgb}{#1}%
  \tikzset{tight background,background rectangle/.style={fill=gpbgfillcolor},show background rectangle}%
}

\tikzset{gp refnode/.style={coordinate,yshift=.12ex}}

\tikzset{gp arrow 1/.style={>=latex}}
\tikzset{gp arrow 2/.style={>=angle 90}}
\tikzset{gp arrow 3/.style={>=angle 60}}
\tikzset{gp arrow 4/.style={>=angle 45}}
\tikzset{gp arrow 5/.style={>=o}}
\tikzset{gp arrow 6/.style={>=*}}
\tikzset{gp arrow 7/.style={>=diamond}}
\tikzset{gp arrow 8/.style={>=open diamond}}
\tikzset{gp arrow 9/.style={>={]}}}
\tikzset{gp arrow 10/.style={>={[}}}
\tikzset{gp arrow 11/.style={>=)}}
\tikzset{gp arrow 12/.style={>=(}}

\tikzset{gp mark 0/.style={mark size=.5\pgflinewidth,mark=*}}
\tikzset{gp mark 1/.style={mark=+}}
\tikzset{gp mark 2/.style={mark=x}}
\tikzset{gp mark 3/.style={mark=star}}
\tikzset{gp mark 4/.style={mark=square}}
\tikzset{gp mark 5/.style={mark=square*}}
\tikzset{gp mark 6/.style={mark=o}}
\tikzset{gp mark 7/.style={mark=*}}
\tikzset{gp mark 8/.style={mark=triangle}}
\tikzset{gp mark 9/.style={mark=triangle*}}
\tikzset{gp mark 10/.style={mark=triangle,every mark/.append style={rotate=180}}}
\tikzset{gp mark 11/.style={mark=triangle*,every mark/.append style={rotate=180}}}
\tikzset{gp mark 12/.style={mark=diamond}}
\tikzset{gp mark 13/.style={mark=diamond*}}
\tikzset{gp mark 14/.style={mark=otimes}}
\tikzset{gp mark 15/.style={mark=oplus}}

\tikzset{gp pattern 0/.style={white}}
\tikzset{gp pattern 1/.style={pattern=north east lines}}
\tikzset{gp pattern 2/.style={pattern=north west lines}}
\tikzset{gp pattern 3/.style={pattern=crosshatch}}
\tikzset{gp pattern 4/.style={pattern=grid}}
\tikzset{gp pattern 5/.style={pattern=vertical lines}}
\tikzset{gp pattern 6/.style={pattern=horizontal lines}}
\tikzset{gp pattern 7/.style={pattern=dots}}
\tikzset{gp pattern 8/.style={pattern=crosshatch dots}}
\tikzset{gp pattern 9/.style={pattern=fivepointed stars}}
\tikzset{gp pattern 10/.style={pattern=sixpointed stars}}
\tikzset{gp pattern 11/.style={pattern=bricks}}

\tikzset{gp plot axes/.style=}
\tikzset{gp plot border/.style=}
\tikzset{gp plot 0/.style=smooth}
\tikzset{gp plot 1/.style=smooth}
\tikzset{gp plot 2/.style=smooth}
\tikzset{gp plot 3/.style=smooth}
\tikzset{gp plot 4/.style=smooth}
\tikzset{gp plot 5/.style=smooth}
\tikzset{gp plot 6/.style=smooth}
\tikzset{gp plot 7/.style=smooth}

\tikzset{gp lt axes/.style=dotted}
\tikzset{gp lt border/.style=solid}

\tikzset{gp lt axes add/.style={}}
\tikzset{gp lt border add/.style={}}
\tikzset{gp lt plot 0 add/.style={}}
\tikzset{gp lt plot 1 add/.style={}}
\tikzset{gp lt plot 2 add/.style={}}
\tikzset{gp lt plot 3 add/.style={}}
\tikzset{gp lt plot 4 add/.style={}}
\tikzset{gp lt plot 5 add/.style={}}
\tikzset{gp lt plot 6 add/.style={}}
\tikzset{gp lt plot 7 add/.style={}}
\tikzset{gp lt plot 0/.style={}}
\tikzset{gp lt plot 1/.style={}}
\tikzset{gp lt plot 2/.style={}}
\tikzset{gp lt plot 3/.style={}}
\tikzset{gp lt plot 4/.style={}}
\tikzset{gp lt plot 5/.style={}}
\tikzset{gp lt plot 6/.style={}}
\tikzset{gp lt plot 7/.style={}}

\colorlet{gp lt color axes}{black!30}
\colorlet{gp lt color border}{black}

\def\gpdashlength{\pgflinewidth}
\tikzset{gp dt 0/.style={solid}}
\tikzset{gp dt 1/.style={solid}}
\tikzset{gp dt 2/.style={dash pattern=on 7.5*\gpdashlength off 7.5*\gpdashlength}}
\tikzset{gp dt 3/.style={dash pattern=on 3.75*\gpdashlength off 5.625*\gpdashlength}}
\tikzset{gp dt 4/.style={dash pattern=on 1*\gpdashlength off 2.8125*\gpdashlength}}
\tikzset{gp dt 5/.style={dash pattern=on 11.25*\gpdashlength off 3.75*\gpdashlength on 1*\gpdashlength off 3.75*\gpdashlength}}
\tikzset{gp dt 6/.style={dash pattern=on 5.625*\gpdashlength off 5.625*\gpdashlength on 1*\gpdashlength off 5.625*\gpdashlength}}
\tikzset{gp dt 7/.style={dash pattern=on 3.75*\gpdashlength off 3.75*\gpdashlength on 3.75*\gpdashlength off 11.25*\gpdashlength}}
\tikzset{gp dt 8/.style={dash pattern=on 1*\gpdashlength off 3.75*\gpdashlength on 11.25*\gpdashlength off 3.75*\gpdashlength on 1*\gpdashlength off 3.75*\gpdashlength}}
\tikzset{gp dt solid/.style={solid}}
\tikzset{gp dt axes/.style={dotted}}

\def\gpcoloredlines{%
  \colorlet{gp lt color 0}{red}%
  \colorlet{gp lt color 1}{green}%
  \colorlet{gp lt color 2}{blue}%
  \colorlet{gp lt color 3}{magenta}%
  \colorlet{gp lt color 4}{cyan}%
  \colorlet{gp lt color 5}{yellow}%
  \colorlet{gp lt color 6}{orange}%
  \colorlet{gp lt color 7}{purple}%
}

\gpcoloredlines
\gpsetpointsize{4}
\gpsetlinetype{gp lt solid}
\gpscalepointsfalse

\usepackage{graphicx}

\usepackage{fancyhdr}
\usepackage{array}
\lstdefinelanguage{algpseudocode}{
keywordstyle=[1]{\keywordstyle},   
keywordstyle=[2]{\operatorstyle},   
keywordstyle=[3]{\typestyle},   
keywordstyle=[4]{\functionstyle},   
identifierstyle={\identifierstyle},   
keywords=[1]{begin,end,output,input,program,procedure,function,subroutine,while,do,next,repeat,until,loop,continue,endwhile,endfor,endloop,   
if,then,else,endif,return},   
literate=
{-}{$-$}1
{^}{$^\wedge$}1
{>}{{$>$\ }}1 
{<}{{$<$\ }}1 
{>=}{{$\geqslant$\ }}1
{<=}{{$\leqslant$\ }}1
{:=}{{$\gets$\ }}1 
{!=}{{$\ne$\ }}1 
{<>}{{$\ne$\ }}1
{->}{{$\;\to\;$}}1 
{&&}{{\keywordstyle and\ }}4 
{{||}}{{\keywordstyle or\ }}3
{;}{\hspace{0.2em};}2 
{,}{\hspace{0.2em},}2, 
}

\newcommand\keywordstyle{\rmfamily\bfseries\upshape} 
\newcommand\operatorstyle{\rmfamily\mdseries\upshape} 
\newcommand\typestyle{\rmfamily\mdseries\upshape} 
\newcommand\functionstyle{\rmfamily\mdseries\scshape}
\newcommand\identifierstyle{\rmfamily\mdseries\itshape}
\newcommand\addkeywords[1]{  
\lstset{morekeywords=[1]{#1}}}

\newcommand\addfunctions[1]{
  \lstset{morekeywords=[4]{#1}}}

\AtBeginDocument{
\addtolength{\abovedisplayskip}{-1ex}
\addtolength{\abovedisplayshortskip}{-1ex}
\addtolength{\belowdisplayskip}{-1ex}
\addtolength{\belowdisplayshortskip}{-1ex}
}

\AtBeginDocument{
  
}

\makeatother

\usepackage[natbib=true,bibstyle=jura2,citestyle=authoryear,backend=bibtex]{biblatex}
\providecommand{\assumptionname}{Assumption}
\providecommand{\corollaryname}{Corollary}
\providecommand{\definitionname}{Definition}
\providecommand{\factname}{Fact}
\providecommand{\lemmaname}{Lemma}
\providecommand{\propositionname}{Proposition}
\providecommand{\theoremname}{Theorem}

\begin{document}
\title{A Mechanism for Optimizing Media Recommender Systems}
\author{Brian McFadden\thanks{Adjunct Scholar, Economics Department, University of Miami Herbert
Business School, email: bdm.econ@gmail.com, X: @BMcFadden\_econ}}
\date{June 22, 2024}
\maketitle
\begin{abstract}
\noindent A mechanism is described that addresses the fundamental
trade off between media producers who want to increase reach and consumers
who provide attention based on the rate of utility received, and where
overreach negatively impacts that rate. An optimal solution can be
achieved when the media source considers the impact of overreach in
a cost function used in determining the optimal distribution of content
to maximize individual consumer utility and participation. The result
is a Nash equilibrium between producer and consumer that is also Pareto
efficient. Comparison with the literature on Recommender systems highlights
the advantages of the mechanism. A practical algorithm to generate
the optimal distribution for each consumer is provided.
\end{abstract}

\section{Overview}

In this analysis a media source or service\footnote{Conceptually the media source includes news outlets, broadcast media,
social networks, advertising, online publishers, newsletters, print
publications, etc.} is where content\footnote{The media is not limited to one type, and could be textual, video,
or audio.} flows from producer to consumer in a many to many way or a one to
many way. In a social network, for example, there are potentially
many users producing content and many consuming it, and those who
produce and consume may overlap. Alternatively, a publisher producing
news or specialized content, for example, may consist of a single
producer with potentially many consumers.%

The media source uses a specialized system to control the distribution
of media from producer to consumer, and seeks to optimize the participation
of producer and consumer on this specialized system. Understanding
how to achieve improved participation requires understanding the dynamics
of an utility maximizing consumer as well as the long and short term
benefits and costs to the producer.%
{} This paper provides a new view of those dynamics and interactions,
and lays out an algorithmic procedure for the media source to improve
sustainable participation%
. 

In this section the particulars of the producers and consumers are
further described. In the next section the dynamics of utility maximizing
consumers and the control options of the media source are discussed
and analyzed. A theorem is postulated for the media source to optimize
utility and participation of the consumer. Further in the second section
the objectives of the producer and the media source are analyzed,
and a class of cost functions are derived. A mechanism using the cost
functions is shown to assure an effective Nash equilibrium where the
media source cannot further improve value to the consumer or the producer,
and the equilibrium is further shown to be Pareto efficient%
.%

The third section discusses the implications of the model and algorithm.
The applicability in different scenarios is discussed as well as counter
examples where the mechanism might not work. Applications and comparisons
with relevant Computer Science literature on Recommender systems is
also discussed. 

The final section concludes with a summary of results and additional
thoughts. One outcome of this analysis is that the media source may
need to accept limits on what is distributed to an individual consumer
to achieve optimal long term participation. The results of this analysis
provide a normative step for the media source to identify and implement
those limits. 

\subsection{\label{subsec:Consumers}Consumers }

It is well known from cognitive science research that the ability
for humans to process images and audio is capacity limited and hence
rate limited, \citet{marois2005capacity}. It is also known that visual
and audio rate limits are relatively independent,\footnote{\citet{marois2005capacity}}
and the rate limits vary by individual.\footnote{\citet{marois2005capacity} and \citet{todd2005posterior}}
The complexity of images affects the rate that they can be processed.\footnote{\citet{marois2005capacity}}
However, nothing in the neural or cognitive science literature suggests
the ability to act on different stimuli simultaneously.\footnote{``processing multiple information streams is...a challenge for human
cognition'' \citet{ophir2009cognitive}, and ``...we can hardly
perform two tasks at once'' Marois and Ivanoff (2005)\nocite{marois2005capacity}.} Hence, true multitasking is not possible.\footnote{Although, individuals differ in their ability to task switch \citet{ophir2009cognitive}}

The fact that individuals differ in their rate limits and that rate
limits differ between visual and audio implicitly suggests how different
individuals might be capable of consuming media at different rates
or preferring different media types. This is consistent with studies
by \citet{pewresearchcenter2016themodern} suggesting different individuals
will prefer different media types for the same content. For example,
some individuals will prefer to watch the news, while others will
prefer to read the news.

Thus, following \citet{mcfadden2019atheory}, the decision to consume
additional media from a media source depends on the rate of utility
or value that the consumer expects to receive from that source. Let
this rate be denoted as $q$. Value to the consumer and utility will
be interchangeable in the rest of this paper.

The consumers have a perception of what the rate will be before they
consume additional media. This perception results from past experience
and current knowledge or signals they receive. The past experience
can be characterized by historical observations of different rates.
An implicit distribution of past rates provides the consumer with
a basis for the expected rate. The signals they may receive represent
an advanced notice from the media source or other events that allow
the consumers to alter the anticipated rate from the rate that might
otherwise be expected.

Much of the prior media consumption related economic research\footnote{See for example, \citet{anderson2005marketprovision,ambrus2016eitheror}}
has focused on the accumulation of value or the total utility. The
total and rate are only equivalent in special cases where the media
consumption interval for each media source is fixed.%
{} Here there is no reason for the the choices to have the same consumption
interval, and especially so for more spontaneous real--time decisions.\footnote{\citet{mcfadden2019atheory}}

The quantity of content is measured in units. The relation between
consumer value for a content unit, $c$, and rate, $q$, for a specific
consumer is $q=\frac{c}{\rho}$ where $\rho$ is the time required
for the consumer to consume a unit of content.\footnote{McFadden (2019)}
The value of $\rho$ varies by consumer and media source%
. Here we are focused on an optimization particular for a consumer
and media source. Thus, in this presentation $\rho$ becomes a scale
parameter and can be set to one. Setting $\rho=1$ implies that the
unit of content will correspond with a unit of time. Any analysis
with multiple media sources or consumers will likely require a specific
$\rho$ for each combination of media source and consumer to account
for the differences.

It is assumed that the values the consumer would assign to units or
groups of units can be assessed with reasonable accuracy by the media
source. For example, utilizing machine learning, statistical methods,
deep learning, or knowledge based systems\footnote{\citet{li2019asurvey}, \citet{wu2023personalized}, and \citet{chen2023whenlarge}
provide further details.} and other indicators of preference.%
{} These values can be normalized or transformed as needed. Negative
values are plausible and could indicate consumer disinterest, and
a value of zero would indicate indifference. %

\subsection{Producers}

For each content unit and consumer there is a value, $p$, that goes
to the producer of the content when the content unit is consumed by
the consumer. Negative values are plausible and indicate that the
producer would prefer that the consumer does not consume the unit.
A zero value would mean that the producer is indifferent or simply
gets no value from that consumer.%
{} Producer values can be normalized and transformed as needed. Further
details are described below.

\section{Media Source Dynamics}

Let $D$ be a discrete frequency distribution of content units over
a contiguous set of $p$ and $c$ combinations. Frequency distributions
could be derived from historical rates for combinations of $p$ and
$c$ and refined by other factors like day of week, time of year,
weather, etc. The distribution could also be limited to a certain
range of $p$ and $c$.

The distribution D can be defined by a region, $R$, bounding the
set of $p$ and $c$ combinations, a normalized density function within
the region, $\Phi(R)$, and a scalar, $N$. The scalar, $N$, also
referenced as the distribution volume, is essentially the multiplier
between the frequency distribution and the normalized density. So,
the distribution, $D_{R}$, represents a tuple $\{R,\Phi,N\}$. Distributions
may be combined with $D'=D_{1}+D_{2}$ corresponding to the additive
combination of frequencies of the distributions $D_{1}$ and $D_{2}$,
thus $R'=R_{1}\cup R_{2}$ and $N'=N_{1}+N_{2}$ and $\Phi'$ is the
normalization of the resulting frequency distribution. Much of the
math used to prove the theorems in this paper revolve around the combination
of frequency distributions. Additional details are in the appendix.

Let $D_{all}$ be the distribution over all possible combinations
of $p$ and $c$. The distribution $D_{all}$ can be viewed as a prediction
of all the content units potentially available to a consumer over
a period of time, or in some situations the candidate content available
to the consumer.

Let $Q$ be a function operating on region $R$ with distribution
$D_{R}$ to produce an anticipated $q$ for $D_{R}$,

\[
q_{\,_{R}}=Q(D_{R})
\]

There may be alternative options for $Q,$ but overall the most logical
and simplest is the expected value function. For example, let $r$
be a point in $R$ so $Q(D_{R})=E(c|D_{R})=\sum_{\forall r\in R}\phi_{r}c_{r}$
where $c_{r}$ is the value of $c$ and $\phi_{r}$ is the density
at point $r$.%

As in McFadden (2019) consumers implicitly rank available media sources
by their anticipated $q$ and consume from the top ranked media source
first. Media is consumed until the additional time needed to consume
more media would generate a higher rate of utility acquisition from
a non-media activity.\footnote{As in \citet{evans1972onthe} it is assumed that utility is derived
from activities (be they physical or mental) and not explicitly from
the properties or characteristics of goods. This contrasts with the
intermediate goods in \citealt{becker1965atheory} or \citealt{deserpa1971atheory}
that are produced and consumed. See \citealt{jara-diaz2003onthe}
for additional details. }%
{} The threshold for the individual's total media consumption per period
varies with fluctuations in available media and availability of other
activities. Logically, a larger $q$ for a media source means the
consumer will consume more because the rate of utility acquisition
from the media source would more likely exceed the rates of other
media sources and alternative activities. Thus, the potential participation
a consumer has for consuming from a media source increases with $q$.

Specifically, let $M(q)$ be the predicted potential participation
in terms of units consumed for a period. $M$ is monotonically increasing
and specific to each consumer. It could be estimated from past observations
and specific to external variables like day of week, time of year,
weather, etc. The potential participation for the distribution $D$
is thus $M\left(Q(D)\right)$ or for notational simplicity $M(D$). 

Consider a process using distribution $D$ to generate the content
made available to the consumer. The value $M(D)$ then reflects the
potential participation of the consumer who is anticipating a process
with distribution $D$.%

Let $N(D)$ be the function that returns the volume, or simply the
number of content units represented by $D$. Further let $W(D)$ be
defined as:

\begin{equation}
W(D)=\begin{cases}
N(D) & if\,\,M(D)>N(D)\\
M(D) & if\,\,M(D)<N(D)
\end{cases}\label{eq:ActualParticipation}
\end{equation}

While $M(D)$ indicates the units of potential participation, $W(D)$
will be the actual participation. If the region for $D$ is small,
but has points with large $c$, then $W$ is limited by $N$, the
number of units. The consumer would potentially consume more, but
is limited by the quantity available, $N$. On the other hand, if
the region is large, and there are lots of points with low $c$ that
push $q$ and $M$ lower, $W$ for the region is limited by $M$.
There are plenty of content units, but the consumer is only motivated
to consume $M$ units. %
{} 

Let $\Theta=\{D_{0},D_{1}\cdots D_{j}\cdots\}$ be a sequence of distributions
where $D_{0}$ is the starting distribution. If the sequence is selected
so that $N(D_{j})$ increases with $j$ and $Q(D_{j})$ decreases
with $j$ (implying that $M$ is also decreasing) going up to the
distribution $D_{all}$, we will have a visualization of equation
\ref{eq:ActualParticipation} as in Figure \ref{fig:Participation}.
Let $D^{*}$ be the distribution where $N=M$.%
\begin{figure}[h]
\begin{centering}
\includegraphics[scale=0.5]{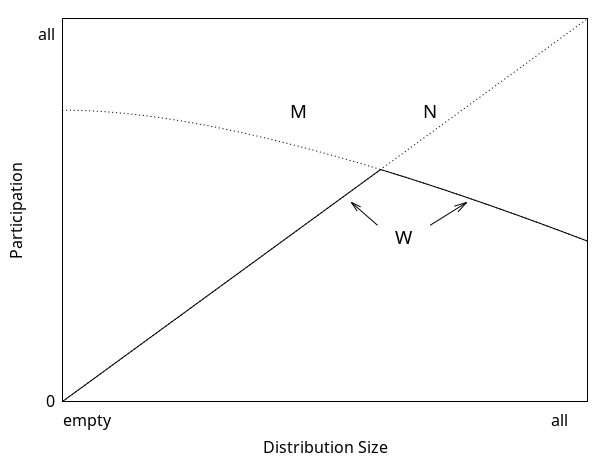}
\par\end{centering}
\caption{\label{fig:Participation}}
\end{figure}

One scenario to describe what might occur over multiple periods is
the case where $D_{all}$ and $\Theta$ are in a steady state with
no fluctuations. For this situation, $D^{*}$ for a particular consumer
would be the same for each time period in the steady state. Assuming
the consumer knows or has learned that the underlying process generating
the content is in a steady state with distribution $D^{*}$, the consumer
would expect the same number of content units $N$, and the same anticipated
$q$, per period. Thus the same participation $M$ for the anticipated
$q$. 

The alternative to the steady state is the case where $D_{all}$ and
$\Theta$ fluctuate. In this scenario the consumer does not know the
values of $N$ or the anticipated $q$. However, if the consumer believes
that $N=M$ regardless of the fluctuations, they can safely consume
the full $N(D)$ without having to speculate or approximate the anticipated
$q$.

As long as $N(D)\leq M(D)$ the consumer's realized $q$ for the period
will be at or above the threshold for consuming the content generated
by $D$. The consumer will not feel that they consumed more than they
should have, nor would they have any anticipation that they should
consume less in future periods. 

On the other hand if $N(D)>M(D)$ the consumer's realized $q$ for
the period will be below the threshold to consume $N(D)$ and they
will be more hesitant to consume the full $N(D)$ in the future. How
the consumer is able to react or adapt to receiving too much content
along with structure imposed by the media source will determine the
net benefit to the consumer and the media source. If the consumer
is able to fully adapt to extra content they will be able to limit
consumption to $M(D)$ otherwise they will react in future periods
by consuming less. Additional details are below and in the appendix. 

\subsection{Consumer Optimal Consumption}
\begin{defn}[Monotonically Decreasing]
\label{def:A-monotonic-sequence-}A sequence of distributions $\Theta=\{D_{0},D_{1}\cdots D_{j}\cdots D_{all}\}$
is monotonically decreasing when $N(D_{j})$ is increasing in $j$,
and $M(D_{j})$ is not increasing in $j$, specifically $N(D_{j-1})<N(D_{j})$
and $M(D_{j-1})\geq M(D_{j})\forall j>0$. 
\end{defn}
\begin{defn}[Incremental Distribution]
\label{def:An-incremental-distribution,}An incremental distribution,
$D'$, is a distribution derived from distribution $D$ with regions
$R'$ and $R$ and scalars $N'$ and $N$ respectively and where $r$
represents a point such that: 

(i) the incremental distribution is expansive when $R\subset R'$,
$N\leq N'$, and $\forall\,\,r\in R,\:N'\phi'_{r}=N\phi_{r}$ , and 

(ii) the incremental distribution is reductive when $R\supset R'$,
$N\geq N'$, and $\forall\:r\in R',\:N'\phi'_{r}=N\phi_{r}$ .
\end{defn}
An expansive incremental distribution expands the scalar $N$ and
the size of the region relative to the original distribution. While
a reductive incremental distribution reduces the scalar N and the
size of the region relative to the original distribution. The sequence
of distributions $\Theta$ will be expansive in the rest of this analysis
unless otherwise indicated.
\begin{defn}[Generally Decreasing]
\label{def:A-generally-decreasing-sequence}A sequence of distributions
$\Theta=\{D_{0},D_{1}\cdots D_{j}\cdots D_{all}\}$ is generally decreasing
in $j$ if $Q(D_{0})\geq Q(D_{j})\forall j$ and $M(Q(D_{0}))>N(D_{0})$
when the sequence is expansive, and $Q(D_{0})\leq Q(D_{j})\forall j$
and $M(Q(D_{0}))<N(D_{0})$ when reductive.
\end{defn}
The restriction that $M$ is generally decreasing in $j$ allows for
the possibility that $M_{j}<M_{j+1}$ provided that on average $M$
is decreasing in $j$. 
\begin{lem}
\label{lem:For-a-DistributionSequence-MathMax}For a monotonically
decreasing sequence of distributions $\Theta=\{D_{0},D_{1}\cdots D_{j}\cdots D_{all}\}$
where $M(D_{0})>N(D_{0})$, and $M(D_{all})<N(D_{all})$, the distribution
$D^{*}$ where $M(D^{*})=N(D^{*})$ will be the distribution that
maximizes $W(D)=min(M(D),N(D))$.
\end{lem}
Lemma \ref{lem:For-a-DistributionSequence-MathMax} confirms distribution
$D^{*}$ maximizes $W(D)$ for the monotonically decreasing sequence
of distributions $\Theta$.%
{} 
\begin{thm}[Optimal Distribution]
\label{thm:optimal-consumer-utility}For a generally decreasing sequence
of distributions $\Theta=\{D_{0},D_{1}\cdots$ $D_{j}\cdots D_{all}\}$,
the distribution $D^{*}$ where $M(D^{*})=N(D^{*})$ will be the distribution
from the sequence $\Theta$ that optimizes consumer utility with respect
to the consumer's participation provided that the sub sequence of
distributions starting just before $D^{*}$ is a monotonically decreasing
sequence of distributions.%
\end{thm}
\begin{proof}
Let distribution $D^{-}$, where $N(D^{-})<M(D^{-})$, and distribution
$D^{+}$, where $N(D^{+})>M(D^{+})$, be deviations from the distribution
$D^{*}$. 

In any period with deviation $D^{-}$ there is a suboptimal loss of
participation because the consumer would optimally consume more at
$D^{*}$. Thus in either a steady state or in a fluctuating scenario
$D^{-}$ deviations are suboptimal. 

In the steady state case, a shift from $D^{*}$ to $D^{+}$ results
in a corresponding shift from $M(D^{*})$ to $M(D^{+})$ and since
$M(D^{*})\geq M(D^{+})$ this results in a decrease in participation
and utility as they are consuming less at a lower $q$ because $Q(D^{*})>Q(D^{+})$. 

In the fluctuating case, consider that both $\Theta$ and $D_{all}$
fluctuate between periods, thus the deviations make the realized distribution
random. The deviations for each period from $\Theta$ for that period
could be $D^{-}$ or $D^{+}$ relative to $D^{*}$ for that $\Theta$.
If the deviation is a $D^{-}$ deviation, it is inferior for reasons
already stated. Let $q_{x}$ be the expected $q$ from the observed
distributions. If the $D^{+}$ distributions are shifted to increase
the average $N$, the $q$ for those distributions will decrease and
$q_{x}$ will decrease because $N(D_{j})$ increases and $M(Q(D_{j}))$
decreases as $j$ increases. Thus there will be less participation.
If the $D^{-}$ distributions are shifted as well to decrease the
average $N$ and to balance out the average with the $D^{+}$ distributions
and maintain the original $q_{x}$ there will be an additional loss
in participation on average. If increasing the average $N$ for the
$D^{+}$ distributions decreases $q_{x}$, it follows that decreasing
the average $N$ for the $D^{+}$ distributions increases $q_{x}$
and increases participation and utility. Similarly increasing the
$N$ for the $D^{-}$ distributions increases optimal consumption
from the additional media units. Thus a convergence of $D^{-}$ and
$D^{+}$ to $D^{*}$ would be optimal. 

Alternatively, consider a variation of the fluctuating case where
the deviations are not confined to being in $\Theta$. In this variation,
if $q_{x}$ is the basis for the anticipated $q$ then the potential
participation would be $M(q_{x})$. Assuming distributions where $M(D)=N(D)$
exist, $M(q_{x})$ would be different from the $M$ for those other
distributions where $M(D)=N(D)$. Thus, if the fluctuations generate
distributions where $M(D)\neq N(D)$, the participation would be suboptimal. 

Thus, from Lemma \ref{lem:For-a-DistributionSequence-MathMax}, for
$W(D)=min(M(D),N(D))$, the $max(W)$ occurs when $M(D^{*})=N(D^{*})$
and from the perspective of the consumer this is the optimal participation. 
\end{proof}
Theorem \ref{thm:optimal-consumer-utility} confirms that $D^{*}$
is optimal for the consumer. The proof considers deviations from $D^{*}$
in the steady state and in situations where both $D_{all}$ and $\Theta$
fluctuate. In the fluctuating case, deviations that are consistent
with $\Theta$ and those that are completely random are considered
in the proof.

While, the theorem makes it clear that deviations from $D^{*}$ are
inferior from the perspective of the consumer,%
{} the possibility that the media source might benefit from deviations
in the positive direction has not been eliminated or quantified. The
question of what is optimal for the media source will be addressed
in the next sections where the objectives of the media source are
considered in determining the sequence of distributions $\Theta$.
{} %
{} Also in the next sections, the case where the sequence is generally
decreasing over the entire sequence is analyzed, and the effect of
the $D^{-}$ and $D^{+}$ deviations on the consumer is quantified. 

Further, it should be noted that while $M(D^{*})=N(D^{*})$ is the
optimal, in practice given the discrete nature of the content units,
the practical criteria may be the closest to equal, denoted $M(D^{*})\approxeq N(D^{*})$.

\subsection{Conditions for Determining $\Theta$ and The Media Source Value Function }

The following additional definitions will aid in the analysis and
determination of the best $\Theta$ for the media source.

Let $\Lambda$ represent a mapping defined for points at least in
the region for the distribution $D$, so that $\Lambda(r,N_{r},D)\rightarrow\lambda$
where $r=\{c,p\}$ is a point, $N_{r}$ is the volume or size of
the point, and $\lambda$ is a single value. 
\begin{defn}[Preferred Incremental Distribution]
\label{def:A-Preferred-incremental-distribution}An incremental distribution
$D'$ relative to distribution $D$ is a preferred incremental distribution
for a preference mapping $\Lambda$ when: 

(i) $\forall\,\,r\in(R'-R)$, $\Lambda(r,N_{r},D_{R})\geq\Lambda(r_{1},N_{r_{1}},D_{R})\forall\,\,r_{1}\notin R'$
and $\Lambda(r,N_{r},D_{R})\leq\Lambda(r_{2},N_{r_{2}},D_{R})\forall\,\,r_{2}\in R$
for an expansive incremental distribution, and 

(ii) $\forall\,\,r\in(R-R'),$ $\Lambda(r,N_{r},D_{R})\leq\Lambda(r_{1},N_{r_{1}},D_{R})\forall\,\,r_{1}\in R'$
and $\Lambda(r,N_{r},D_{R})\geq\Lambda(r_{2},N_{r_{2}},D_{R})\forall\,\,r_{2}\notin R$
for a reductive incremental distribution.
\end{defn}
In words, when expansive, every point in $R'-R$ is preferred to every
point not in $R'$, and any point in $R$ is preferred to every point
in $R'-R$. When reductive, every point in $R'$ is preferred to any
point in $R-R'$ and every point in $R-R'$ is preferred to every
point not in $R$.
\begin{defn}[Preferred Distribution Sequence]
\label{def:A-preferred-distribution-sequence}A preferred distribution
sequence is a sequence of distributions $\Theta=\{D_{0},D_{1}\cdots D_{j}\cdots\}$
where $D_{0}$ is the starting distribution and each $D_{j+1}$ is
a preferred incremental distribution relative to $D_{j}$ with a mapping
$\Lambda$. 
\end{defn}
Letting $M$ be generally decreasing in $j$ as in Definition \ref{def:A-generally-decreasing-sequence}
allows for the possibility that $M_{j}<M_{j+1}$ provided that on
average $M$ is decreasing in $j$. This is needed for the value
functions derived in the next section.

Let $T$ be a transformation function over $p$ that maps the producer
priorities to a suitable scale. The average per unit value of distribution
$D$ to a sole producer that produced all the content generated by
$D$ would be $E(T|D)$. %
Since producer value is dependent on actual consumer participation,
the anticipated actual value to the sole producer from a consumer
with distribution $D$ would be:

\begin{equation}
S(D)=E(T|D)min(M(D),N(D))\label{eq:SingleProducerValueFunction}
\end{equation}

Alternatively, from the perspective of the media source, the value
of a distribution $D$ when $M(D)>N(D)$ should consider the potential
value of $M(D)$ that could be realized by distribution $D$, and
when $M(D)<N(D)$ the value is limited by $M(D)$. Thus, this suggests
a value function for the media source for $D$ as:

\begin{equation}
V(D)=E(T|D)M(D)\label{eq:MediaSourceValueFunction}
\end{equation}

The equation indicates that the potential value of a distribution
$D$ to the media source is the expected producer value of the distribution
multiplied by the potential participation of the consumer. It should
be noted that the value function for the sole producer and the media
source overlap when $M(D)<N(D)$.%

Using equation \ref{eq:MediaSourceValueFunction} the change in potential
value for an incremental distribution change resulting from an incremental
shift from distribution $D$ to $D'$ is:

\begin{equation}
\Delta_{V}=V'-V=E(T|D)N(D)\left(\frac{M(D')}{N(D')}-\frac{M(D)}{N(D)}\right)+T(p)\left(N(D')-N(D)\right)\frac{M(D')}{N(D')}\label{eq:MediaSourceIncrementalValueMapping}
\end{equation}

The equation shows the change in value from a possible loss of participation
in the first term relative to the possible gain in producer value
from the incremental addition to the distributions.%
{} 

The function $\Delta_{V}$ maps $(r,N_{r},D)$ to a single value and
indicates the cost of expanding to distribution $D'$ when negative,
and the benefit of doing so when positive. Thus, $\Delta_{V}$ can
be used as the media source preference mapping, $\Lambda$, and cost
function for selecting a preferred incremental distribution sequence. 

Additional Details and alternative value functions, that can be used
as a preference mapping, are provided in the distribution difference
math section of the appendix. 

\subsection{\label{subsec:Analyzing-Nash-Equilibrium}Analyzing Nash Equilibrium
and Pareto Efficient Outcomes}

In this model, while the media source is making all the decisions
regarding the determination of $D^{*}$, it is assumed that the consumer
and sole producer could exert influence on the media source. This
influence could take the form of non participation or monetary incentives.\footnote{For example, consumers moving to another media source or producers
making special payments to the media source.} So the potential gains or losses to the producer and consumer are
considered as if they were influencing the decisions of the media
source. The distribution $D^{*}$ could be considered equivalent to
a Nash equilibrium when neither the sole producer or the consumer
would prefer a shift from $D^{*}$ to a distribution before or after
$D^{*}$ in the sequence $\Theta$, and the media source has no incentive
to deviate from $D^{*}$ either. 

From equation \ref{eq:SingleProducerValueFunction}, let $\Delta_{S}=S(D')-S(D)$
be the change in value for the sole producer going from distribution
$D$ to $D'$. So $\Delta_{S}>0$ indicates an improvement and $\Delta_{S}<0$
indicates a decline in $S$ for the sole producer. The formula for
$\Delta_{S}$ has two forms. When $N\leq M$ for both $D$ and $D'$,
$\Delta_{S}=S(D')-S(D)=E(T|D')N'-E(T|D)N$ and when $N\geq M$ for
both $D$ and $D'$, $\Delta_{S}=\Delta_{V}$.
\begin{lem}
\label{lem:for-DeltaS-Below-D*}For an expansive sequence where $D'=D+Y$
and $N(D')\leq M(D')$, $\Delta_{S}>0\iff$ $E(T|Y)>0$.
\end{lem}
Lemma \ref{lem:for-DeltaS-Below-D*} indicates that for distributions
in the sequence $\Theta$ before $D^{*}$ a shift from $D$ to $D'$
is beneficial to the sole producer provided that the incremental addition
is positive. A similar lemma could also be stated for the reductive
case. 
\begin{assumption}
\label{assu:The-sole-producer-Sees}The sole producer only sees the
end result, $D^{*}$, or any possible final distribution other than
$D^{*}$, and does not see intermediate steps. 
\end{assumption}
This assumption implies that the sole producer is only able to influence
the media source's final distribution choice. Which should be reasonable
since only the end result matters to the sole producer.%

\begin{lem}
\label{lem:LimitsOnTp-M'-afterD*}For a preferred distribution sequence
let $r_{2}$ be a point with distribution $I_{r_{2}}$ such that $D'=D^{*}+I_{r_{2}}$,
then $\frac{T(p_{r_{2}})}{E(T|D^{*})}>1$ $\iff$$\frac{M(D')}{M(D^{*})}<1$,
and $\frac{M(D')}{M(D^{*})}>1$ $\iff$ $\frac{T(p_{r_{2}})}{E(T|D^{*})}<1$,
or equivalently $\frac{E(T|D')}{E(T|D^{*})}>1$ $\iff$$\frac{M(D')}{M(D^{*})}<1$,
and $\frac{M(D')}{M(D^{*})}>1$ $\iff$ $\frac{E(T|D')}{E(T|D^{*})}<1$.
\end{lem}
The implication being any point added after $D^{*}$ from a \textit{preferred
distribution sequence} cannot have both a $E(T|D')>E(T|D^{*})$ and
a $M(D')>M(D^{*})$. In other terms, Lemma \ref{lem:LimitsOnTp-M'-afterD*}
says for any point $r$ added after $D^{*}$ it is not possible that
both $T(p_{r}$) and $c_{r}$ to be above the average for $D^{*}$.
The implications for this are significant as it allows for the analysis
to be limited to focus on either an improvement in $M$ or $T(p)$
for points added after $D^{*}.$%
{} The proof and additional details are in the appendix. 
\begin{defn}[Viable Distribution Increment]
\label{def:Viable-Y-At-D}A distribution Increment $Y$, where $D'=D+Y$,
is viable in the preferred distribution sequence after $D$, if and
only if the distribution $Y$ cannot optimally be added to the sequence
prior to $D$. 
\end{defn}
For the rest of this section, let $r_{1}$ indicate the last point
added in the sequence prior to $D^{*}$ such that $D^{*}=D_{a}+I_{r_{1}}$,
and let $N_{r_{1}}$ be the size of that point. Also, unless otherwise
indicated the values of $N$ and $M$ are normalized so that $N^{*}=N(D^{*})=1$,
and $D'=D^{*}+I_{r_{2}}$. Additional details are in the appendix. 

\subsubsection{When $M(D')<M(D^{*})$ or $\frac{E(T|D')}{E(T|D^{*})}>1$}
\begin{lem}
\label{lem:For+TpNoPosDeltaS-afterD*}For a preferred distribution
sequence, let $r_{2}$ be a point with distribution $I_{r_{2}}$ such
that $D'=D^{*}+I_{r_{2}}$ where $\frac{T(p_{r_{2}})}{E(T|D^{*})}>1$,
then for at least a $N_{r_{2}}$ and $N_{r_{1}}$ where $0<N_{r_{2}}<1$
and $0<N_{r_{1}}<1$ and further limited by this condition,
\begin{equation}
\frac{\left(1+\frac{c_{r_{2}}}{Q(D_{a})}N_{r_{2}}\right)\left(1-N_{r_{2}}\right)}{1-N_{r_{2}}\left(1-\frac{c_{r_{1}}}{Q(D_{a})}\right)}>N_{r_{1}}\label{eq:restriction-for-Nr1-Nr2}
\end{equation}

there is no value for $N_{r_{2}}$ where $I_{r_{2}}$ is viable and
$\Delta_{S}>0$ or $\Delta_{V}>0$. Further, if both $T(p_{r_{2}})<E(T|D^{*})$
and a $M(D')<M(D^{*})$ then there is also no value for $N_{r_{2}}>0$
where $r_{2}$ would be viable and either $\Delta_{S}>0$ or $\Delta_{V}>0$.
\end{lem}
This Lemma provides a $S.C.$ that makes it clear that in the case
of $\frac{T(p_{r_{2}})}{E(T|D^{*})}>1$, $D^{*}$ is as good as it
gets for the sole producer and media source. For intuition, substitute
the extremes $\frac{c_{r_{2}}}{Q(D_{a})}\rightarrow0$ and $\frac{c_{r_{1}}}{Q(D_{a})}\rightarrow1$
into inequality \ref{eq:restriction-for-Nr1-Nr2} to yield $1-N_{r_{2}}>N_{r_{1}}$.%
{} Further details and proof are in the appendix.%

\begin{thm}[Equilibrium with $M(D')<M(D^{*})$]
\label{thm:kappa-lt-0}For a preferred distribution sequence including
$D'$ and a point $r_{2}=\{c_{2},p_{2}\}$ with distribution $I_{r_{2}}$
such that $D'=D^{*}+I_{r_{2}}$, and where $M(D')<M(D^{*})$, with
reasonable limits on $N_{r_{2}}$ when $E(T|D^{*})<E(T|D')$, both
consumer and sole producer prefer $D^{*}$ and the media source has
no incentive to shift to $D'$, \textup{thus }$D^{*}$\textup{ is
a Nash equilibrium.}
\end{thm}
\begin{proof}
The proof follows from Lemma \ref{lem:LimitsOnTp-M'-afterD*} which
indicates that a distribution $D'$ after $D^{*}$ in the preferred
incremental sequence cannot have both a $E(T|D')>E(T|D^{*})$ and
a $M(D')>M(D^{*})$, and Lemma \ref{lem:For+TpNoPosDeltaS-afterD*}
that confirms that there are no viable incremental additions to $D^{*}$
where $\Delta_{S}=\Delta_{V}>0$ when $E(T|D')>E(T|D^{*})$ or when
$M(D')<M(D^{*})$. Further since $\Delta_{S}<0$ there is no gain
for the producer and from Theorem \ref{thm:optimal-consumer-utility}
there is no gain for the consumer, in part because $M(D')<M(D^{*})$
implies $Q(D')<Q(D^{*})$.
\end{proof}
The theorem makes it clear that when $M(D')<M(D^{*})$ neither the
media source, sole producer, or consumer would gain from a distribution
beyond $D^{*}$ in the sequence, even if there was an increase in
the average per unit content value to the sole producer. Thus the
only possibility for an optimal beyond $D^{*}$ would be if $M(D')>M(D^{*})$
.

\subsubsection{When $M(D')>M(D^{*})$}

Given the nature of the proposed media source value function $\Delta_{V}$
the possibility of a distribution $D_{j+1}$ from the preferred distribution
sequence where $M(D_{j+1})>M(D_{j})$ cannot be excluded. However,
considering the assumption that $M$ is generally decreasing in the
preferred distribution sequence from Definition \ref{def:A-preferred-distribution-sequence},
and that $M$ is monotonically increasing in $q$ while $Q(D_{j})$
generally decreases as $N(D_{j})$ increases with $j$ in the sequence,
the scenario where $M(D_{j+1})>M(D_{j})$ should be relatively rare,
and within a sequence even less likely to occur exactly just after
$D^{*}$, but possible. Thus, the case where $M(D')>M(D^{*})$ is
further evaluated. 

\[
\text{Let, }\kappa_{r_{2}}=\frac{M(D')-M(D^{*})}{N(D')-N(D^{*})}\text{ where \ensuremath{D'=D^{*}+I_{r_{2}}} and let, \ensuremath{\kappa_{ar_{2}}=\frac{M(D_{ar_{2}})-M(D_{a})}{N(D_{ar_{2}})-N(D_{a})}\text{ where \ensuremath{D_{ar_{2}}=D_{a}+I_{r_{2}}}.}}}
\]
 In the case where $M(D')>M(D^{*})$, by definition $\kappa>0$.

\subsubsection{When $0<\kappa_{r}<1$}

When $M(D^{*})<M(D')<N(D')$ how the consumer is able to react or
adapt to receiving too much content needs to be considered. If the
consumer is fully adaptive, the consumer will consume exactly $M(D')$.
However the units not consumed will not be from $I_{r_{2}}$ since
it is assumed that the fully adaptive will know at least what points
have $c<Q(D^{*})$. On the other hand, if the consumer only realizes
$M'$ after consuming $N'$, they will need to be reactive to protect
against overreach.\footnote{The degree of adaptability likely depends more on the media type and
the structure imposed by the media source, than on the individual
consumer.}

Let $X_{l\kappa}$ be the threshold on $\kappa_{r}$ for $\Delta_{S}=\Delta_{V}>0$
when $0<\kappa_{r}<1$.%
{} The thresholds differ in the reactive and adaptive case%
,

\begin{singlespace}
\noindent 
\begin{align}
\text{adaptive: \ensuremath{} } & X_{l\kappa}=\frac{1-\frac{T(p_{2})}{E(T|D^{*})}}{1+\frac{T(p_{2})}{E(T|D^{*})}N_{r_{2}}}<\kappa_{r_{2}}\label{eq:Xlkappa-Threshold}\\
\text{reactive: \ensuremath{} } & \ensuremath{X_{l\kappa}=1-}\frac{T(p_{r_{2}})}{E(T|D^{*})}\left(1+N_{r_{2}}\right)<\kappa_{r_{2}}\label{eq:Xlkappa-Reactive}
\end{align}

\end{singlespace}

Let $X_{u\kappa}$ be the threshold for viability, 

\begin{equation}
X_{u\kappa}=\text{\ensuremath{\frac{1-\frac{T(p_{2})}{E(T|D^{*})}}{1-N_{r_{1}}+\frac{T(p_{2})}{E(T|D^{*})}N_{r_{2}}}}}\label{eq:Xukappa-threshold}
\end{equation}

The derivation of the thresholds are in the appendix. The adaptive
case is more restrictive (higher threshold) because the reactive consumer,
expecting $M=N$, consumes more, which is advantageous to the sole
producer. 
\begin{lem}
\label{lem:For-0-lt-kappa-lt-1}For a preferred distribution sequence
and point $r_{2}=\{c_{2},p_{2}\}$ with distribution $I_{r_{2}}$
such that $D'=D^{*}+I_{r_{2}}$, and where $0<\kappa_{r_{2}}<1$,
then $\Delta_{S}=\Delta_{V}>0$ when $X_{l\kappa}<\kappa_{r_{2}}<1$
\textup{, and the point will be viable when $\kappa_{ar_{2}}<X_{u\kappa}$.
}Thus for $D'$ to be viable and beneficial to the sole producer,
this inequality must hold $X_{l\kappa}<\kappa_{r_{2}}<\kappa_{ar_{2}}<X_{u\kappa}$. 
\end{lem}
Lemma \ref{lem:For-0-lt-kappa-lt-1} provides thresholds for a viable
$I_{r_{2}}$ and positive $\Delta_{S}=\Delta_{V}$.The proof and
additional details are in the appendix. 
\begin{lem}
\label{lem:Impact-Tp1-on-viability}For a preferred distribution sequence
as in Lemma \ref{lem:For-0-lt-kappa-lt-1}, let $r_{1}=\{c_{1},p_{1}\}$
be the last point added prior to $D^{*}$ with distribution increment
$I_{r_{1}}$ such that $D^{*}=D_{a}+I_{r_{1}}$, and where $\frac{T(p_{1})}{E(T|D^{*})}<1$
, then if $M(D')>M(D^{*})$ and $0<\kappa_{r_{2}}$, there will be
an alternative more restrictive viability limit as

\textup{
\begin{equation}
X_{u\kappa}=\frac{\left(1-N_{r_{1}}+N_{r_{2}}\right)\left(\frac{T(p_{1})}{E(T|D_{a})}-1\right)\frac{N_{r_{1}}}{N_{r_{2}}}+\left(1-\frac{T(p_{2})}{E(T|D^{*})}\right)}{1-N_{r_{1}}+\frac{T(p_{2})}{E(T|D^{*})}N_{r_{2}}}\label{eq:ViabilityLimit-alternative}
\end{equation}
}

and a range of $N_{r_{2}}$ values starting at $N_{r_{2}}=0$ where
$X_{u\kappa}<X_{l\kappa}$ and the range will include all $N_{r_{2}}>0$
when:

\begin{equation}
\frac{T(p_{1})}{E(T|D^{*})}<\frac{\left(1-N_{r_{1}}\right)+\left(2+N_{r_{2}}-N_{r_{1}}\right)\frac{T(p_{2})}{E(T|D^{*})}N_{r_{2}}}{\left(1-N_{r_{1}}+N_{r_{2}}\right)\left(1+\frac{T(p_{2})}{E(T|D^{*})}N_{r_{2}}\right)}<1\label{eq:Tp1-threshold}
\end{equation}
\end{lem}
Lemma \ref{lem:Impact-Tp1-on-viability} provides a more restrictive
alternative upper limit on viability by allowing the comparison point
$r_{1}$ to have a lower $p_{1}$ relative to the average $p$ for
the relevant distribution. This is shown visually in figure \ref{fig:Xu-Xl-comparison-T(p2)}
as the alt $X_{u\kappa}$. The proof and additional details are in
the appendix.

\begin{figure}[H]
\caption{\label{fig:Xu-Xl-comparison-T(p2)}}

\input{kappa-Nr2.tikz}
\end{figure}

The graph in \ref{fig:Xu-Xl-comparison-T(p2)} shows 3 bands of thresholds
(top, middle and bottom) with $X_{l\kappa}\text{, }X_{u\kappa}$,
and alternative $X_{u\kappa}$(as in Lemma \ref{lem:Impact-Tp1-on-viability}).
Also, associated with each band are possible overlapping $\kappa_{r_{2}}\text{ and }\kappa_{ar_{2}}$
lines. 
\begin{lem}
\label{lem:limits on Xu - Xl }The set of points that would fall between
$X_{l\kappa}\text{ and }X_{u\kappa}$ is limited by $\kappa_{ar_{2}}-\kappa_{r_{2}}$,
and the range of $\kappa_{ar_{2}}-\kappa_{r_{2}}$ increases as $\kappa_{ar_{2}}$
increases or similarly as $c_{2}$ increases%
. While, the value $X_{u\kappa}-X_{l\kappa}$ is increasing as $T(p_{2})$
decreases, and further as $T(p_{2})$ decreases both $X_{l\kappa}\text{ and }X_{u\kappa}$
increase. Thus, as $\kappa_{r_{2}}\rightarrow0$ , it must be the
case that $T(p_{2})\rightarrow1$ so that $\kappa_{ar_{2}}\text{ and }\kappa_{r_{2}}$
are in the range of the thresholds for $X_{l\kappa}\text{ and }X_{u\kappa}$.
\end{lem}
The proof and additional details are in the appendix. Lemma \ref{lem:limits on Xu - Xl }
shows that the necessary inequalities hold only when $\kappa_{ar_{2}}\text{ and }\kappa_{r_{2}}$
are in the same range as the thresholds from Lemma \ref{lem:For-0-lt-kappa-lt-1}.
A viable $D'$ with $\Delta_{S}=\Delta_{V}>0$ is unlikely, as values
for $\kappa$ closer to $1$ require $\frac{T(p_{2})}{E(T|D^{*})}$
to be closer to $0$ to be viable. Also, as $\frac{T(p_{2})}{E(T|D^{*})}$
approaches $1$, the chance of a viable option with $\Delta_{V}>0$
diminishes.

For further confirmation see the graph in figure \ref{fig:Xu-Xl-comparison-T(p2)}.%
{} This shows the narrowing of the distance between $X_{l\kappa}\text{ and }X_{u\kappa}$
as $\ensuremath{\frac{T(p_{2})}{E(T|D^{*})}}\ensuremath{\rightarrow1}$
and also between the $\kappa_{ar_{2}}\text{ and }\kappa_{r_{2}}$
that could be limited by the thresholds. Thus visually confirming
Lemma \ref{lem:limits on Xu - Xl }. For values of $\kappa$ closer
to $1$, $\frac{T(p_{2})}{E(T|D^{*})}$ needs to be closer to $0$
to be viable, and as $\frac{T(p_{2})}{E(T|D^{*})}$ approaches $1$,
the chance of a viable option with $\Delta_{V}>0$ diminishes.

Let $U(D)$ be the utility the consumer receives from consuming distribution
D, and let $\Delta_{U}=U(D')-U(D)$ the change in utility between
$D$ and $D'$. When $0<\kappa<1$ and the consumer is fully adaptive
the change in utility between $D^{*}$ and $D'$ is $\Delta_{U}=Q(D')M(D')-Q(D^{*})N(D^{*})$.
Since $Q(D')>Q(D^{*})$ and $M(D')>N(D^{*})$, it must be that $\Delta U>0$.
Conversely, $\kappa<0$ $\implies$ $\Delta U<0$. 

For the fully reactive consumer a 2 period analysis is used to approximate
the lifetime impact of when the consumer assumes $M=N$ and it is
not.%
{} Also, need to account for the opportunity cost for misappropriated
time when the consumer is over consuming at $N'>M'$. This would be
$-\left(N'-M'\right)\iota Q(D^{*})$, where $\iota>0$ would be the
gain in utility relative to $Q(D^{*})$ if the consumer is able to
fully adapt and optimally allocate time to limit consumption from
the media source at $M'$. Thus, the comparison for the reactive consumer
not to lose is,

\[
\left[1-\left(N'-M'\right)\iota\right]Q(D^{*})+N_{r_{2}}c_{r_{2}}+\frac{M'}{N'}Q(D^{*})>2Q(D^{*})
\]

The term $\frac{M'}{N'}Q(D^{*})$ is not an actual 2nd period outcome,
but an estimate of the lifetime impact from the overreach. Substituting
$N'=1+N_{r_{2}}$, $M(D^{*})=1$, $M(D')=1+\kappa N_{r_{2}}$, and
$Q(D')=\frac{Q(D^{*})+N_{r_{2}}c_{r_{2}}}{1+N_{r_{2}}}$ yields the
additional threshold,

\begin{equation}
X_{c\kappa}=\frac{1-\frac{c_{r_{2}}}{Q(D^{*})}\left(1+N_{r_{2}}\right)+\iota\left(1+N_{r_{2}}\right)}{1+\iota\left(1+N_{r_{2}}\right)}<\kappa_{r}\label{eq:Consumer-Kappa-Threshold}
\end{equation}

Combining \ref{eq:Consumer-Kappa-Threshold} and \ref{eq:Xlkappa-Reactive}
shows that $X_{c\kappa}>X_{l\kappa}$ when, %

\begin{equation}
\frac{T(p_{r_{2}})}{E(T|D^{*})}\left[1+\iota\left(1+N_{r_{2}}\right)\right]>\frac{c_{r_{2}}}{Q(D^{*})}\label{eq:Cond-for-ckappa-gt-lkappa}
\end{equation}

\begin{thm}[Equilibrium with $0<\kappa<1$]
\label{thm:0-lt-kappa-lt-1} When $M(D')>M(D^{*})$ and $0<\kappa<1$
the media source utilizing a preferred distribution sequence, the
media source will encounter 4 possible scenarios:

(i) where $\Delta_{S}=\Delta_{V}>0$ and $X_{c\kappa}<\kappa$ the
media source, sole producer, and consumer prefer $D'$,

(ii) where $X_{c\kappa}<\kappa<X_{l\kappa}$ the consumer prefers
$D'$ while the sole producer $D^{*}$,

(iii) where $X_{l\kappa}<\kappa<X_{c\kappa}$ the sole producer prefers
$D'$ while the consumer $D^{*}$,

(iv) where $\Delta_{S}=\Delta_{V}<0$ and $\kappa<X_{c\kappa}$ all
prefer to stay at $D^{*}$ and not move to $D'$.

{\setlength{\parindent}{0pt}{Further,}}

(a) Only scenarios (i) and (iv) are Nash equilibrium,

(b) The media source choice in scenarios (ii) and (iii) requires additional
assumptions,

(c) The more likely scenarios are (ii) and (iv),

(d) Scenarios (iii) and (i) are less likely to occur. %
{} 
\end{thm}
\begin{proof}
The distribution $D'$ is viable by definition \ref{def:A-preferred-distribution-sequence}
for a preferred incremental sequence. Scenarios (i) and (iv) result
from Lemma \ref{lem:For-0-lt-kappa-lt-1} and derivation of equation
\ref{eq:Consumer-Kappa-Threshold}. Scenarios (ii) and (iii) follow
from equations \ref{eq:Consumer-Kappa-Threshold} and \ref{eq:Xlkappa-Threshold}
and Lemma \ref{lem:For-0-lt-kappa-lt-1}. 

From equation \ref{eq:Cond-for-ckappa-gt-lkappa} when $\frac{T(p_{r_{2}})}{E(T|D^{*})}\left(1+N_{r_{2}}\right)<1$
it is clear that as $c_{r}$ increases $\iota$ must increase by a
larger multiple and from Lemma \ref{lem:limits on Xu - Xl } $\frac{T(p_{r_{2}})}{E(T|D^{*})}$
must decrease as $c_{r}\rightarrow1$ to maintain $\Delta_{S}=\Delta_{V}>0$,
and further given the reasonable upper limit on $\iota$, $X_{c\kappa}>X_{l\kappa}$
is much less likely to occur. Thus less likely for scenario (iii)
to occur. Scenario (i) is less likely to occur simply because $\kappa>X_{l\kappa}$
is less likely to occur per Lemmas \ref{lem:For-0-lt-kappa-lt-1}
- \ref{lem:limits on Xu - Xl }. Thus with scenarios (i) and (iii)
less likely to occur, (ii) and (iv) are the most likely to occur when
$M(D')>M(D^{*})$. 

The Nash equilibrium for (i) and (iv) is obvious given that neither
the sole producer or consumer have any incentive to alter choice of
the media source between $D^{*}$ and $D'$, and likewise lacking
in scenarios (ii) and (iii).
\end{proof}
Since scenarios (ii) and (iii) don't provide a definite outcome, it
would seem that the media source would need an alternative mechanism
to handle the scenarios, and especially since together they are the
most likely outcome when $M(D')>M(D^{*})$.

One possibility would be for the media source to adjust the final
distribution to $D^{*+}=D^{*}-Y+I_{r_{2}}$, where $Y$ is the distribution
removed from $D^{*}$, so that $M(D^{*+})=N(D^{*+})$ and make the
adjustments so that the sole producer was no worse off than with $D^{*}$
or $D'$. This would provide an optimal outcome for the consumer and
the media source and could be at least Pareto efficient for the sole
producer. It would also adhere to the objective of having $M(D)=N(D)$,
and move to a higher volume distribution. The following Lemma provides
a $S.C.$ for $Y.$
\begin{lem}
\label{lem:CriteriaForD*+}When $0<\kappa_{r_{2}}<1$ There is a distribution
$Y$ where $Y\subset D^{*}$ , $D'^{+}=D^{*}-Y+I_{r_{2}}$, $N(Y)=N_{Y}\leq N(D')-M(D')<N_{r_{2}}$,
$\frac{dN_{Y}}{dQ(Y)}>0$, and $M(D'^{+})=N(D'^{+})$ if this inequality
holds,

\begin{equation}
\frac{N_{Y}}{N_{r_{2}}}\leq\frac{c_{r_{2}}}{Q(Y)}\label{eq:ExistenceYforM-eq-N}
\end{equation}

And it will also be Pareto efficient between the consumer and sole
producer when this inequality also holds, 
\end{lem}
\begin{equation}
\frac{N_{Y}}{N_{r_{2}}}\leq\frac{T(p_{r_{2}})}{E(T|Y)}\label{eq:CondForParetoY}
\end{equation}

Proof is in the appendix. The algorithm for generating the carveout
$Y$ %
is listed below and discussed further in section \ref{sec:Discussion}. 

\subsubsection{When $1\protect\leq\kappa_{r}$}
\begin{lem}
\label{lem:For-kappa-geq-1}For a preferred distribution sequence,
let $r_{2}=\{c_{2},p_{2}\}$ be a point with distribution $I_{r_{2}}$
such that $D'=D^{*}+I_{r_{2}}$ and where $1<\kappa_{r_{2}}$, then:

(i)\textup{ $\Delta_{S}=N_{r_{2}}T(p_{2})$, } 

(ii) the condition for $\Delta_{V}>0$ is, $X_{l\kappa}=\frac{1-\frac{T(p_{2})}{E(T|D^{*})}}{1+\frac{T(p_{2})}{E(T|D^{*})}N_{r_{2}}}<\kappa_{r_{2}}\text{, and}$

(iii)$X_{l\kappa}>1\iff\frac{T(p_{2})}{E(T|D^{*})}<0$ and $X_{l\kappa}=1\iff\frac{T(p_{2})}{E(T|D^{*})}=0$,
thus 

(iv) $X_{l\kappa}>1\implies\text{\ensuremath{\Delta_{S}<0}}$ and
$X_{l\kappa}=1\text{ and \ensuremath{\kappa_{r_{2}}=1}}\implies\text{\ensuremath{\Delta_{S}=0}}$. 
\end{lem}
This lemma indicates that when $1\leq X_{l\kappa}<\kappa_{r_{2}}$
and point $r_{2}$ is viable, it will be optimal for the media source
and the consumer to go beyond $D^{*}$, but not for the producer,
unless $T(p_{2})=0$, see Lemma \ref{lem:for-DeltaS-Below-D*}. Further
details and the proof are in the appendix.%

\begin{defn}[Additional $D^{*}$]
 Let $D^{n*}$ be an additional distribution after the initial $D^{*}$
where $M(D^{n*})\approxeq N(D^{n*})$, and where $n$ indicates the
$nth$ occurrence.
\end{defn}
\begin{lem}
\label{lem:DeltaV-between-2-local-optimal}For a preferred distribution
sequence a $D^{2*}$ can only exist if $\kappa\geq1$ , and if $D^{2*}$
exists in the sequence, then $\Delta_{V}=0$ and $\Delta_{S}=0$ between
$D^{*}$ and $D^{2*}.$ 
\end{lem}
See appendix for the proof, and note the Lemma generalizes to $D^{n*}$
and $D^{n+1*}$. When there is a $\kappa\geq1$ it effectively returns
to the state where $M(D)>N(D)$ and repeats the steps from before
the first $D^{*}$. 

\begin{algorithm}[H]
\caption{\label{alg:DeterminingD*}Determining $D^{*}$}

\lstset{
  language={algpseudocode},
  columns=fullflexible, mathescape=true,
  sensitive=false,basicstyle=\small,
  tabsize=4,
  numbers=left,
  numberstyle=\scriptsize, }
\addfunctions{BestIncrement,BestNextInSequence}
\addkeywords{exit,when,loop}

\begin{lstlisting}[firstnumber=1,xleftmargin=12pt,framextopmargin=-1pt,numbersep=5pt]
Input:
	$D_{all}$ := distribution with all content possible or predicted to be possible for the consumer
	RatioThreshold := threshold for $\frac{M}{N}$ to indicate close enough to 1 // eg. 1.05 or 0.95
Output:	$D^{*}$ and $N^{*}$
Procedure:
	D := $D_{0}$ + BestIncrement($D_{0}$,$D_{all}$) // Initial distribution
repeat
	if ( $\frac{M(D)}{N(D)}$ $\leq$ RatioThreshold )
		$D'$ :=  BestNextInSequence(D,$D_{all}$)
		$\kappa$ := $\frac{(M(D')-M(D))}{(N(D')-N(D))}$
		when ( $\kappa \leq 0$ ) exit loop
		when ( $\kappa \geq 1$ ) D := $D'$
		when ( $0 < \kappa < 1$ ) D :=  $D' - $ getCarveout(D,$D'$)  // See Algorithm 2
		continue loop
	endif
	D := D + BestIncrement(D,$D_{all}$)
until exit
$D^{*}$ := D
$N^{*}$ := N$(D^{*})$
function BestIncrement(D,$D_{all}$) 
	points := getPointsToEvalute(D,$D_{all}$) // gets subset of points in $D_{all} - $ D
	use mapping $\Lambda(r,D)$ to rank points and get TopPoint // could get multiple points 
	return TopPoint 
end function
function BestNextInSequence(D,$D_{all}$) // 
	$D'$ := D
	repeat
		$D'$ := $D'$ + BestIncrement($D'$,$D_{all}$)
		$\kappa$ := $\frac{(M(D')-M(D))}{(N(D')-N(D))}$
		when ( $\kappa$ $\geq$ 1 || $\kappa$ $\leq$ 0) return $D'$
		when ($\kappa$ is increasing ) continue loop
		return $D'$
	until exit
end function

\end{lstlisting}
\end{algorithm}

\begin{thm}[Equilibrium with $1<\kappa_{r_{2}}$]
\label{thm:kappa-gt-1}For a preferred distribution sequence and
a point $r_{2}=\{c_{2},p_{2}\}$ with distribution $I_{r_{2}}$ such
that $D'=D^{*}+I_{r_{2}}$ and where $1<\kappa_{r_{2}}$, if $D'$
is viable and $\Delta_{V}>0$, the media source and the consumer would
prefer $D'$ over $D^{*}$and the media source would want to continue
through the sequence of distributions after $D'$ until $D^{2*}$
\textup{was reached, and $D^{2*}$ would be a Nash Equilibrium unless
there was a $D^{3*}$. }
\end{thm}
\begin{proof}
The proof follows from Lemmas \ref{lem:For-kappa-geq-1} and \ref{lem:DeltaV-between-2-local-optimal}.
Further, the sole producer would not directly benefit from the shift
from $D^{*}$ to $D^{2*}$, but would not lose either, and would have
no incentive to oppose the shift.
\end{proof}
The Theorem provides for a possible shift from one $D^{*}$ to another.
However, the availability of this option would be an infrequent outcome.
The more likely case would be when there is no viable shift past $D^{*}$
where $\Delta_{V}>0$ and $1<\kappa$. That is not to say that a $D^{n*}$
is not possible. Over a period of time with many revisions of the
sequence $\Theta$, along with determining a new $D^{*}$ for each
revision, or across many consumers each with their own sequence $\Theta$,
a few occurrences of a second $D^{*}$ could be likely under certain
conditions. In practice it may be beneficial to check several distributions
above $D^{*}$ in the sequence for a distribution that has $\kappa>1$
relative to $D^{*}$. If found that distribution would be the $D'$
in Theorem \ref{thm:kappa-gt-1}, see Algorithm \ref{alg:DeterminingD*}.
\begin{algorithm}[h]
\caption{\label{alg:GeneratingCarveout}Generating the Carveout $Y$}

\lstset{
  language={algpseudocode},
  columns=fullflexible, mathescape=true,
  sensitive=false,basicstyle=\small,
  tabsize=4,
  numbers=left,
  numberstyle=\scriptsize, }
\addfunctions{BestIncrement,BestNextInSequence}
\addkeywords{exit}

\begin{lstlisting}[firstnumber=1,xleftmargin=12pt,framextopmargin=-1pt,numbersep=5pt]
Input:	$D^{*}$
Output:	$Y$ // Carveout Distribution
Procedure:
	$\bar{N_{Y}}$ := 0 // interim value for $\sum_{{x}\in{Y}}{N_x} $
	$S_T$ := 0 // interim value for $\sum_{{x}\in{Y}}{N_x}{T(p_x)} $
	$S_c$ := 0 // interim value for $\sum_{{x}\in{Y}}{N_x}{c_x} $
	$Y$ := empty distribution
	repeat
		Points := getPointsToEvalute($D^*-Y$)  
		x := the point from points with the lowest $c<c_{r_{2}}$ that satisfies the condition $S_c + N_{x}c_{x}\leq c_{r_{2}}N_{r_{2}}$ and
		also satisfies the condition $S_T + N_{x}T(p_{x})\leq N_{r_{2}}T(p_{r_{2}})$ as in Lemma 8
		if ( no point satisfies the condition ) exit
		$\bar{N_{Y}}$ := $\bar{N_{Y}}+N_{x}$
		$S_T$ := $S_T + N_{x}T(p_{x})$
		$S_c$ := $S_c + N_{x}c_{x}$
		$Y$ := $Y+I_{x}$
	until $N'-M\left(\frac{Q(D^{*})-\bar{N_{Y}}Q(Y)+N_{r_{2}}c_{r_{2}}}{1-     \bar{N_{Y}}+N_{r_{2}}}\right)\approxeq\bar{N_{Y}}$
\end{lstlisting}
\end{algorithm}

\section{\label{sec:Discussion}Discussion}

\subsection{Algorithm for Determining $D^{*}$}

Algorithm \ref{alg:DeterminingD*} provides a simple suggestive option
for determining $D^{*}$ for each consumer, where the determination
of $D^{*}$ coincides with the determination of at least part of the
sequence $\Theta$. Variations of this algorithm are possible,
including altering the stopping conditions or the handling for the
$0\leq\kappa\leq1$ case.

An alternative approach, that may be advantageous in some cases, computes
the entire sequence $\Theta$ first, then selects the distribution
with the highest volume, $N$, from those distributions that have
a $\text{\ensuremath{N\approxeq M}}.$ Another approach that may have
some advantages, starts at $D_{all}$ or another significantly large
distribution, and generates the preferred incremental sequence using
reductive steps. The possible disadvantage of this alternative is
that there might be more uncertainty with starting closer to $D_{all}$. 

In practice, a procedure of this type would need to integrate with
other media source systems. In particular, a system for assigning
$c$ and $p$ values to content units for a consumer, a system for
generating the $D_{all}$ distribution, and a system for determining
the mapping between $q$ and $M$.%
{} 

Further note, if the $D_{all}$ distribution represents the actual
candidate content then $D^{*}$ would be the content made available
for the consumer. Alternatively, if the $D_{all}$ distribution is
predicted content, the region associated with $D^{*}$ becomes the
determinator of what content would be made available to the consumer. 

Algorithm \ref{alg:GeneratingCarveout} provides the steps for generating
the carveout $Y\subset D^{*}$ by iteratively selecting points that
satisfy the conditions of Lemma \ref{lem:CriteriaForD*+}. This is
a suggestive example, and variations are possible.%
{} The algorithm leverages the fact from Lemma \ref{lem:CriteriaForD*+}
that since $N_{Y}<N_{r_{2}}$ it is possible that $E(T|Y)>T(p_{r_{2}})$
and inequality \ref{eq:CondForParetoY} will still hold. Also since
$N_{Y}$ decreases as $Q(Y)$ decreases, there is even more room to
increase $E(T|Y)$. Thus choosing points with $c<c_{r_{2}}$ have
a direct impact on increasing $M(D'-Y)$ and an indirect impact on
equation \ref{eq:CondForParetoY} via the decrease in $N_{Y}$ . 

\subsection{Summary of Results}

As shown in Theorem \ref{thm:kappa-lt-0}, where $D'=D^{*}+I_{r_{2}}$
and $M(D')<M(D^{*})$, under reasonable limits on $N_{r_{2}}$, $D^{*}$
is best for the consumer, producer, and media source. This includes
the case where $E(T|D^{*})<E(T|D')$, and thus showing that even with
an increase in the producer average unit value, deviations from $D^{*}$
are not optimal for the sole producer. Thus, the media source could
not directly, or via influence from the sole producer, benefit from
deviations to $D^{*}$ when $M(D')<M(D^{*})$ either. This clearly
answers the questions raised after the proof of theorem \ref{thm:optimal-consumer-utility}.%

Only when $M(D')>M(D^{*})$ is there a possibility of going beyond
$D^{*}$, where $N(D^{*})=M(D^{*})$. There are 2 cases where $D'$
could be preferred by the media source. The first is when $0<\kappa<1$,
where $\kappa$ represents the slope between $M(D^{*})$ and $M(D')$,
and second when $1<\kappa$.

In the first case, the analysis of $D^{*}$ and $D'$ is done in
the plane \{$\kappa,N_{r_{2}}\}$ where $\kappa\in[-\infty,\infty]$
and $N_{r_{2}}\in[0,1)$ normalized by $N(D^{*})=1$, and where $N_{r_{2}}=N(D')-N(D^{*})$.%
{} A set of thresholds basically dependent on $p_{r_{2}}$ and $N_{r_{2}}$,
and a set of curves, mostly dependent on $c_{r_{2}}$ and $N_{r_{2}}$,
express the relationships and limits between $\kappa$ and $N_{r_{2}}$
for $M(D')$. This is illustrated in Figure \ref{fig:Xu-Xl-comparison-T(p2)}
and described in Theorem \ref{thm:0-lt-kappa-lt-1} where the four
possible outcomes between consumer and sole producer are outlined.
Two of the outcomes are Nash equilibrium where the outcome is best
for both consumer or sole producer. The other two outcomes lead to
conflict between consumer and sole producer. To resolve the conflicts,
conditions for altering $D'$ by removing a subset of $D^{*}$ in
a Pareto efficient way are outlined in Lemma \ref{lem:CriteriaForD*+}.
Thus turning the potential conflict into an outcome where the participation
is increased and neither consumer or sole producer is made worse off.
An implementation of Lemma \ref{lem:CriteriaForD*+} is provided in
algorithm \ref{alg:GeneratingCarveout}. Also, in this case as shown
in Lemma 7, the possibility of an equilibrium beyond $D^{*}$ diminishes
as $\kappa\rightarrow0$. 

In the second case where $1<\kappa$, theorem \ref{thm:kappa-gt-1}
indicates a consensus between the consumer and media source to move
to $D'$. While the initial move to $D'$ will be a loss for the sole
producer, the eventual end at $D^{2*}$ will be a net $0$ change
for the sole producer. Although, the sole producer might have an ancillary
benefit from the fact that there are more units consumed at $D^{2*}$
than $D^{*}$.

Thus it can be concluded that the approach outlined here leads to
a distribution $D$ that is optimal for producer, consumer and media
source. It is also counter to a popular view that it is necessary
for a media source or producer to trick media consumers to consume
more in order to maximize profit. It may be the case with other algorithms,
but as shown here, not with this algorithm and cost function. 

For the 6 outcomes from Theorems \ref{thm:kappa-lt-0}, \ref{thm:0-lt-kappa-lt-1},
and \ref{thm:kappa-gt-1} there are no possible improvements from
the Nash equilibrium result for any of the stakeholders. Thus the
outcomes are Pareto efficient.

\subsection{Counter Examples}

Further insights into the model can be seen by looking at where the
model or Nash equilibrium might break down.

If there are multiple producers interacting directly with the media
source, or with agents of the media source, one producer could have
an incentive to influence the media source to go beyond $D^{*}$ even
when $M(D')<M(D)$. This would bypass the interests of a sole producer
and would be at a loss for the other producers as well as the consumer.
As an example, consider a case where a producer would pay or pay extra
to increase the producer value or to have the media source include
the producer's content without regard to the preferred sequence. If
consumers could respond by switching to another media source, the
media source would have incentive to correct the situation. However,
if the same scenario is in place at the alternative media sources,
or there are no alternative media sources because the media source
is the sole source for the media it provides, the consumer has few,
if any, options to counter influence the media source.%

A media source that is not optimizing by individual consumers, or
at least homogeneous consumers, is unlikely to be at an optimal. This
could occur when the cost of operating or implementing a system to
monitor individual consumers is prohibitive. Applying alternative
\textit{ad hoc} procedures designed to improve aggregate media source
objectives by their nature leave few if any at the optimal. As in
the proof for Theorem \ref{thm:optimal-consumer-utility} letting
$q_{x}$ be the basis for the anticipated $q$ for an individual consumer
under an \textit{ad hoc} system, then the potential participation
would be $M(q_{x})$ which would almost always be different from $M(D^{*})$,
and thus suboptimal relative to the hypothetical $D^{*}$ in this
situation. 

Another case to consider is when the relationship between $M$ and
$q$ spikes early. Normally a gradual increase at a decreasing rate
going to an asymptote reflecting the physical time limit for an individual
would be expected. However, in some cases where there is a lack of
alternative activities for a consumer,\footnote{See \citet{mcfadden2019atheory}}
$M$ reaches the asymptote much more quickly. Meaning at a relative
low value of $q$ the consumer is willing to spend nearly all their
available time consuming the content from the media source. Another
case is when, for possibly psychological reasons, the consumer has
elevated the q for the content of the media source to an exceptionally
high level. In either of these cases $D^{*}$ is beyond the asymptote
and it is always the case that $M>N$. When this situation is combined
with a media source that has a practically unlimited flow of content
at a value of $q$ sufficient for this special consumer to allocate
all their time available for media consumption to the media source.
Thus the physical constraint on time limits the consumption. The implication
of this is that $M(D)$ is relatively flat and changes in $q$ are
ineffective. Thus the value function would simply prioritize on producer
unit value. 

The different case would occur when there is insufficient content
to reach $M=N$. Here the consumer is under served and both the consumer
and sole producer would like to increase $N$. So it would be expected
that this case would be resolved quickly by an enterprising media
source.

\subsection{Comparison to the Recommender Systems (RS) Literature}

There is a significant body of research in the computer science literature
related to predicting what content a consumer would be interested
in. This is essentially related to the determination of the $c$ values
used in this paper. 

For the analysis here, it was possible to assume the rate of utility
$q$ and $c$ could be combined as in \ref{subsec:Consumers} because
$\rho$ could be normalized to $1$ for a single media source and
consumer combination. However, it would not be possible to directly
cross use interest metrics between media sources or between different
media or even between different formats or contexts, and certainly
not between users in a community. Any of those cross use combinations
would require measuring or estimating the different $\rho$ needed.

There are a variety of interpretations in the literature of what the
determination of interest should include and what if any other objectives
should be considered.\footnote{\citet{silveira2019howgood}, \citet{meng2023asurvey}, \citet{wu2023personalized},
\citet{chen2023whenlarge} \citet{li2019asurvey}} The term \textit{utility} in this literature is generally the level
of interest based on user profile and historical activity, \citet{silveira2019howgood}.
Some of the additional objectives include concepts like \textit{accuracy},
\textit{diversity}, and \textit{novelty}.\footnote{\citet{silveira2019howgood}, \citet{wang2020multiobjective}, \citet{meng2023asurvey},
\citet{wu2023personalized}, \citet{li2019asurvey}} Where \textit{accuracy} generally measures how well the recommendations
match the consumer's interests, \citet{chen2023whenlarge}, \citet{wu2023personalized}.
The concepts \textit{diversity} and \textit{novelty} generally reflect
a desire to include items that are outside of what the consumer may
have been previously exposed to, \citet{guy2010socialmedia}.

While there are different approaches for combining the objectives,
\citet{ribeiro2014multiobjective}, \citet{wang2020multiobjective},
\citet{ge2022towardpareto} and others have suggested determining
a Pareto optimal set of recommendation lists (the Pareto frontier)
where the measure for one objective cannot be improved without worsening
the measure for another objective and then selecting a final recommendation
list from the set, \citet{wang2020multiobjective}, \citet{ribeiro2014multiobjective},
\citet{ge2022towardpareto}. 

Pareto efficiency is of course different from Nash equilibrium, \citet{myerson1997gametheory}.
In the model here, the stopping point is reached because actual stakeholders
cannot improve their position, and not because another stakeholder
would be worse off. However, the resulting Nash Equilibria shown
in section \ref{subsec:Analyzing-Nash-Equilibrium} above are also
Pareto efficient in the sense that there is no other outcome that
makes at least one stakeholder better off.

Recent work from \citet{jagadeesan2024supplyside} and \citet{acharya2024producers}
and earlier work from \citealt{ben-porat2018agametheoretic,ben-porat2020content}
present models related to recommended content with Nash equilibrium
outcomes between independent producers.

\citet{cai2016mechanism} and \citet{cai2022recommender} apply mechanism
design to product RS, and \citet{boutilier2024recommender} discuss
using mechanism design with a social choice function in RS. The mechanism
design discussed in the research there is based on mechanism design
in economics that facilitates the reveal of private information in
a game with incomplete information \citet{myerson1997gametheory},
and is different from the mechanism described here that assures Nash
equilibrium between stakeholders. 

The concerns raised in \citet{boutilier2023modeling} regarding the
use of mechanism design in RS are generally resolved here. Specifically,
the use of participation, as the common value between producer and
consumer, eliminates the need for a shared monetary metric and mitigates
the other social choice function concerns. Also, using the anticipated
rate of utility mitigates the concerns with adding utility and the
potential need for further sequential optimization as discussed there.%
{} 

\citet{ge2022towardpareto} remarks, ``approaches on recommendation
with multiple objectives to achieve Pareto efficiency can be categorized
into two groups: evolutionary algorithm {[}\citet{zitzler2001spea2improving}{]}
and scalarization {[}\citet{lin2019aparetoefficient}{]}''. The approach
here differs from both those alternatives, yet using an endogenous
approach reaches a Nash equilibrium and a Pareto efficient solution
without needing to explicitly compute a Pareto frontier.%
{} %

A further distinguishing feature of the model here, relative to the
computer science literature, is the separation of the producer and
consumer objectives combined with a cost function that explicitly
considers the effect that consumer participation has on overall value
to the producers. This allows not only a means to combine objectives
in a Nash equilibrium and Pareto efficient way, but also a means to
determine a distribution of content that optimizes participation.%
{} 

In the review of RS in \citealt{silveira2019howgood}, it is noted
that \textit{accuracy} and \textit{utility} are user (consumer) dependent
while \textit{diversity}, and \textit{novelty} are not. Also the objectives
for \textit{diversity}, and \textit{novelty} would appear to correspond
with longer term goals for the sole producer or the media source \citet{wu2023personalized}.
Thus, it is suggested that \textit{diversity}, and \textit{novelty}
could be reflected by the producer value, $p$, in the model presented
here. Although, \citet{wu2023personalized} also suggests that using
the right ``content features'' when learning the consumer interests
and preferences could produce news recommendations that have greater
diversity.

A further notable result is that this model implicitly provides a
threshold for how many units should be made available to the consumer.
Specifically, the threshold for $D^{*}$ and the corresponding $N^{*}$
are endogenously determined. While in the literature there is only
a fixed number or an exogenous cutoff based on some external factor. 

\subsection{Applications}

\begin{singlespace}
\noindent The description of the media source has so far been short
of application details. This was intentional to emphasize the algorithm
and the results. To tie in the general nature of this presentation
to real world situations some examples are provided in Table \ref{Table:MediaSource}.
Depending on type, the media source could be the sole producer or
there could be many independent producers. Also included in the table
are considerations for how the different media types might interpret
producer value, $p$, for consumption by a specific consumer.
\end{singlespace}

The examples are suggestive, and there may be cases where one of the
examples is a different type. For example, in some cases, a media
source for advertising may be a sole producer, or a media source for
news may choose to let editors or contributors act as independent
producers.%

When the media source is the sole producer, the producer value, $p$,
should reflect the media source's objectives for increasing long term
participation by increasing exposure to content outside what would
otherwise be made available to a consumer based on their profile and
activity. The $p$ value in this case is in effect a means for promoting
certain content to specific consumers. The media source would need
to consider (1) the expected future increase in $Q(D^{*})$ from promoting
the content item and (2) the expected loss in the current period from
promoting an item that the consumer is not interested in. While the
goals for the media source to promote certain content are in general
similar to the objectives discussed in the Recommender Systems literature
for content \textit{diversity} and \textit{novelty}, the approach
here is different because the decision to promote an item is explicitly
considering the changes in $D^{*}$ for the consumer.

\begin{table*}
\caption{Media Source Examples}
\label{Table:MediaSource}
\begin{singlespace}
\noindent \centering{}%
\begin{tabular}{>{\centering}p{0.1\linewidth}>{\centering}p{0.36\linewidth}>{\centering}p{0.34\linewidth}}
\toprule 
\textbf{Type} & \textbf{Media source is sole producer} & \textbf{Producers are independent}\tabularnewline
\midrule
\midrule 
\textbf{Examples} & news outlets, news aggregators, trade publications, niche publishers,
content recommendations, status messaging, autonomous agents & social media, user generated content, group discussions, advertising,
promotions and offers, personal aggregators\tabularnewline
\midrule 
\textbf{Producer \LyXZeroWidthSpace Value} & Generated based on media source objectives & Determined by the producer's audience ranking\tabularnewline
\bottomrule
\end{tabular}
\end{singlespace}
\end{table*}
When the producers are independent from the media source, the producer
value $p$ should be determined by using limits that incentivize the
content producers to identify and rank desired consumers. Then the
media source would assign producer value, $p$, specific to the item
and consumer based on the rank provided by the producer. Having
the producers identify and rank consumers by their profile characteristics
provides a mechanism for producers to reveal producer value, $p$,
and optimize their reach.

\section{Conclusion}

A model has been described that addresses the fundamental trade off
between media producers who want to increase reach and media consumers
who provide attention based on the rate of utility, $q$, that they
receive from the content units they consume. With a declining relationship
between an increase in content units from a media source and the consumers
expected rate of utility from those units, it is shown that the level
of content units consumed by the consumer in a period of time depends
on the underlying distribution, $D$, of potentially available content
units in that period, $W(D)$. This relation is a crucial part of
the value function for both the media source and a sole producer.
Where the sole producer values the combination of producer unit value
and consumer participation, $W(D)$, and the media source values
the combination of producer unit value and potential consumer participation,
$M(D)$.

The algorithm generates a distribution of content, $D^{*}$, that
optimizes relative to the value functions of the media source and
as a byproduct optimizes the consumer's participation and the result
for the sole producer. The distribution, $D^{*}$, is shown to correspond
with a Nash equilibrium between the media source, consumer, and producer
under very general conditions. Even when there is an increase in expected
unit producer value, neither the sole producer or media source have
incentive to move away from $D^{*}$. 

The analysis of the algorithmic procedure shows that the best possible
distribution beyond $D^{*}$ cannot have both an increase in expected
unit producer value and an increase in the potential participation
volume, $M(D)$, from the consumer. The analysis also shows that only
with an increase in the potential volume will there be a possibility
of going beyond $D^{*}$. Depending on the slope of the increase of
the potential volume relative to the change in the number of units,
there are two possibilities of going beyond $D^{*}$. The first occurs
when the slope is positive and $<1$. This requires a carveout of
$D^{*}$ as detailed in Algorithm \ref{alg:GeneratingCarveout}, and
leads to a Pareto efficient Nash equilibrium under realistic conditions.
The other occurs when the slope is $\geq1$ and leads to a second
local Nash equilibrium that is preferred over the original by the
media source and consumer and at least indifferent to the sole producer. 

The mechanism depends on the media source to make the decisions, but
under the influence of incentives and disincentives from the consumer
and sole producer. These incentives could be monetary or a threat
to disengage with the media source. Since the result of the algorithms
is Nash equilibrium, the incentives do not really come into play unless
there are other considerations causing deviations. 

Cases where this normative mechanism will not apply were also discussed,
although in some of those cases it may be beneficial for all the participants
to utilize the mechanism. Further in these cases, the model provides
useful insights to understand where inefficiencies exist, and also
where other mechanisms lead to inferior outcomes.

While consumer utility and participation dominates the analysis, overall
producer participation benefits as well via the gains of the sole
producer. Further analysis of individual producer participation is
left to future research.

A comparison with the literature for Recommender systems shows that
the approach here provides some distinct advantages in general and
especially in cases combining multiple objectives where the different
objectives can be considered as producer or consumer objectives. Consumer
participation is the common metric shared between both the consumer
and the sole producer that leads to a unique and endogenous mechanism.
The mechanism in turn provides both a Nash equilibrium and Pareto
efficient way for combining multiple objectives between consumers
and producers as well as the endogenous determination of an optimal
content volume for each consumer. %

\nocite{*}
\printbibliography

---------
\noindent \begin{flushleft}
\noindent \begin{flushleft}
\pagestyle{empty}\textbf{\Large{}\hspace{0.4\columnwidth}Mathematical
Appendix}{\Large\par}
\par\end{flushleft}

\section*{Proofs for Lemmas \ref{lem:For-a-DistributionSequence-MathMax} -
\ref{lem:DeltaV-between-2-local-optimal} and Derivations of Equations
\ref{eq:Xlkappa-Threshold}, \ref{eq:Xlkappa-Reactive}, and \ref{eq:Xukappa-threshold}.}
\begin{itemize}
\item Proof of \textbf{Lemma }\ref{lem:For-a-DistributionSequence-MathMax}:
This is obvious from optimizing the cases where $M>N$ and $M<N$
as the optimal for each case converges to the point where $M=N$. 
\item For proof of \textbf{Lemma \ref{lem:for-DeltaS-Below-D*}} see Proposition
\ref{prop:+DelS-before-D*} below.
\item For proof of \textbf{Lemma \ref{lem:LimitsOnTp-M'-afterD*}} see Proposition
\ref{prop:For-a-point-added-after-D*} and note that $\frac{T(p_{r_{2}})}{E(T|D^{*})}>1$
implies that $\frac{E(T|D')}{E(T|D^{*})}>1$. 
\item For proof of \textbf{Lemma \ref{lem:For+TpNoPosDeltaS-afterD*}} see
Proposition \ref{Prop:F_low-gt-F_up} and also propositions \ref{prop:The-functions-Mr2-Mar2},
\ref{prop:A-SC-for-No-Gain}, and \ref{prop:A-SC-For-Gain}. When
as in Proposition \ref{Prop:F_low-gt-F_up}\textbf{ $F_{low}(N_{r_{2}})>F_{up}(N_{r_{2}})$
}holds, where $F_{low}$ is the lower bounds needed for $\Delta_{S}>0$
and $F_{up}$ is the upper bonds needed for the point $r_{1}$ to
be added before point $r_{2}$. However, since \textbf{$F_{low}(N_{r_{2}})>F_{up}(N_{r_{2}})$}
the threshold for positive $\Delta_{S}$ is a range where the point
$r_{2}$ is not viable (meaning it cannot be added after $D^{*}$).
Thus at least under restrictions on $N_{r_{2}}$ and $N_{r_{1}}$
from the proposition, a $\Delta_{S}>0$ is not possible.
\item For derivations of threshold \textbf{equation \ref{eq:Xlkappa-Threshold}},
the reactive case of $X_{l\kappa}$, see derivation of equation \ref{eq:Xl-condition}
and \ref{eq:Delta-S-Inequality-with-kappa} below. 
\item The derivation of \textbf{equation \ref{eq:Xlkappa-Reactive}}, the
reactive case of $X_{l\kappa}$, follows: When choosing $D'$ the
media source and sole producer will have combined two period value
of $E(T|D')N(D')+\frac{M'}{N'}E(T|D^{*})M(D^{*})$. When not choosing
$D'$ the media source and sole producer will have combined value
of $2E(T|D^{*})M(D^{*})$. Comparing the two choices:

\[
E(T|D')N'+\frac{M'}{N'}E(T|D^{*})M(D^{*})>2E(T|D^{*})M(D^{*})
\]

Substituting $N'=1+N_{r_{2}}$, $M(D^{*})=1$ , $M(D')=1+\kappa N_{r_{2}}$,
and $E(T|D')=\frac{E(T|D^{*})+N_{r_{2}}T(p_{r_{2}})}{1+N_{r_{2}}}$
yields,

\[
\frac{M'}{N'}E(T|D^{*})>2E(T|D^{*})-E(T|D')N'
\]

\[
1+\kappa N_{r_{2}}>\frac{\left(1+N_{r_{2}}\right)\left[2E(T|D^{*})-E(T|D')\left(1+N_{r_{2}}\right)\right]}{E(T|D^{*})}
\]

and as a restriction on $\kappa$,

\[
\kappa>1-\frac{T(p_{r_{2}})}{E(T|D^{*})}\left(1+N_{r_{2}}\right)
\]

\item For the derivation of \textbf{equation \ref{eq:Xukappa-threshold}}
see derivation of equations \ref{eq:Xu-condition} and \ref{eq:limit-on-kappa-ar2}
below. 
\item For proof of \textbf{Lemma \ref{lem:For-0-lt-kappa-lt-1}} see Facts
\ref{fact:UnderB3-Xl-lt-Xu}, \ref{fact:Xu(0)=00003D1 FD-lt-1}, \ref{fact:B4-Xu-FD-lt-1},
\ref{fact:kappa-lt-1-only}, and \ref{fact:SC-4-kappa-r2-lt-1} below.
\item For proof of \textbf{Lemma \ref{lem:Impact-Tp1-on-viability}} see
Fact \ref{fact:Xl-gt-Xu-B4} below.
\item For proof of \textbf{Lemma \ref{lem:limits on Xu - Xl }} see Facts
\ref{fact:SC-4-kappa-r2-lt-1}, \ref{fact:KappaDif-gt-Xuk-Xlk}, and
\ref{fact:Xuk-Xlk-properties} below.
\item Proof of \textbf{Lemma} \ref{lem:CriteriaForD*+}

Let, $D^{-*}=D^{*}-Y$ where $Y\subset D^{*}$, and let $D'^{+}=D^{*}-Y+I_{r_{2}}$
where $I_{r_{2}}=D'-D^{*}$. Also, normalization $M(D^{*})=N(D^{*})=1$
(see details below). For the consumer to prefer $D'^{+}$ this inequality
must hold,

\[
N_{r_{2}}c_{r_{2}}+Q(D^{-*})\left(1-N_{Y}\right)\geq Q(D^{*})
\]

Substituting $\left(1-N_{Y}\right)Q(D^{-*})=Q(D^{*})-N_{Y}Q(Y)$ yields
$N_{r_{2}}c_{r_{2}}-N_{Y}Q(Y)>0$ proving the inequality holds when
the gain to the consumer is positive. The condition for this is:

\begin{equation}
Q(Y)N_{Y}\leq N_{r_{2}}c_{r_{2}}\label{eq:ConsumerConditionForSwap}
\end{equation}

This simplifies to $S_{c}=\sum_{x\in Y}N_{x}c_{x}$ and $S_{c}\leq c_{r_{2}}N_{r_{2}}$.
Also $\ensuremath{\bar{N_{Y}}}=\ensuremath{\sum_{x\in Y}N_{x}}$

The formula for $N_{Y}$ is, 

\begin{equation}
N_{Y}=N'-M\left(\frac{Q(D^{*})-S_{c}+N_{r_{2}}c_{r_{2}}}{1-\bar{N_{Y}}+N_{r_{2}}}\right)
\end{equation}

Thus holding $\bar{N_{Y}}$ fixed at a value $<N_{Y}$, if at least
one $c_{x}$ decreases $S_{c}$ will decrease and $M$ will increase.
As more points are added to $Y$, $\bar{N_{Y}}$ increases, and further
increasing $M$, thus decreasing $N_{Y}$ until $\bar{N_{Y}}\approxeq N_{Y}$.
Thus $N_{Y}<N(D')-M(D')$ decreases as $Q(Y)$ decreases. Thus proving
$\frac{dN_{Y}}{dQ(Y)}>0$ and the first part of the Lemma. 

For the sole producer to be at least indifferent between $D^{*}$
and $D'^{+}$ this partial inequality must hold.

\[
E(T|D^{-*})\left(1-N_{Y}\right)+N_{r_{2}}T(p_{r_{2}})\geq E(T|D^{*})M(D^{*})
\]

Note that $M(D^{*})=1$ and $E(T|D^{-*})=\frac{E(T|D^{*})-E(T|Y)N_{Y}}{\left(1-N_{Y}\right)}$
and substituting,

\[
E(T|D^{*})-E(T|Y)N_{Y}+N_{r_{2}}T(p_{r_{2}})\geq E(T|D^{*})
\]

Reducing yields,

\begin{equation}
N_{r_{2}}T(p_{r_{2}})\geq E(T|Y)N_{Y}
\end{equation}

Thus completing the proof. 
\item Proof of \textbf{Lemma \ref{lem:For-kappa-geq-1}}: The condition
for $\Delta_{V}>0$ is,
\[
\Delta_{V}=\left[\frac{N(D^{*})}{N(D')}E(T|D^{*})+\left(1-\frac{N(D^{*})}{N(D')}\right)T(p)\right]M(D')-E(T|D^{*})M(D^{*})
\]

Normalizing $N$ around $D^{*}$ as $M(D^{*})=N(D^{*})=1$ (see rationale
below)%

\[
\Delta_{V}=\left[\frac{1}{N(D')}E(T|D^{*})+\frac{N(D')-1}{N(D')}T(p)\right]M(D')-E(T|D^{*})>0
\]

\[
\left[E(T|D^{*})+(N(D')-1)T(p)\right]M(D')-N(D')E(T|D^{*})>0
\]

\[
\left[1+(N(D')-1)\frac{T(p)}{E(T|D^{*})}\right]M(D')>N(D')
\]

\[
M(D')>\frac{N(D')}{1+(N(D')-1)\frac{T(p)}{E(T|D^{*})}}
\]

\[
M(D')>\frac{1+N_{r_{2}}}{1+N_{r_{2}}\frac{T(p)}{E(T|D^{*})}}
\]

\[
M(D')=1+\kappa_{r_{2}}N_{r_{2}}>\frac{1+N_{r_{2}}}{1+N_{r_{2}}\frac{T(p)}{E(T|D^{*})}}
\]

\[
\kappa_{r_{2}}N_{r_{2}}>\frac{1+N_{r_{2}}-1+N_{r_{2}}\frac{T(p)}{E(T|D^{*})}}{1+N_{r_{2}}\frac{T(p)}{E(T|D^{*})}}=\frac{N_{r_{2}}-N_{r_{2}}\frac{T(p)}{E(T|D^{*})}}{1+N_{r_{2}}\frac{T(p)}{E(T|D^{*})}}
\]

\[
\kappa_{r_{2}}>\frac{1-\frac{T(p)}{E(T|D^{*})}}{1+N_{r_{2}}\frac{T(p)}{E(T|D^{*})}}
\]

Thus, showing $X_{l\kappa}>1\iff\frac{T(p_{2})}{E(T|D^{*})}<0$ and
$X_{l\kappa}=1\iff\frac{T(p_{2})}{E(T|D^{*})}=0$.
\item Proof of \textbf{Lemma \ref{lem:DeltaV-between-2-local-optimal}}:
the $\Delta_{V}$ between $D^{*}$ and $D^{2*}$ can be determined
by noting that $M(D^{*})=N(D^{*})$ and $M(D^{2*})=N(D^{2*})$ and
substituting into equation \ref{eq:MediaSourceIncrementalValueMapping}
to get: 

\[
\Delta_{V}=\left[\frac{N(D^{*})}{N(D^{2*})}E(T|D^{*})+\left(1-\frac{N(D^{*})}{N(D^{2*})}\right)T(p)\right]M(D^{2*})-E(T|D^{*})M(D^{*})
\]

\[
\Delta_{V}=\left[N(D^{*})E(T|D^{*})+\left(N(D^{2*})-N(D^{*})\right)T(p)\right]\frac{M(D^{2*})}{N(D^{2*})}-E(T|D^{*})M(D^{*})
\]

Letting $M(D^{*})=N(D^{*})=1$ and noting $T(p)=0$ yields,

\[
\Delta_{V}=E(T|D^{*})+\left(N(D^{2*})-1\right)T(p)-E(T|D)=0
\]

\end{itemize}
\pagestyle{fancy}
\fancyhead{}
\fancyhead[R]{\textbf{\small Appendix}}

\section*{Distribution Math}

As in the main paper, a distribution $D$ consists of a probability
density function or mapping $\Phi$ over a discrete multidimensional
region $R$ and a volume of scalar $N$. The frequency for the distribution
is $N\phi_{r}\text{ \ensuremath{\forall r\in R}}$. Combining two
distributions $D_{1}$ and $D_{2}$ simply involves the combination
of the frequencies at each point in the combined region. So for $D'=D_{1}+D_{2}$
the frequency for $D'$ at each point $r\in R_{1}\cup R_{2}$ is $N_{1}\phi_{1_{r}}+N_{2}\phi_{2_{r}}$
and the density for each point $r\in R_{1}\cup R_{2}$ is $\frac{1}{N_{1}+N_{2}}\left[N_{1}\phi_{1_{r}}+N_{2}\phi_{2_{r}}\right]$. 

A special distribution used in the analysis is the distribution for
a single point $r=\{c,p\}$ denoted as $I_{r}$.
\begin{assumption}[Additive Q]
\label{assu:Q-For-2-non-overlapping}For 2 non-overlapping distributions
$D_{1}$ and $D_{2}$ with the combined distribution $D_{12}=D_{1}+D_{2}$
and regions such that $R_{12}=R_{1}+R_{2}$, the following relation
for $Q$ holds 
\begin{equation}
Q(D_{12})=\frac{N_{1}}{N_{1}+N_{2}}Q(D_{1})+\frac{N_{2}}{N_{1}+N_{2}}Q(D_{2})\label{eq:combined-q}
\end{equation}
.
\end{assumption}
This assumption is inclusive of most, if not all, practical functions
for $Q$.%

\begin{prop}
\label{prop:For-OrderAdded}Let $D''$ be the distribution after two
incremental changes, and let $1,2$ or $2,1$ be the order of the
incremental changes, then $Q(D''|1,2)=Q(D''|2,1)$.
\end{prop}
\begin{proof}
This follows from the assumption \ref{assu:Q-For-2-non-overlapping}
on $Q$ for altering the order that distributions are combined:

\[
Q(D_{R'})=\frac{N}{N+N_{I}}Q(D_{R})+\frac{N_{I}}{N+N_{I}}Q(D_{I})
\]
Where $R'-R=I$ with $D_{I}$ and scalar $N_{I}$. Using the formula
to add the increments $I_{1}$ and $I_{2}$ in different order yields
the same result. 

\[
Q(D''|2,1)=Q(D''|2,1)=\frac{N}{N+N_{I_{1}}+N_{I_{2}}}\left[NQ(D_{R})+N_{I_{1}}Q(D_{I_{2}})+N_{I_{2}}Q(D_{I_{2}})\right]
\]
\end{proof}
\begin{flushleft}
Proposition \ref{prop:For-OrderAdded} confirms that the order that
incremental distributions are added does not impact $q$.
\par\end{flushleft}

\subsection*{Value Functions}

Let $\Phi$ and $\Phi'$ be a density functions on a finite discrete
region $R$ that bounds a set of combinations, with $\sum_{\forall r\in R}\phi_{r}=1$
and $0\leq\phi_{r}<1$. Let $T$ be a function or mapping with a value
for each point in the region. So, the expected value for $T$ over
the region $R$ given $\Phi$ would be $E(T|\Phi)=\sum_{\forall r\in R}\phi_{r}T_{r}$
, where $T_{r}$ is the value of $T$ at the point $r$. Let $D$
be the distribution with density function $\Phi$ and volume $N$,
and let $D'$ be the distribution with density $\Phi'$ and volume
$N'$. Further let Y be a distribution such that $D'=D+Y$ and let
$H(T|Y)$ be $E(T|\Phi_{Y})$. Then,

\begin{equation}
E(T|\Phi')=E(T|\Phi)\frac{N}{N'}+H(T|Y)\frac{N'-N}{N'}\label{eq:ExpectedT}
\end{equation}

Let $Z(D_{R})$ be a function determined from the distribution $D_{R}$
. For example, the potential participation for $D_{R}$, where $Z(D_{R})$
would be $M(Q(D_{R}))$.
\begin{fact}
\label{fact:value-change}For a value function $V=E(T|\Phi)Z$, where
$Z$ is a value determined from the distribution $D$ and $Z'$ is
determined from distribution $D'$, the change in the value function
\begin{equation}
\Delta_{V}=V'-V=E(T|\Phi')Z'-E(T|\Phi)Z=\left(E(T|\Phi)\frac{N}{N'}+H(T|Y)\frac{N'-N}{N'}\right)Z'-E(T|\Phi)Z\label{eq:DeltaV}
\end{equation}
\end{fact}
\begin{proof}
Plugging equation \ref{eq:ExpectedT} into $V'-V$, yields 

\[
\begin{aligned}\begin{aligned}V'-V=\end{aligned}
 & \left(E(T|\Phi)\frac{N}{N'}+H(T|Y)\frac{N'-N}{N'}\right)Z'-E(T|\Phi)Z=\\
 & E(T|\Phi)(\frac{N}{N'}Z'-Z)+H(T|Y)\frac{N'-N}{N'}Z'
\end{aligned}
\]
\end{proof}

\subsection*{Special Cases of $\Delta_{V}$}

\paragraph*{Unit Incremental Distributions}

Next consider $\triangle_{V}$ when $D'$ results from an unit incremental
change in $D$. Let $I_{r}$ be the increment distribution, so $D'=I_{r}+D$
and $N_{I_{r}}=1$. Thus, $E(T|D')=E(T|D)\frac{N}{N+1}+T_{r}\frac{1}{N+1}$
, and %
\begin{equation}
\Delta_{V}=E(T|D)\left(\frac{N}{(N+1)}Z'-Z\right)+T_{r}\frac{Z'}{N+1}\label{eq:incremental-distribution-shift}
\end{equation}

\paragraph*{Incremental Point Shifts}

These shifts will occur at one point, but will not be limited by size.
Effectively, for an incremental point added the range is $0<(N'-N)<\infty$.
Thus $\Delta_{V}$ from a distribution change at point $r$ is,

\begin{equation}
\Delta_{V}=E(T|D)\left(\frac{N}{N'}Z'-Z\right)+T_{r}\frac{(N'-N)}{N'}Z'\label{eq:IncPointShiftValChange}
\end{equation}

\paragraph*{General Distribution Shifts}

For a more general shift let $Y$ be a distribution with standard
density $\Phi$ with $-\infty<N_{Y}<\infty$ subject to the limitation
that for each point $r\in R'$ is $N\phi_{r}+N_{Y}\phi_{Y_{r}}\geq0$.
Then for the case where $D'=D+Y$ we have the most general case of
a distribution change. Thus $N'=N+N_{Y}$ and 

\begin{equation}
E(T|D')=E(T|D)\frac{N}{N'}+H(T|Y)\frac{N_{Y}}{N'}\label{eq:ShiftInProducerUnitValue}
\end{equation}

This form includes the incremental cases above. Here the distribution
change is,

\begin{equation}
\Delta_{V}=E(T|D)\left(\frac{N}{N'}Z'-Z\right)+H(T|Y)\frac{N_{Y}}{N'}Z'\label{eq:GeneralShiftValueChange}
\end{equation}

It should be clear that equations \ref{eq:incremental-distribution-shift},
and \ref{eq:IncPointShiftValChange} are special cases of equation
\ref{eq:GeneralShiftValueChange} and both reductive and expansive
increments are included in the formulation.

\subsection*{Value Function with $M$}

Letting $Z=M(D)$ and $Z'=M(D')$ in equation \ref{eq:GeneralShiftValueChange},
simplifying, and specifying for a single point with production value
$p$ yields:

\begin{equation}
\Delta_{V}=E(T|D)N(D)\left(\frac{M(D')}{N(D')}-\frac{M(D)}{N(D)}\right)+T(p)\left(N(D')-N(D)\right)\frac{M(D')}{N(D')}\label{eq:RefinedValueChange}
\end{equation}

Equation \ref{eq:RefinedValueChange} is the same as \ref{eq:MediaSourceIncrementalValueMapping}.
This equation can also be expressed in terms of the probability of
an incremental point $r$ with distribution $I_{r}$ as,

\begin{equation}
\Delta_{V}=\left[(1-\phi_{r})E(T|D)+\phi_{r}T(p)\right]M(D')-E(T|D)M(D)\label{eq:DeltaV-with-phi}
\end{equation}

Where $D'=D+I_{r}$ and $\phi_{r}=1-\frac{N(D)}{N(D')}$ is the probability
of incremental point $r$, and $1-\phi_{r}$ is the probability of
the points in the initial distribution. %

The condition for $\Delta_{V}>0$ can be derived from either Equation
\ref{eq:RefinedValueChange} or \ref{eq:DeltaV-with-phi} as: %

\begin{equation}
\left[(1-\phi_{r})+\phi_{r}\frac{T(p)}{E(T|D)}\right]>\frac{M(D)}{M(D')}\implies\Delta_{V}>0\label{eq:ConditionFor-DeltaV-gt-0}
\end{equation}

or equivalently,

\begin{equation}
\left(N(D)\frac{1}{N(D')}\right)+\frac{T(p)}{E(T|D)}\left(N(D')-N(D)\right)\frac{1}{N(D')}>\frac{M(D)}{M(D')}
\end{equation}

Note if $D=D^{*}$ and normalizing volume at $D^{*}$ so $N(D)=M(D)=1$
and substituting yields, %

\[
\frac{T(p)}{E(T|D)}\left(N(D')-1\right)>\frac{N(D')-M(D')}{M(D')}
\]

further defining$\text{ \ensuremath{\frac{M(D')-1}{N(D')-1}}=\ensuremath{\kappa}}$
and $N_{r_{2}}+1=N(D')$, and substituting yields,%

\begin{equation}
\kappa>\frac{1-\frac{T(p)}{E(T|D)}}{\frac{T(p)}{E(T|D)}N_{r_{2}}+1}
\end{equation}

\subsection*{Alternative Value Functions}

For simplicity when comparing points to choose for the best incremental
increase, it is only necessary to compare $\left[(1-\phi_{r})E(T|D)+\phi_{r}T(p)\right]M(D')$
as the term $-\frac{M(D)}{N(D)}E(T|D)N(D)$ will always be the same
for the points to be compared.\footnote{note $\phi_{r}=1-\frac{N(D)}{N(D')}$ and $1-\phi_{r}=\frac{N(D)}{N(D')}$ }
Further can break down $M(D')$ as $M((1-\phi_{r})Q(D)+\phi_{r}c)$.
The comparison function becomes:

\[
\Xi(c,p,\phi_{r}|D)=\left[(1-\phi_{r})E(T|D)+\phi_{r}T(p)\right]M((1-\phi_{r})Q(D)+\phi_{r}c)
\]

Rearranging,

\begin{equation}
\Xi(c,p,\phi_{r}|D)=\left[E(T|D)+\phi_{r}\left(T(p)-E(T|D)\right)\right]M(Q(D)+\phi_{r}(c-Q(D)))\label{eq:ApendixAltVF}
\end{equation}

This is generally easier to work with and otherwise equivalent to
$\Delta_{V}$. It is clear from the formulation above that as long
as $c\leq Q(D)$ and $D_{0}$ includes a point with the highest available
$c$, that $M$ is going to decline or stay the same with the addition
of the next point.
\begin{prop}
\label{prop:Equivalent-cost-functions}The value functions $\Lambda(r,N_{r},D)=\Delta_{V}$
and $\Lambda(r,N_{r},D)=\Xi$ are equivalent when applied to determine
the preferred incremental distribution relative to a common base distribution
$D$. 
\end{prop}
\begin{proof}
The equivalence of $\Delta_{V}$ and $\Xi$ can be shown. Since functions
$\Delta_{V}$ and $\Xi$ differ only by a constant so any ranking
of incremental distributions or shifts relative to distribution $D$
will be the same for either $\Delta_{V}$ or $\Xi$.
\end{proof}
Properties of $\Xi$:

\[
\frac{\partial\Xi}{\partial\phi_{r}}=E(T|D)\frac{dM}{dq}(c-Q(D))+\left(T(p)-E(T|D)\right)\left(M+\frac{dM}{d(q)}(c-Q(D))\phi_{r}\right)
\]

\[
\frac{\partial\Xi}{\partial c}=\frac{dM}{dq}\phi_{r}
\]

\[
\frac{\partial\Xi}{\partial p}=\frac{dT}{dp}M\phi_{r}
\]

If there is a direction where $c$ decreases and $p$ increases the
ratio $\frac{\frac{dM}{dq}}{\frac{dT}{dp}M}$ specifies a threshold
for how many units of $p$ per unit of $c$ are needed at the margin
to increase $\Xi$.%

As in equation \ref{eq:incremental-distribution-shift} there are
cases where it might be desired to assume $N'-N=1$. Applying this
to equation \ref{eq:ApendixAltVF} yields:

\[
\varUpsilon=\left[E(T|D)\frac{N(D)}{N(D)+1}+\frac{1}{N(D)+1}T(p)\right]M(Q(D)-\frac{1}{N(D)+1}(Q(D)-c))
\]

Properties of $\Upsilon$:

\[
\frac{\partial\varUpsilon}{\partial c}=\frac{dM}{dq}\frac{1}{N(D)+1}
\]

\[
\frac{\partial\varUpsilon}{\partial p}=\frac{dT}{dp}M\frac{1}{N(D)+1}
\]

As with $\Xi$, if there is a direction where $c$ decreases and $p$
increases, the ratio $\frac{\frac{dM}{dq}}{\frac{dT}{dp}M}$ specifies
a threshold for how many units of $p$ per unit of $c$ are needed
at the margin to increase $\varUpsilon$.

Also, as a value function, $\varUpsilon$ is equivalent to $\Delta_{V}$
for the same reasons as for $\Xi$.

\section*{Sequences of Incremental Distributions}

This section provides analysis of conditions for the order incremental
distributions would be added in the sequence of distributions $\Theta$.
The resulting restrictions are useful in further proofs. 

\subsection*{Comparing Initial Distributions}

Consider two possible initial distributions $Y_{a}$ and $Y_{b}$
with value functions respectively $V_{a}=E(T|Y_{a})M(Y_{a})$ and
$V_{b}=E(T|Y_{b})M(Y_{b})$. Thus initial distribution $Y_{a}$ is
preferred over $Y_{b}$ when,

\[
E(T|Y_{a})M(q_{a})>E(T|Y_{b})M(q_{b})
\]

Assuming $M(q)=\zeta q^{\alpha}$ is at least an approximation of
$M(q)$ where $0<\alpha<1$, then 

\[
\frac{M(q_{a})}{M(q_{b})}=\left(\frac{q_{a}}{q_{b}}\right)^{\alpha}>\frac{E(T|Y_{b})}{E(T|Y_{a})}
\]

Obviously, if both $q_{b}<q_{a}$ and $E(T|Y_{b})<E(T|Y_{a})$, then
$Y_{a}$ is preferred. So consider the case where $q_{b}<q_{a}$ and
$E(T|Y_{b})>E(T|Y_{a})$. In this case, it is possible, especially
for small values of $\alpha$ that, 

\begin{equation}
\frac{q_{a}}{q_{b}}>\frac{E(T|Y_{b})}{E(T|Y_{a})}>\left(\frac{q_{a}}{q_{b}}\right)^{\alpha}\label{eq:Ya-vs-Yb}
\end{equation}

Indicating that that percentage advantage of $q_{a}$ over $q_{b}$
must be significantly greater than the advantage of $E(T|Y_{b})\text{ over }E(T|Y_{a})$.

Consider the case when $Y_{a}$ is a single point, $\{c,p\}$, so,

\begin{equation}
\frac{M(q_{a})}{M(q_{b})}=\left(\frac{c}{q_{b}}\right)^{\alpha}>\frac{E(T|Y_{b})}{T(p)}\label{eq:r-vs-Ya}
\end{equation}

\subsection*{Sequence Order with a Distribution Shift\label{subsubsec:Sequence-Order}}

In this subsection, conditions for a point $r$ to be added to the
sequence after a distribution $Y$ are analyzed.
\begin{description}
\item [{Sequence\ Scenario\ A:}] A distribution shift with an initial
distribution $D_{a}$ and a general incremental distribution $Y$
covering potentially multiple points as in equation \ref{eq:GeneralShiftValueChange},
and a second distribution $D_{b}$ such that $D_{b}=D_{a}+Y$. Also
let $N_{a}$ be the volume of the initial distribution, and $N_{Y}$
be the volume of the incremental distribution $Y$. 
\end{description}
The formula for $\Xi$ in this scenario is:
\[
\Xi=\left(E(T|D_{a})\frac{N_{a}}{N_{b}}+H(T|Y)\frac{N_{a}-N_{b}}{N_{b}}\right)M_{b}
\]

Assume $Q(D_{b})<Q(D_{a})$ so that $M_{b}<M_{a}$, and consider adding
an incremental point $r=\{c,p\}$, where $c>Q(D_{b})$. So $D'=D_{b}+I_{r}$
yields:

\[
Q(D')=\left[\frac{N_{a}}{N_{a}+N_{Y}}Q(D_{a})+\frac{N_{Y}}{N_{a}+N_{Y}}Q(Y)\right]\frac{N_{a}+N_{Y}}{N_{a}+N_{Y}+N_{r}}+c\frac{N_{r}}{N_{a}+N_{Y}+N_{r}}
\]

Note from above and Proposition \ref{prop:For-OrderAdded} that $D_{a}+I_{r}+Y$
results in the same distribution regardless if $I_{r}$ is added to
$D_{a}$ first or if $Y$ is added first. However, knowing conditions
for adding $Y$ before $r$ provides a potentially useful constraint.
For $Y$ to be added to in the sequence before $I_{r}$ this inequality
of the value functions must hold,

\begin{equation}
\left[E(T|D_{a})\frac{N_{a}}{N_{a}+N_{Y}}+H(T|Y)\frac{N_{Y}}{N_{a}+N_{Y}}\right]M_{b}>\left[E(T|D_{a})\frac{N_{a}}{N_{a}+N_{r}}+T(p)\frac{N_{r}}{N_{a}+N_{r}}\right]M_{ar}\label{eq:ConditionForOrder}
\end{equation}

Where $D_{ar}=D_{a}+I_{r}$ and $M_{ar}=M(Q(D_{ar}))$. Rearranging
and expressing $T(p)$ and $H(T|Y)$ relative to $E(T|D_{a})$ yields,%

\begin{equation}
\frac{\frac{N_{a}}{N_{a}+N_{Y}}+\frac{H(T|Y)}{E(T|D_{a})}\frac{N_{Y}}{N_{a}+N_{Y}}}{\frac{N_{a}}{N_{a}+N_{r}}+\frac{T(p)}{E(T|D_{a})}\frac{N_{r}}{N_{a}+N_{r}}}>\frac{M_{ar}}{M_{b}}>1\label{eq:RestrictionOnMar}
\end{equation}

Note that $\frac{M_{ar}}{M_{b}}>1$ because $c>Q(D_{b})$. This implies
that if $Y$ was added first over $I_{r}$ then $\frac{N_{a}}{N_{a}+N_{Y}}+\frac{H(T|Y)}{E(T|D_{a})}\frac{N_{Y}}{N_{a}+N_{Y}}>\frac{N_{a}}{N_{a}+N_{r}}+\frac{T(p)}{E(T|D_{a})}\frac{N_{r}}{N_{a}+N_{r}}$.
\textbf{}%

Assuming $M(q)=\zeta q^{\alpha}$ is at least an approximation of
$M(q)$ where $0<\alpha<1$, then $\frac{M_{ar}}{M_{b}}=\left(\frac{q_{ar}}{q_{b}}\right)^{\alpha}$.
Consider, 

\begin{align}
\frac{q_{ar}}{q_{b}}=\frac{Q(D_{ar})}{Q(D_{b})}= & \frac{\frac{N_{a}}{N_{a}+N_{r}}Q(D_{a})+\frac{N_{r}}{N_{a}+N_{r}}c}{\frac{N_{a}}{N_{a}+N_{Y}}Q(D_{a})+\frac{N_{Y}}{N_{a}+N_{Y}}Q(Y)}=\label{eq:Q-Ratio}\\
 & \frac{\frac{N_{a}}{N_{a}+N_{r}}+\frac{N_{r}}{N_{a}+N_{r}}\frac{c}{Q(D_{a})}}{\frac{N_{a}}{N_{a}+N_{Y}}+\frac{N_{Y}}{N_{a}+N_{Y}}\frac{Q(Y)}{Q(D_{a})}}>1\nonumber 
\end{align}

Substituting equation \ref{eq:Q-Ratio} into $\frac{M_{ar}}{M_{b}}=\left(\frac{q_{ar}}{q_{b}}\right)^{\alpha}$,
yields:

\[
\frac{M_{ar}}{M_{b}}=\left(\frac{\frac{N_{a}}{N_{a}+N_{r}}+\frac{N_{r}}{N_{a}+N_{r}}\frac{c}{Q(D_{a})}}{\frac{N_{a}}{N_{a}+N_{Y}}+\frac{N_{Y}}{N_{a}+N_{Y}}\frac{Q(Y)}{Q(D_{a})}}\right)^{\alpha}
\]

then substituting into equation \ref{eq:RestrictionOnMar} yields:

\begin{equation}
\frac{\frac{N_{a}}{N_{a}+N_{Y}}+\frac{H(T|Y)}{E(T|D_{a})}\frac{N_{Y}}{N_{a}+N_{Y}}}{\frac{N_{a}}{N_{a}+N_{r}}+\frac{T(p)}{E(T|D_{a})}\frac{N_{r}}{N_{a}+N_{r}}}>\frac{M_{ar}}{M_{b}}=\left(\frac{\frac{N_{a}}{N_{a}+N_{r}}+\frac{N_{r}}{N_{a}+N_{r}}\frac{c}{Q(D_{a})}}{\frac{N_{a}}{N_{a}+N_{Y}}+\frac{N_{Y}}{N_{a}+N_{Y}}\frac{Q(Y)}{Q(D_{a})}}\right)^{\alpha}>1
\end{equation}

Rearranging terms implies this condition on $c$:

\begin{equation}
\frac{c}{Q(D_{a})}<\frac{N_{a}+N_{r}}{N_{r}}\left(M_{b}\frac{\frac{N_{a}}{N_{a}+N_{Y}}+\frac{H(T|Y)}{E(T|D_{a})}\frac{N_{Y}}{N_{a}+N_{Y}}}{\frac{N_{a}}{N_{a}+N_{r}}+\frac{T(p)}{E(T|D_{a})}\frac{N_{r}}{N_{a}+N_{r}}}\right)^{\frac{1}{\alpha}}-\frac{N_{a}}{N_{r}}\label{eq:LimitForY-Before-c}
\end{equation}

The inequality in equation \ref{eq:LimitForY-Before-c} must hold
for the point $r$ to be added in the sequence after the distribution
$Y$. Further note that $c$ is not limited to be $<Q(D_{a})$.

\subsubsection*{Comparing Points for Sequence Order \label{subsubsec:Comparing-Points}}

In this subsection conditions for a point $r_{1}$ to be added in
the sequence before point $r_{2}$ are analyzed. A useful constraint
is derived to provide the threshold for when a point $r_{1}$ would
be added before point $r_{2}$. 

Letting $D_{a}$ be the initial distribution as above and $M_{ar}$
and $N_{r}$ be the respective values for increment for point $r\in\{r_{1},r_{2}\}$
and substituting into equation \ref{eq:ConditionForOrder} yields,

\begin{equation}
\left[\frac{N_{a}}{N_{a}+N_{r_{1}}}+\frac{T(p_{1})}{E(T|D_{a})}\frac{N_{r_{1}}}{N_{a}+N_{r_{1}}}\right]M_{ar_{1}}>\left[\frac{N_{a}}{N_{a}+N_{r_{2}}}+\frac{T(p_{2})}{E(T|D_{a})}\frac{N_{r_{2}}}{N_{a}+N_{r_{2}}}\right]M_{ar_{2}}\label{eq:OrderCond-r1-r2}
\end{equation}

\begin{fact}
\label{fact:For-two-points}For two points $r_{1}$ and $r_{2}$ with
the same incremental volume $(N_{I_{1}}=N_{I_{2}})$, point $r_{1}$
is preferred over $r_{2}$ if (i) $p_{r_{1}}>p_{r_{2}}$ and $c_{r_{1}}>c_{r_{2}}$,
or (ii) $p_{r_{1}}=p_{r_{2}}$ and $c_{r_{1}}>c_{r_{2}}$, or (iii)
$p_{r_{1}}>p_{r_{2}}$ and $c_{r_{1}}=c_{r_{2}}$.
\end{fact}
\begin{proof}
Utilizing equation \ref{eq:OrderCond-r1-r2} with $N_{r}=N_{r_{1}}=N_{r_{2}}$
reveals that 
\[
\left[\frac{N_{a}}{N_{a}+N_{r}}+\frac{T(p_{1})}{E(T|D_{a})}\frac{N_{r}}{N_{a}+N_{r}}\right]M_{ar_{1}}>\left[\frac{N_{a}}{N_{a}+N_{r}}+\frac{T(p_{2})}{E(T|D_{a})}\frac{N_{r}}{N_{a}+N_{r}}\right]M_{ar_{2}}
\]
\end{proof}
Note an alternative proof exits for $\varUpsilon$ and $\Delta_{V}$
when $T(p)$ is monotonically increasing in $p$, and $M$ is monotonically
increasing in c.
\begin{prop}
\label{prop:Preferred-When-Nr1>Nr2} For a preferred distribution
sequence and points $r_{1}$ and $r_{2}$, where $N_{r_{2}}<N_{r_{1}}<1=N_{a}$
, $r_{1}$ is added before $r_{2}$ when: 

\begin{doublespace}
\noindent (i) $T(p_{1})=T(p_{2})=E(T|D)$ and $1<\frac{1+N_{r_{2}}\frac{c_{r_{2}}}{Q(D)}}{1+N_{r_{2}}}<\frac{1+N_{r_{2}}\left(2\frac{c_{r_{2}}}{Q(D)}-1\right)}{1+N_{r_{2}}}<\frac{c_{1}}{Q(D)}<\frac{c_{2}}{Q(D)}$
or alternatively $N_{r_{2}}<\frac{1}{2\eta-1}$ and $N_{r_{1}}>\frac{\eta N_{r_{2}}}{1+N_{r_{2}}\left(1-\eta\right)}$,
where $\eta=\frac{c_{r_{2}}-Q(D)}{c_{r_{1}}-Q(D)}$ ; or 

\noindent (ii) $c_{1}=c_{2}=Q(D$) and $1<\frac{1+N_{r_{2}}\frac{T(p_{2})}{E(T|D_{a})}}{1+N_{r_{2}}}<\frac{1+N_{r_{2}}\left(2\frac{T(p_{2})}{E(T|D_{a})}-1\right)}{1+N_{r_{2}}}<\frac{T(p_{1})}{E(T|D_{a})}<\frac{T(p_{2})}{E(T|D_{a})}$
or alternatively $N_{r_{2}}<\frac{1}{2\varsigma-1}$ and $N_{r_{1}}>\frac{\varsigma N_{r_{2}}}{1+N_{r_{2}}-\varsigma N_{r_{2}}}$,
where $\varsigma=\frac{T(p_{2})-E(T|D_{a})}{T(p_{1})-E(T|D_{a})}$.
\end{doublespace}
\end{prop}
\begin{cor*}
The proposition further implies that for any $N_{a}$ or any $N_{r_{1}}$,
if $1=\frac{c_{1}}{Q(D_{a})}<\frac{c_{2}}{Q(D_{a})}$ and $1=\frac{T(p_{1})}{E(T|D_{a})}<\frac{T(p_{2})}{E(T|D_{a})}$,
$r_{2}$ is preferred to $r_{1}$.
\end{cor*}
\begin{proof}
Consider $D_{1}=D_{a}+I_{r_{1}}$ and $D_{2}=D_{a}+I_{r_{2}}$. If
$r_{1}$ is added to the arbitrary distribution $D_{a}$ before $r_{2}$
this inequality must hold,
\[
\left[\frac{N_{a}}{N_{a}+N_{r_{1}}}+\frac{T(p_{1})}{E(T|D_{a})}\frac{N_{r_{1}}}{N_{a}+N_{r_{1}}}\right]M_{ar_{1}}>\left[\frac{N_{a}}{N_{a}+N_{r_{2}}}+\frac{T(p_{2})}{E(T|D_{a})}\frac{N_{r_{2}}}{N_{a}+N_{r_{2}}}\right]M_{ar_{2}}
\]

For case (i) since $T(p_{1})=T(p_{2})=E(T|D)$ only $M_{ar_{1}}\text{ and }M_{ar_{2}}$
are of interest. In general,$\text{ for }D'=D_{a}+I_{r}$ 
\[
Q(D')=\frac{N_{a}}{N_{a}+N_{r}}Q(D_{a})+\frac{N_{r}}{N_{a}+N_{r}}c
\]

Thus, $Q(D'(N_{r}))$ is increasing in $N_{r}$ for $c>Q(D_{a})$
. 

For $M_{ar_{1}}>M_{ar_{2}}$ it must be that $Q(D_{1})>Q(D_{2})$
and the threshold on $N_{r_{1}}$ for that inequality will be the
value of $N_{r_{1}}$ that solves, %
{} 
\[
Q(D_{1})=\frac{N_{a}}{N_{a}+N_{r_{1}}}Q(D_{a})+\frac{N_{r_{1}}}{N_{a}+N_{r_{1}}}c_{r_{1}}=\frac{N_{a}}{N_{a}+N_{r_{2}}}Q(D_{a})+\frac{N_{r_{2}}}{N_{a}+N_{r_{2}}}c_{r_{2}}=Q(D_{2})
\]

Rearranging, normalizing all $N$ values relative to $N_{a}=1$, and
presenting $c$ relative to $Q(D_{a})$, simplifies to,
\[
1+N_{r_{1}}\frac{c_{r_{1}}}{Q(D_{a})}=\frac{1+N_{r_{1}}}{1+N_{r_{2}}}\left(1+N_{r_{2}}\frac{c_{r_{2}}}{Q(D_{a})}\right)
\]

Rearranging further yields a simplified equation for the threshold
as,%
\[
N_{r_{1}}\left[\left(1+N_{r_{2}}\right)\left(c_{r_{1}}-Q(D_{a})\right)-N_{r_{2}}\left(c_{r_{2}}-Q(D_{a})\right)\right]=N_{r_{2}}\left(c_{r_{2}}-Q(D_{a})\right)
\]

Expressing the threshold as an inequality for $r_{1}$ to be added
before $r_{2}$, and letting $\eta=\frac{c_{r_{2}}-Q(D)}{c_{r_{1}}-Q(D)}$ 

\[
N_{r_{1}}\left[1+N_{r_{2}}-\eta N_{r_{2}}\right]>N_{r_{2}}\eta
\]

When $\frac{1+N_{r_{2}}}{N_{r_{2}}}<\eta$ the $L.H.S<0$ and the
inequality cannot hold for values of $N_{r_{1}}>0$. Also when $\frac{1+N_{r_{2}}}{N_{r_{2}}}=\eta$
the denominator is zero and $N_{r_{1}}$ is undefined. So the only
viable case is when 

\begin{equation}
\frac{1+N_{r_{2}}}{N_{r_{2}}}>\eta\label{eq:eta-threshold}
\end{equation}

\begin{equation}
N_{r_{2}}<\frac{1}{\eta-1}\label{eq:Nr2-threshold}
\end{equation}
 where the inequality for $N_{r_{1}}$ is,

\begin{equation}
N_{r_{1}}>\frac{N_{r_{2}}\eta}{1+N_{r_{2}}\left(1-\eta\right)}\label{eq:Nr1-Threshold}
\end{equation}

Alternatively substituting for $\eta$ into inequality yields this
restriction on $c_{r_{1}}$,%

\begin{equation}
\frac{c_{r_{1}}}{Q(D_{a})}>\frac{1+N_{r_{2}}\frac{c_{r_{2}}}{Q(D_{a})}}{1+N_{r_{2}}}\label{eq:Cr1-threshold-0}
\end{equation}

 With the constraint $N_{r_{1}}<1$, this additional more restrictive
restriction would apply. %

\begin{equation}
\frac{c_{r_{1}}}{Q(D)}>1+\frac{2N_{r_{2}}}{1+N_{r_{2}}}\left(\frac{c_{r_{2}}}{Q(D)}-1\right)\label{eq:Cr1-threshold-1}
\end{equation}

Further, it can be shown that,%

\[
N_{r_{1}}<1\implies N_{r_{2}}<\frac{1}{2\eta-1}
\]

Also note that $\frac{1+N_{r_{2}}\frac{c_{r_{2}}}{Q(D_{a})}}{1+N_{r_{2}}}<\frac{c_{r_{1}}}{Q(D_{a})}<\frac{c_{r_{2}}}{Q(D_{a})}$$\implies$
$1<\frac{c_{r_{2}}}{Q(D_{a})}$%
{} completing the proof for (i).

------

For (ii) When $c_{1}=c_{2}=Q(D$) and $\frac{T(p_{1})}{E(T|D_{a})}<\frac{T(p_{2})}{E(T|D_{a})}$,
$M_{ar_{1}}=M_{ar_{2}}$, and need to find the value of $N_{r_{1}}$
that solves:

\[
\frac{N_{a}}{N_{a}+N_{r_{1}}}+\frac{T(p_{1})}{E(T|D_{a})}\frac{N_{r_{1}}}{N_{a}+N_{r_{1}}}=\frac{N_{a}}{N_{a}+N_{r_{2}}}+\frac{T(p_{2})}{E(T|D_{a})}\frac{N_{r_{2}}}{N_{a}+N_{r_{2}}}
\]

Rearranging, normalizing all $N$ values relative to $N_{a}=1$, simplifies
the equation to:

\[
1+N_{r_{1}}\frac{T(p_{1})}{E(T|D_{a})}=\frac{1+N_{r_{1}}}{1+N_{r_{2}}}\left(1+N_{r_{2}}\frac{T(p_{2})}{E(T|D_{a})}\right)
\]

Rearranging, 

\[
N_{r_{1}}=\frac{N_{r_{2}}\left(\frac{T(p_{2})}{E(T|D_{a})}-1\right)}{\frac{T(p_{1})}{E(T|D_{a})}-1-N_{r_{2}}\left(\frac{T(p_{2})}{E(T|D_{a})}-\frac{T(p_{1})}{E(T|D_{a})}\right)}
\]

\begin{equation}
N_{r_{1}}=\frac{N_{r_{2}}\left(\frac{T(p_{2})}{E(T|D_{a})}-1\right)}{\left(\frac{T(p_{1})}{E(T|D_{a})}-1\right)\left(1+N_{r_{2}}\right)-N_{r_{2}}\left(\frac{T(p_{2})}{E(T|D_{a})}-1\right)}\label{eq:Nr1-Threshold-p}
\end{equation}

Similar to the algebra for proving (i) rearranging \ref{eq:Nr1-Threshold-p}
yields, 

\[
N_{r_{1}}=\frac{\varsigma N_{r_{2}}}{1+N_{r_{2}}\left(1-\varsigma\right)}
\]

where $\varsigma=\frac{T(p_{2})-E(T|D_{a})}{T(p_{1})-E(T|D_{a})}$
and the additional constraint on $\frac{T(p_{1})}{E(T|D_{a})}$,

\[
\frac{1+N_{r_{2}}\frac{T(p_{2})}{E(T|D_{a})}}{1+N_{r_{2}}}<\frac{T(p_{1})}{E(T|D_{a})}
\]

\[
N_{r_{1}}<1\implies N_{r_{2}}<\frac{1}{2\varsigma-1}
\]

For proof of the Corollary, note that since $\frac{c_{1}}{Q(D_{a})}=1$
and $\frac{T(p_{1})}{E(T|D_{a})}=1$ the values for $\frac{c_{1}}{Q(D_{a})}$
and $\frac{T(p_{1})}{E(T|D_{a})}$ are relatively too small to meet
the minimum thresholds as in the inequalities $\frac{c_{1}}{Q(D_{a})}<\frac{1+N_{r_{2}}\left(2\frac{c_{2}}{Q(D_{a})}-1\right)}{1+N_{r_{2}}}$
and $\frac{T(p_{1})}{E(T|D_{a})}<\frac{1+N_{r_{2}}\left(2\frac{T(p_{2})}{E(T|D_{a})}-1\right)}{1+N_{r_{2}}}$
. 

Thus completing the proof. 
\end{proof}
The results in the proposition derive from the extra weighting attributable
to $N_{r_{1}}$ that allow a point with a lower $c_{1}$ or $p_{1}$,
and a large enough $N_{r_{1}}$to be preferred over another point
with higher values. However, there are limits on the values of $c_{1}$
and $p_{1}$, and this is the basis for the corollary. These lower
limits on $c_{1}$ result from the upper limit $N_{r_{1}}<1$. 

Note that the corollary could potentially be made less restrictive
considering the restrictions on $\frac{c_{1}}{Q(D_{a})}$ and $\frac{T(p_{1})}{E(T|D_{a})}$,
indicating that there could be values for $r_{1}$ that are not added
to the sequence before $r_{2}$ when

\[
1<\frac{c_{1}}{Q(D_{a})}<\frac{1+N_{r_{2}}\left(2\frac{c_{2}}{Q(D_{a})}-1\right)}{1+N_{r_{2}}}\text{ and }1<\frac{T(p_{1})}{E(T|D_{a})}<\frac{1+N_{r_{2}}\left(2\frac{T(p_{2})}{E(T|D_{a})}-1\right)}{1+N_{r_{2}}}
\]

For further reference, the following 4 Sequence Scenarios are defined
to differentiate assumptions for points $r_{1}$ and $r_{2}$:
\begin{description}
\item [{Sequence\ Scenario\ B1:}] For preferred distribution sequence
let $r_{1}$ be such that $D^{*}=D_{a}+I_{r_{1}}$ where $D_{a}$
is the distribution in the sequence prior to $D^{*}$. 
\item [{Sequence\ Scenario\ B2:}] The Sequence Scenario B1 with the addition
of a point $r_{2}$ with distribution $I_{r_{2}}$ where $D'=D^{*}+I_{r_{2}}$.
\item [{Sequence\ Scenario\ B3:}] The Sequence Scenario B2 where $r_{1}=\{c_{r_{1}},p_{r_{1}}\}$
such that $c_{r_{1}}=Q(D_{a})$ and $T(p_{r_{1}})=E(T|D_{a})$. In
this scenario $Q(D_{a})=Q(D_{ar_{1}})$ and $E(T|D_{a})=E(T|D_{ar_{1}})$
so there is no change in $M$ or the producer unit value $E(T,D)$
between $D_{a}$ and $D^{*}$. Substituting into equation \ref{eq:OrderCond-r1-r2}
shows that $\frac{N_{a}}{N_{a}+N_{r_{1}}}+\frac{T(p_{1})}{E(T|D_{a})}\frac{N_{r_{1}}}{N_{a}+N_{r_{1}}}=1$
regardless of the value for $N_{r_{1}}$.
\item [{Sequence\ Scenario\ B4:}] The Sequence Scenario B2 where $r_{1}=\{c_{r_{1}},p_{r_{1}}\}$
such that $c_{r_{1}}\leq Q(D_{a})$ and $T(p_{r_{1}})<E(T|D_{a})$.
In this scenario $Q(D_{ar_{1}})\leq Q(D_{a})$ and $E(T|D_{ar_{1}})<E(T|D_{a})$.
Substituting into equation \ref{eq:OrderCond-r1-r2} shows that $\frac{N_{a}}{N_{a}+N_{r_{1}}}+\frac{T(p_{r_{1}})}{E(T|D_{a})}\frac{N_{r_{1}}}{N_{a}+N_{r_{1}}}<1$
regardless of the value for $N_{r_{1}}$.
\end{description}
The Sequence Scenarios B1 - B4 provide a simplified way to assure
consistency of assumptions for the facts, propositions, and discussions
in the remainder of this appendix. Sequence Scenarios B1 and B2 are
basic descriptions. Scenario B3 provides a specialized case that while
not guaranteed to exist is plausible in consideration of Proposition
\ref{prop:For-OrderAdded}, and the restrictions are useful. Scenario
B4 is a less restrictive alternative to B3. 

For notational convenience let $M_{r_{2}}$ be the value of $M(D')$
where $D'=D^{*}+I_{r_{2}}$, and let $M_{ar_{2}}$ be the value of
$M(D')$ where $D'=D_{a}+I_{r_{2}}$. Specifically, 
\[
M_{r_{2}}(N_{r_{2}})=M\left(\frac{N^{*}}{N^{*}+N_{r_{2}}}Q(D^{*})+\frac{N_{r_{2}}}{N^{*}+N_{r_{2}}}c\right)
\]

\[
M_{ar_{2}}(N_{r_{2}})=M\left(\frac{N_{a}}{N_{a}+N_{r_{2}}}Q(D_{a})+\frac{N_{r_{2}}}{N_{a}+N_{r_{2}}}c\right)\text{ and \ensuremath{M_{a}}\ensuremath{=M(Q(\ensuremath{D_{a}}))}}
\]

\begin{prop}
\label{prop:For-a-point-added-after-D*} For the sequence scenario
B2,  $\frac{T(p_{r_{2}})}{E(T|D^{*})}>1$ $\iff$$N^{*}>M_{r_{2}},$
and $N^{*}<M_{r_{2}}$ $\iff$ $\frac{T(p_{r_{2}})}{E(T|D^{*})}<1$.
\end{prop}
\begin{proof}
Consider the case where both $\frac{T(p_{r_{2}})}{E(T|D^{*})}>1$
and $N^{*}<M_{r_{2}}$. In this case by the corollary of Proposition
\ref{prop:Preferred-When-Nr1>Nr2} point $r_{2}$ would be preferred
to any point $r_{1}$ with $\frac{c_{r_{1}}}{Q(D_{a})}=1$ and $\frac{T(p_{r_{1}})}{E(T|D_{a})}=1$,
and further distribution $I_{r_{2}}$ would be preferred to $D_{a}$,
and this would be regardless of the values for $N_{r_{1}}$ or $N_{a}$.
Thus only $\frac{T(p_{r_{2}})}{E(T|D^{*})}>1$ or $N^{*}<M_{r_{2}}$
are possible if $\frac{c_{r_{1}}}{Q(D_{a})}\leq1$ and $\frac{T(p_{r_{1}})}{E(T|D_{a})}\leq1$
and point $r_{1}$ is to be preferred over point $r_{2}$. Thus completing
the proof.
\end{proof}
The proposition says that a point $r_{2}$ with the properties specified,
would not be optimally added after $D^{*}$ if both $T(p_{r_{2}})>E(T|D^{*})$
and $M_{r_{2}}>N^{*}$. Thus only one or none of the two inequalities
can hold if the point $r_{2}$ is added after $D^{*}$. Also note
the requirement for strict inequality in the proposition and as shown
in the proof.

The proposition works because both $T(p_{r_{2}})>E(T|D^{*})$ and
$M_{r_{2}}>N^{*}$ implies $r_{2}$ would have been added to the sequence
prior to at least one of the points in $D^{*}$ or prior to $D_{a}$.\textbf{}
\begin{prop}
\label{prop:The-functions-Mr2-Mar2} For $r_{2}=\{c_{r_{2}},p_{r_{2}}\}$
and $N_{r_{2}}>0$ and assuming $M(q)=\zeta q^{\alpha}$,

(i) with sequence scenario B3 if $\frac{c_{r_{2}}}{Q(D_{a})}>1$,
$M_{ar_{2}}>M_{r_{2}}$ , and if $\frac{c_{r_{2}}}{Q(D_{a})}<1$,
$M_{ar_{2}}<M_{r_{2}}$, and $\lvert M_{a}-M_{ar_{2}}\rvert>\lvert N^{*}-M_{r_{2}}\rvert$,
and 

(ii) with sequence scenario B4 if $\frac{c_{r_{2}}}{Q(D_{a})}>\frac{c_{r_{1}}}{Q(D_{a})}-\frac{N_{a}}{N_{r_{2}}}\left(1-\frac{c_{r_{1}}}{Q(D_{a})}\right)$
or if $\frac{c_{r_{1}}}{Q(D_{a})}<\frac{\frac{c_{r_{2}}}{Q(D_{a})}N_{r_{2}}+N_{a}}{N_{a}+N_{r_{2}}}$
 then $M_{ar_{2}}>M_{r_{2}}$, and otherwise $M_{ar_{2}}<M_{r_{2}}$.
\end{prop}
\begin{proof}
Prove (i) by evaluating this inequality.

\begin{equation}
\frac{M_{ar_{2}}}{M_{r_{2}}}=\left(\frac{\frac{N_{a}}{N_{a}+N_{r_{2}}}Q(D_{a})+\frac{N_{r_{2}}}{N_{a}+N_{r_{2}}}c_{r_{2}}}{\frac{N^{*}}{N^{*}+N_{r_{2}}}Q(D^{*})+\frac{N_{r_{2}}}{N^{*}+N_{r_{2}}}c_{r_{2}}}\right)^{\alpha}>1\label{eq:Mar2-Mr2}
\end{equation}

Equivalently, evaluate,

\[
\frac{\frac{N_{a}}{N_{a}+N_{r_{2}}}Q(D_{a})+\frac{N_{r_{2}}}{N_{a}+N_{r_{2}}}c_{r_{2}}}{\frac{N^{*}}{N^{*}+N_{r_{2}}}Q(D^{*})+\frac{N_{r_{2}}}{N^{*}+N_{r_{2}}}c_{r_{2}}}>1
\]

Note $D^{*}=D_{a}+I_{r_{1}}$ and $Q(D^{*})=\frac{N_{a}}{N^{*}}Q(D_{a})+\frac{N_{r_{1}}}{N^{*}}c_{r_{1}}$
and under the assumptions of $B3$, $c_{r_{1}}=Q(D_{a})$. Substituting
yields,

\[
\frac{\frac{N_{a}}{N_{a}+N_{r_{2}}}Q(D_{a})+\frac{N_{r_{2}}}{N_{a}+N_{r_{2}}}c_{r_{2}}}{\frac{N_{a}+N_{r_{1}}}{N_{a}+N_{r_{1}}+N_{r_{2}}}Q(D_{a})+\frac{N_{r_{2}}}{N_{a}+N_{r_{1}}+N_{r_{2}}}c_{r_{2}}}>1
\]

Thus in the B3 case, because the weights shift more to $Q(D_{a})$
in the denominator,%
{} the denominator will be smaller when $c_{r_{2}}>Q(D_{a}$), thus
proving:%

\[
c_{r_{2}}>Q(D_{a})\implies M_{ar_{2}}>M_{r_{2}}
\]

and likewise,

\[
c_{r_{2}}<Q(D_{a})\implies M_{ar_{2}}<M_{r_{2}}
\]

Finally, noting that $N^{*}=M_{ar_{1}}=M_{a}$, and When $c_{r_{2}}>Q(D_{a})$,
$M_{ar_{2}}-M_{a}>M_{r_{2}}-M_{a}$ and when $c_{r_{2}}<Q(D_{a})$,
$M_{a}-M_{ar_{2}}>M_{a}-M_{r_{2}}$ thus $\lvert M_{a}-M_{ar_{2}}\rvert>\lvert N^{*}-M_{r_{2}}\rvert$

Thus proving (i). 

Expanding to include $c_{r_{1}}$ under the assumptions of B4, $c_{r_{1}}<Q(D_{a})$,
so $0<\frac{c_{r_{1}}}{Q(D_{a})}<1$ and the inequality to evaluate
in this more general case is:

\[
\left(\frac{N_{a}+N_{r_{1}}+N_{r_{2}}}{N_{a}+N_{r_{2}}}\right)\frac{N_{a}Q(D_{a})+N_{r_{2}}c_{r_{2}}}{\left(N_{a}+\frac{c_{r_{1}}}{Q(D_{a})}N_{r_{1}}\right)Q(D_{a})+N_{r_{2}}c_{r_{2}}}>1
\]

Where $\frac{c_{r_{1}}}{Q(D_{a})}=1$ corresponds with the B3 assumptions.
Rearranging and solving for $\frac{c_{r_{2}}}{Q(D_{a})}$, %

\[
\frac{N_{a}+N_{r_{1}}+N_{r_{2}}}{N_{a}+N_{r_{2}}}N_{a}+\frac{N_{a}+N_{r_{1}}+N_{r_{2}}}{N_{a}+N_{r_{2}}}N_{r_{2}}\frac{c_{r_{2}}}{Q(D_{a})}>N_{a}+\frac{c_{r_{1}}}{Q(D_{a})}N_{r_{1}}+N_{r_{2}}\frac{c_{r_{2}}}{Q(D_{a})}
\]

Further rearranging yields,

\[
\frac{1}{N_{a}+N_{r_{2}}}N_{a}+\frac{1}{N_{a}+N_{r_{2}}}N_{r_{2}}\frac{c_{r_{2}}}{Q(D_{a})}>\frac{c_{r_{1}}}{Q(D_{a})}
\]

Thus, $M_{ar_{2}}>M_{r_{2}}$ when,

\begin{equation}
\frac{c_{r_{2}}}{Q(D_{a})}>\left(\frac{N_{a}}{N_{r_{2}}}+1\right)\frac{c_{r_{1}}}{Q(D_{a})}-\frac{N_{a}}{N_{r_{2}}}\label{eq:Mar2>Mr2-limit}
\end{equation}

Alternatively, 

\begin{equation}
\frac{c_{r_{1}}}{Q(D_{a})}<\frac{\frac{c_{r_{2}}}{Q(D_{a})}N_{r_{2}}+N_{a}}{N_{a}+N_{r_{2}}}\label{eq:Mar2>Mr2-limit-on-cr1}
\end{equation}

Thus it will be the case that $M_{ar_{2}}>M_{r_{2}}$ when the threshold
of either inequality \ref{eq:Mar2>Mr2-limit} or \ref{eq:Mar2>Mr2-limit-on-cr1}
holds. Thus proving (ii).
\end{proof}
Intuitively, while $M$ depends on $q$ the shift from $D_{a}$ to
$D^{*}$ changes the weights from the $N$ values ($N_{a}$ vs $N^{*})$
when adding the incremental point $r_{2}$ leading to differences
in $M_{ar_{2}}$ and $M_{r_{2}}$ even though $M_{ar_{1}}=M_{a}$.

Note in general $N_{r_{2}}$is much less than $N_{a}$ so the threshold
from inequality \ref{eq:Mar2>Mr2-limit-on-cr1} is likely to hold
in most cases. For example, consider $N_{r_{2}}\in[0,N_{a}]$, generally
the practical range for $N_{r_{2}}$. In this range, the limit with
$\frac{c_{r_{2}}}{Q(D_{a})}=0$ ranges from $[1,0.5]$ and with $\frac{c_{r_{2}}}{Q(D_{a})}=0.5$
the range is $[1,0.75]$.%

\section*{Sole Producer Value Function}

The sole producer function $S(D)=E(T|D)min(M(D),N(D))$ corresponds
to the media source value function $V(D)$ for $N(D')>M(D'$ and is
simply $E(T|D)N(D)$ for $N(D')\leq M(D')$. It will be of interest
to know when a step from $D$ to $D'$ is a positive change, $\Delta_{S}>0$,
or a negative, $\Delta_{S}<0$, for the sole producer. This is analyzed
for cases where $N(D')\leq M(D')$ and where $N(D')>M(D')$.

\subsubsection*{The case when $N(D')\protect\leq M(D')$}

Noting that $E(T|D')=E(T|D)\frac{N}{N'}+T(p)\frac{N'-N}{N'}$, the
change in $S(D)$ for $N(D)<N^{*}$ is: 

\begin{equation}
\begin{aligned}\Delta_{S}=S(D')-S(D)= & E(T|D')N'-E(T|D)N=\\
 & \left(E(T|D)\frac{N}{N'}+T(p)\frac{N'-N}{N'}\right)N'-E(T|D)N=T(p)(N'-N)
\end{aligned}
\label{eq:SP-delta}
\end{equation}

Thus $\Delta_{S}>0$ in this case when $N$ is increasing provided
that $T(p)>0$.%

For the more general case where the change in the distribution involves
multiple points, 

\begin{equation}
\Delta_{S}=H(T|Y)(N'-N)\label{eq:SP-Delta-Region}
\end{equation}

This more general case shows that as long as the net producer value
over the distribution $Y$ is positive, that is $H(T|Y)>0$, then
$\Delta_{S}>0$ and there is a net gain to the sole producer even
if some of the points in the distribution $Y$ have $p<0$.%

\begin{prop}
\label{prop:+DelS-before-D*}For a shift from $D$ to $D'$ where
$N(D')\leq M(D')$ and $D'=D+Y$ , $\Delta_{S}>0$ implies that $H(T|Y)>0$. 
\end{prop}
\begin{proof}
This follows from equations \ref{eq:SP-delta} and \ref{eq:SP-Delta-Region}
and their derivations above. 
\end{proof}

\subsubsection*{The case when \textmd{$N(D')>M(D')$ }}

The condition needed for $\Delta_{S}>0$ is,

\[
\Delta_{S}=\Delta_{V}=E(T|D)N(D)\left(\frac{M(D')}{N(D')}-\frac{M(D)}{N(D)}\right)+T(p)\left(N(D')-N(D)\right)\frac{M(D')}{N(D')}>0
\]

Rearranging,%

\begin{equation}
\left(\frac{N(D)}{N(D')}+\frac{T(p)}{E(T|D)}\frac{N(D')-N(D)}{N(D')}\right)\frac{M(D')}{M(D)}>1\label{eq:ConditionForS+}
\end{equation}

Since $\frac{N(D)}{N(D')}$ and $\frac{N(D')-N(D)}{N(D')}$ are normalized
weights summing to $1$, this term: 

\[
\left(\frac{N(D)}{N(D')}+\frac{T(p)}{E(T|D)}\frac{N(D')-N(D)}{N(D')}\right)
\]
 will be $<1$ when $\frac{T(p)}{E(T|D)}<1$. This suggests a simple
$S.C.$ for $\Delta_{S}<0$ when $N(D')\geq N^{*}$.%

\begin{prop}
\label{prop:A-SC-for-No-Gain} When \textup{$N(D')>M(D')$, a $S.C.$
for} $\Delta_{S}<0$ from adding an incremental point $r=\{c,p\}$
such that $D'=D+I_{r}$ is: 
\begin{equation}
\frac{T(p)}{E(T|D)}\leq1\text{ and }\frac{M(D')}{M(D)}\leq1\label{eq:SimpleSCondForS-}
\end{equation}

Conversely a $N.C.$ for $\Delta_{S}>0$ is:

\begin{equation}
\frac{T(p)}{E(T|D)}>1\text{ or }\frac{M(D')}{M(D)}>1
\end{equation}
\end{prop}
\begin{proof}
See the derivations above in this subsection.
\end{proof}
Alternatively, the condition on $\frac{T(p)}{E(T|D)}$ for $\Delta_{S}>0$
is,

\begin{equation}
\frac{T(p)}{E(T|D)}>\left(\frac{M(D)}{M(D')}-\frac{N(D)}{N(D')}\right)\frac{N(D')}{N(D')-N(D)}\label{eq:T(p)CondForS+}
\end{equation}

Consider the case where $D=D^{*}$ so $M(D)=N(D)$. Substituting this
into the $R.H.S.$ of equation \ref{eq:T(p)CondForS+} and rearranging%
{} yields:

\[
\left(\frac{M(D^{*})}{M(D')}\right)\frac{N(D')-M(D')}{N(D')-M(D^{*})}>1
\]

When $M(D)>M(D')$ the $>$ inequality holds because, $N(D')-M(D')>N(D')-M(D)$
$\iff$ $M(D)>M(D')$. Thus $M(D^{*})>M(D')$ implies the condition
$\frac{T(p)}{E(T|D^{*})}>1$ for $\Delta_{S}<0$.
\begin{prop}
\label{prop:A-SC-For-Gain}At $D^{*}$ where $M(D^{*})=N(D^{*})$
when adding an incremental point $r=\{c,p\}$ such that $D'=D^{*}+I_{r}$,
a $S.C.$ for $\Delta_{S}>0$ is:

\[
\frac{T(p)}{E(T|D^{*})}>\left(\frac{M(D^{*})}{M(D')}-\frac{N(D^{*})}{N(D')}\right)\frac{N(D')}{N(D')-N(D^{*})}=\left(\frac{M(D^{*})}{M(D')}\right)\frac{N(D')-M(D')}{N(D')-M(D^{*})}
\]

and the $N.C.$ for $\left(\frac{M(D^{*})}{M(D')}-\frac{N(D^{*})}{N(D')}\right)\frac{N(D')}{N(D')-N(D^{*})}>1\text{ is for }M(D')<M(D^{*})$. 

Alternatively, $M(D')>M(D^{*})$ $\implies$ $\frac{T(p)}{E(T|D^{*})}<1$
and a $S.C.$ for $\Delta_{S}>0$ as,
\begin{equation}
\frac{M(D')}{M(D^{*})}>\frac{1}{\left(\frac{N(D^{*})}{N(D')}+\frac{T(p)}{E(T|D^{*})}\frac{N(D')-N(D^{*})}{N(D')}\right)}>1\label{eq:M+CondForS+}
\end{equation}
\end{prop}
\begin{proof}
See the derivations above in this subsection.
\end{proof}
The proposition provides thresholds at $D^{*}$ for increased sole
producer value, $\Delta_{S}>0$ beyond $D^{*}$. The threshold on
$\frac{T(p)}{E(T|D^{*})}$ from equation \ref{eq:T(p)CondForS+} is
$>1$. The alternative threshold on $\frac{M(D')}{M(D^{*})}$ from
equation \ref{eq:M+CondForS+} is also $>1$.

Proposition \ref{prop:A-SC-For-Gain} partially reinforces the conclusion
from Proposition \ref{prop:For-a-point-added-after-D*} that $\frac{T(p_{r_{2}})}{E(T|D^{*})}>1$
implies $N^{*}>M(D')$. However, here the condition relates to a constraint
on $\frac{T(p)}{E(T|D^{*})}$ for a positive $\Delta_{S}$ while in
Proposition \ref{prop:For-a-point-added-after-D*} the conditions
relate to the order that the points are added. As it turns out combining
the constraints on order and those on $\Delta_{S}>0$ provides for
an interesting conclusion. 

\section*{Consideration of Prior Order}

In this section the order of which points would be added to the sequence
of distributions is considered, and in particular if the point $r_{2}$
in Sequence Scenarios B3 and B4 would be a viable point to evaluate
or would it be a point that would have optimally already been added
in the sequence prior to $D^{*}$. In view of proposition \ref{prop:For-a-point-added-after-D*}
the two cases that need to be considered are when $M(D')<M(D^{*})\text{ and }T(p_{r_{2}})>E(T|D^{*})$
and when $M(D')>M(D^{*})\text{ and }T(p_{r_{2}})<E(T|D^{*})$.

\subsection*{\emph{The case where $M(D')<M(D^{*})\text{ and }T(p_{r_{2}})>E(T|D^{*})$}}

Rearranging equation \ref{eq:ConditionForS+} and presenting in $N^{*}$
and $N_{r_{2}}$ notation:

\[
\left(\frac{N^{*}}{N^{*}+N_{r_{2}}}+\frac{T(p_{2})}{E(T|D^{*})}\frac{N_{r_{2}}}{N^{*}+N_{r_{2}}}\right)\frac{M_{r_{2}}}{N^{*}}>1
\]

Utilizing the revised notation for equation \ref{eq:T(p)CondForS+},%

\begin{equation}
F_{low}(N_{r_{2}})=\frac{N^{*}}{N_{r_{2}}}\left(\frac{N^{*}}{M_{r_{2}}}-1\right)+\frac{N^{*}}{M_{r_{2}}}<\frac{T(p_{2})}{E(T|D^{*})}\label{eq:Tp-Threshold}
\end{equation}

Where $F_{low}(N_{r_{2}})$ is the lower bounds on $\frac{T(p_{2})}{E(T|D^{*})}$
for $\Delta_{S}>0$ as a function of $N_{r_{2}}$. There are two effects
from an increase in $N_{r_{2}}$. First there is the impact directly
from $\frac{N^{*}}{N_{r_{2}}}$ where the increase in $N_{r_{2}}$
lowers $\frac{N^{*}}{N_{r_{2}}}$ and hence $F_{low}$. The second
increases $F_{low}$ via the impact on $M_{r_{2}}(N_{r_{2}})$, where
$\frac{N^{*}}{M_{r_{2}}}$ is increasing at a decreasing rate starting
with $\frac{N^{*}}{M_{r_{2}}}=1$ and ultimately converging to an
asymptote where $\frac{N^{*}}{M_{r_{2}}(c_{r_{2}})}>1$. Note from
Proposition \ref{prop:For-a-point-added-after-D*} it cannot be that
both $\frac{T(p_{2})}{E(T|D^{*})}>1\text{ and }\frac{N^{*}}{M_{r_{2}}}<1$.

Next consider under what cases $r_{2}$ would have been optimally
added to a distribution in the sequence $\Theta$ prior to $D^{*}$
as in Sequence Scenario B2. Substituting the points $r_{1}$ and $r_{2}$
into equation \ref{eq:OrderCond-r1-r2}, and noting that $N_{a}=N^{*}-N_{r_{1}}$
and $M_{ar_{1}}=M(D^{*})=N^{*}$, yields: %

\begin{equation}
\frac{T(p_{2})}{E(T|D_{a})}<(N^{*}-N_{r_{1}}+N_{r_{2}})\left[\frac{N^{*}-N_{r_{1}}}{N^{*}}+\frac{T(p_{1})}{E(T|D_{a})}\frac{N_{r_{1}}}{N^{*}}\right]\frac{N^{*}}{M_{ar_{2}}}-\frac{N^{*}-N_{r_{1}}}{N_{r_{2}}}=F_{up}(N_{r_{2}})\label{eq:rev-OrderCond-r1-r2}
\end{equation}

Where $F_{up}(N_{r_{2}})$ is the upper bounds for $\frac{T(p_{2})}{E(T|D)}$
to be viable.%

As shown in the following Proposition the lower bounds $F_{low}$
exceeds the upper bounds $F_{up}$ over the practical range for $N_{r_{2}}$.
\begin{prop}
\label{Prop:F_low-gt-F_up}\textbf{$F_{low}(N_{r_{2}})>F_{up}(N_{r_{2}})$}
for at least $N_{r_{2}}$ and $N_{r_{1}}$ such that $0<N_{r_{2}}<1$
and $0<N_{r_{1}}$, and 
\[
\frac{\left(1+\frac{c_{r_{2}}}{Q(D_{a})}N_{r_{2}}\right)\left(1-N_{r_{2}}\right)}{1-N_{r_{2}}\left(1-\frac{c_{r_{1}}}{Q(D_{a})}\right)}>N_{r_{1}}
\]
\end{prop}
\begin{proof}
Assume Sequence Scenario B2 and show by algebraic derivation:

\[
F_{low}(N_{r_{2}})=\frac{N^{*}}{N_{r_{2}}}\left(\frac{N^{*}}{M_{r_{2}}}-1\right)+\frac{N^{*}}{M_{r_{2}}}>(N^{*}-N_{r_{1}}+N_{r_{2}})\left[\frac{N^{*}-N_{r_{1}}}{N^{*}}+\frac{T(p_{1})}{E(T|D_{a})}\frac{N_{r_{1}}}{N^{*}}\right]\frac{N^{*}}{M_{ar_{2}}}-\frac{N^{*}-N_{r_{1}}}{N_{r_{2}}}=F_{up}(N_{r_{2}})
\]

Remove $\frac{N^{*}}{N_{r_{2}}}$ on both sides

\[
\frac{N^{*}}{N_{r_{2}}}\frac{N^{*}}{M_{r_{2}}}+\frac{N^{*}}{M_{r_{2}}}>(N^{*}-N_{r_{1}}+N_{r_{2}})\left[\frac{N^{*}-N_{r_{1}}}{N^{*}}+\frac{T(p_{1})}{E(T|D_{a})}\frac{N_{r_{1}}}{N^{*}}\right]\frac{N^{*}}{M_{ar_{2}}}+\frac{N_{r_{1}}}{N_{r_{2}}}
\]

Set $N^{*}=1$ as the base unit so all values of $N$ and $M$ are
relative to $N^{*}$, then reduce,%

\[
\frac{1-M_{r_{2}}N_{r_{1}}+N_{r_{2}}}{(1-N_{r_{1}}+N_{r_{2}})N_{r_{2}}}\frac{M_{ar_{2}}}{M_{r_{2}}}>1-N_{r_{1}}+\frac{T(p_{1})}{E(T|D_{a})}N_{r_{1}}
\]

Analyzing if $\frac{1-M_{r_{2}}N_{r_{1}}+N_{r_{2}}}{(1-N_{r_{1}}+N_{r_{2}})N_{r_{2}}}>1$
yields,%

\[
1-M_{r_{2}}N_{r_{1}}>(N_{r_{2}}-N_{r_{1}})N_{r_{2}}
\]

The inequality holds at least for all $N_{r_{2}}<1$. To confirm note
that relative to $N^{*}=1$, $1>N_{r_{1}}$ and $1>M_{r_{2}}$(from
Proposition \ref{prop:For-a-point-added-after-D*} and $\frac{T(p_{2})}{E(T|D^{*})}>1$)
implying $-M_{r_{2}}N_{r_{1}}\geq-N_{r_{1}}$ so for $N_{r_{2}}-N_{r_{1}}\geq0$
$1-M_{r_{2}}N_{r_{1}}>(N_{r_{2}}-N_{r_{1}})\geq(N_{r_{2}}-N_{r_{1}})N_{r_{2}}$
and if $N_{r_{2}}-N_{r_{1}}<0$ the inequality also holds. 

Assuming under Sequence Scenario B3 or B4 $\frac{T(p_{1})}{E(T|D_{a})}\leq1$
thus $1-N_{r_{1}}+\frac{T(p_{1})}{E(T|D_{a})}N_{r_{1}}\leq1$. Thus
if $\frac{M_{ar_{2}}}{M_{r_{2}}}>1$ per conditions of Proposition
\ref{prop:The-functions-Mr2-Mar2} the inequality holds and proving
that $F_{low}(N_{r_{2}})>F_{up}(N_{r_{2}})$ for those conditions. 

In addition, even if $\frac{M_{ar_{2}}}{M_{r_{2}}}<1$, implying that
that $c_{r_{2}}<Q(D_{a})$, assuming $M$ is at least approximated
by $M(q)=\zeta q^{\alpha}$ with $0<\alpha<1$, it is still possible
to show that $F_{low}(N_{r_{2}})>F_{up}(N_{r_{2}})$ by showing the
inequality from substituting $N_{a}+N_{r_{1}}=N^{*}=1$ into \ref{eq:Mar2-Mr2}
holds, 

\[
1>\frac{M_{ar_{2}}}{M_{r_{2}}}=\left(\left(\frac{1+N_{r_{2}}}{1-N_{r_{1}}+N_{r_{2}}}\right)\frac{1-N_{r_{1}}+N_{r_{2}}\frac{c_{r_{2}}}{Q(D_{a})}}{\left(1-N_{r_{1}}+\frac{c_{r_{1}}}{Q(D_{a})}N_{r_{1}}\right)+N_{r_{2}}\frac{c_{r_{2}}}{Q(D_{a})}}\right)^{\alpha}>\frac{1+N_{r_{1}}\left(\frac{T(p_{1})}{E(T|D_{a})}-1\right)}{\frac{1-M_{r_{2}}N_{r_{1}}+N_{r_{2}}}{(1-N_{r_{1}}+N_{r_{2}})N_{r_{2}}}}
\]

Considering that $\frac{M_{ar_{2}}}{M_{r_{2}}}<1$ and $0<\alpha<1$
a more restrictive inequality is:

\begin{align*}
1>\left(\left(\frac{1+N_{r_{2}}}{1-N_{r_{1}}+N_{r_{2}}}\right)\frac{1-N_{r_{1}}+N_{r_{2}}\frac{c_{r_{2}}}{Q(D_{a})}}{\left(1-N_{r_{1}}+\frac{c_{r_{1}}}{Q(D_{a})}N_{r_{1}}\right)+N_{r_{2}}\frac{c_{r_{2}}}{Q(D_{a})}}\right)^{\alpha} & >\\
\left(\left(\frac{1+N_{r_{2}}}{1-N_{r_{1}}+N_{r_{2}}}\right)\frac{1-N_{r_{1}}+N_{r_{2}}\frac{c_{r_{2}}}{Q(D_{a})}}{\left(1-N_{r_{1}}+\frac{c_{r_{1}}}{Q(D_{a})}N_{r_{1}}\right)+N_{r_{2}}\frac{c_{r_{2}}}{Q(D_{a})}}\right) & >\left(\frac{1+N_{r_{1}}\left(\frac{T(p_{1})}{E(T|D_{a})}-1\right)}{\frac{1-M_{r_{2}}N_{r_{1}}+N_{r_{2}}}{(1-N_{r_{1}}+N_{r_{2}})N_{r_{2}}}}\right)
\end{align*}

Thus this more restricted inequality would be a $S.C.$ :

\begin{multline}
1>\left(\frac{1+N_{r_{2}}}{1-N_{r_{1}}+N_{r_{2}}}\right)\frac{1-N_{r_{1}}+N_{r_{2}}\frac{c_{r_{2}}}{Q(D_{a})}}{1-N_{r_{1}}+\frac{c_{r_{1}}}{Q(D_{a})}N_{r_{1}}+N_{r_{2}}\frac{c_{r_{2}}}{Q(D_{a})}}>\\
\left(1+N_{r_{1}}\left(\frac{T(p_{1})}{E(T|D_{a})}-1\right)\right)\frac{1-N_{r_{1}}+N_{r_{2}}}{1-M_{r_{2}}N_{r_{1}}+N_{r_{2}}}N_{r_{2}}\label{eq:SC-for-Flow-gt-Fup}
\end{multline}

Rearranging,

\[
\left[\frac{1+N_{r_{2}}}{1-N_{r_{1}}+N_{r_{2}}}\right]\left[\frac{1-N_{r_{1}}+N_{r_{2}}\frac{c_{r_{2}}}{Q(D_{a})}}{N_{r_{2}}\left(1-N_{r_{1}}+N_{r_{2}}\frac{c_{r_{2}}}{Q(D_{a})}+\frac{c_{r_{1}}}{Q(D_{a})}N_{r_{1}}\right)}\right]>\left[1+N_{r_{1}}\left(\frac{T(p_{1})}{E(T|D_{a})}-1\right)\right]\left[\frac{1-N_{r_{1}}+N_{r_{2}}}{1-M_{r_{2}}N_{r_{1}}+N_{r_{2}}}\right]
\]

The 2 terms on the $R.H.S.$of the inequality are $<1$ thus the $R.H.S<1$.
The first term on the $L.H.S.$ is $>1$. So if the second term on
the left is $>1$ so the proof is complete. 

Checking,

\[
x=\left(1-N_{r_{1}}+N_{r_{2}}\frac{c_{r_{2}}}{Q(D_{a})}\right)>N_{r_{2}}\left(1-N_{r_{1}}+N_{r_{2}}\frac{c_{r_{2}}}{Q(D_{a})}+\frac{c_{r_{1}}}{Q(D_{a})}N_{r_{1}}\right)
\]

Where $x$ is a temporary variable for notational convenience. Representing
the inequality using $x$ yields, 

\[
x>N_{r_{2}}x+N_{r_{2}}\frac{c_{r_{1}}}{Q(D_{a})}N_{r_{1}}
\]

\[
x\left(1-N_{r_{2}}\right)>N_{r_{1}}N_{r_{2}}\frac{c_{r_{1}}}{Q(D_{a})}
\]

Assume $\frac{c_{r_{2}}}{Q(D_{a})}<1$, then expand $x$ in the inequality:

\[
\left(1-N_{r_{1}}\right)\left(1-N_{r_{2}}\right)+\frac{c_{r_{2}}}{Q(D_{a})}N_{r_{2}}\left(1-N_{r_{2}}\right)>\frac{c_{r_{1}}}{Q(D_{a})}N_{r_{1}}N_{r_{2}}
\]

Simplified

\[
1+N_{r_{1}}N_{r_{2}}-N_{r_{1}}-N_{r_{2}}+\frac{c_{r_{2}}}{Q(D_{a})}N_{r_{2}}\left(1-N_{r_{2}}\right)>\frac{c_{r_{1}}}{Q(D_{a})}N_{r_{1}}N_{r_{2}}
\]

\[
1+N_{r_{1}}\left[N_{r_{2}}\left(1-\frac{c_{r_{1}}}{Q(D_{a})}\right)-1\right]-N_{r_{2}}+\frac{c_{r_{2}}}{Q(D_{a})}N_{r_{2}}\left(1-N_{r_{2}}\right)>0
\]

\[
1-N_{r_{2}}+\frac{c_{r_{2}}}{Q(D_{a})}N_{r_{2}}\left(1-N_{r_{2}}\right)>N_{r_{1}}\left[1-N_{r_{2}}\left(1-\frac{c_{r_{1}}}{Q(D_{a})}\right)\right]
\]

Thus a more general less restrictive $S.C.$ for $F_{low}(N_{r_{2}})>F_{up}(N_{r_{2}})$ 

\[
\frac{\left(1+\frac{c_{r_{2}}}{Q(D_{a})}N_{r_{2}}\right)\left(1-N_{r_{2}}\right)}{1-N_{r_{2}}\left(1-\frac{c_{r_{1}}}{Q(D_{a})}\right)}>N_{r_{1}}
\]

Completing the proof.
\end{proof}
For intuition, set $\frac{c_{r_{1}}}{Q(D_{a})}=1$ the highest value
it could have and set $\frac{c_{r_{2}}}{Q(D_{a})}=0$ as the reasonable
lowest value possible. With $\frac{c_{r_{1}}}{Q(D_{a})}=1$

\[
\left(1-N_{r_{2}}\right)\left(1+\frac{c_{r_{2}}}{Q(D_{a})}N_{r_{2}}\right)>N_{r_{1}}
\]

And with $\frac{c_{r_{2}}}{Q(D_{a})}=0$ ,

\[
1-N_{r_{2}}>N_{r_{1}}
\]

Thus an alternative simple $S.C.$ would be $N_{r_{2}}<\frac{1}{2}N^{*}$and
$N_{r_{1}}<\frac{1}{2}N^{*}$. 

 The graph of \textbf{$F_{low}(N_{r_{2}})\text{ vs }F_{up}(N_{r_{2}})$}
provides a further check that\textbf{ $F_{low}(N_{r_{2}})>F_{up}(N_{r_{2}})$}
at least in the range where $1>N_{r_{2}}+N_{r_{1}}$. The conclusion
that\textbf{ $F_{low}(N_{r_{2}})>F_{up}(N_{r_{2}})$} implies that
any $r_{2}$ that would have $\Delta_{S}>0$ would be preferred over
any $r_{1}$ where $1\geq\frac{T(p_{r_{1}})}{E(T|D_{a})}$, and there
must be at least one such point $r_{1}$ by the definition of $E(T|D_{a})$. 

\noindent \includegraphics[scale=0.75]{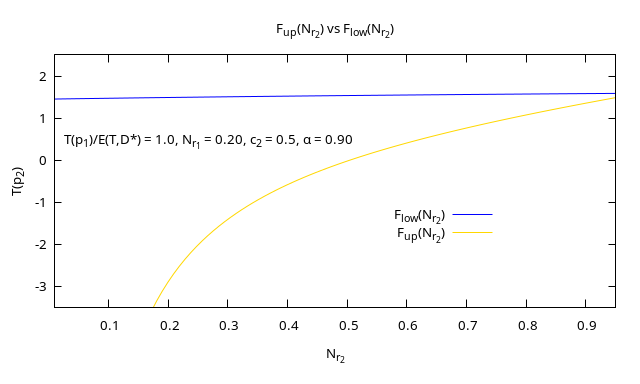}

\subsection*{\emph{The case where $M(D')>M(D^{*})\text{ and }T(p_{r_{2}})<E(T|D^{*})$}}

Let $\kappa$ represent the slope of $M$ over an interval between
two values of $N$. In this section where $M(D')>M(D^{*})$, $\kappa>0$
by definition. Further let $\kappa_{ar_{2}}$ be the slope for $M$
over the interval of $N_{a}$ to $N_{a}+N_{r_{2}}$, and assume that
$M_{a}=M_{ar_{1}}=N^{*}$ as in Sequence Scenario B3. Thus $M_{ar_{2}}$
is parameterized as, 
\begin{equation}
M_{ar_{2}}=N^{*}+\kappa_{ar2}N_{r_{2}}\label{eq:Mar2-parameterized}
\end{equation}
 Likewise for $M_{r_{2}}$, since $M_{ar_{1}}=N^{*}$,

\begin{equation}
M_{r_{2}}=N^{*}+\kappa_{r2}N_{r_{2}}\label{eq:Mr2-Parameterized}
\end{equation}

For the case where $0<\kappa<1$, $M(D')\leq N(D')$ and the conditions
for $\Delta_{S}>0$ can be determined by substituting $M(D^{*})=N(D^{*})=N^{*}$,
$N(D')=N^{*}+N_{r_{2}}$ and $M(D')=M_{r_{2}}$ into equation \ref{eq:M+CondForS+}
from Proposition \ref{prop:A-SC-For-Gain} to get,%

\begin{equation}
X_{l}=\frac{N^{*}+N_{r_{2}}}{N^{*}+\frac{T(p_{2})}{E(T|D^{*})}N_{r_{2}}}<\frac{M_{r_{2}}}{N^{*}}\label{eq:Xl-condition}
\end{equation}

For the case where $\kappa\geq1$, $M(D')\geq N(D')$ and as in Proposition
\ref{prop:+DelS-before-D*} the change in the sole producer value
is: 

\[
\Delta_{S}=E(T|D')N'-E(T|D^{*})N^{*}=E(T|D')\left(N^{*}+N_{r_{2}}\right)-E(T|D^{*})N^{*}
\]

Substituting $E(T|D')=E(T|D^{*})\frac{1}{1+N_{r_{2}}}+T(p_{2})\frac{N_{r_{2}}}{1+N_{r_{2}}}$
with $N^{*}=1$ and $N'=1+N_{r_{2}}$ yields,%

\begin{equation}
\Delta_{S}=T(p_{2})N_{r_{2}}\label{eq:DeltaS=00003DTp2Nr2}
\end{equation}

Thus, in this case, $\Delta_{S}>0$ when $T(p_{2})>0$.

When $\kappa>0$, from equation \ref{eq:OrderCond-r1-r2} the order
requirement for $r_{1}$ before $r_{2}$ is,

\[
\left[\frac{N_{a}}{N_{a}+N_{r_{1}}}+\frac{T(p_{1})}{E(T|D_{a})}\frac{N_{r_{1}}}{N_{a}+N_{r_{1}}}\right]M_{ar_{1}}>\left[\frac{N_{a}}{N_{a}+N_{r_{2}}}+\frac{T(p_{2})}{E(T|D_{a})}\frac{N_{r_{2}}}{N_{a}+N_{r_{2}}}\right]M_{ar_{2}}
\]

From the assumptions of Sequence Scenario B3, $M_{ar_{1}}=N^{*}=N_{a}+N_{r_{1}}$,
$E(T|D_{a})=E(T|D^{*})$%
, and $\frac{T(p_{r_{1}})}{E(T|D_{a})}=1$. Substituting those values
into the inequality above yields,

\begin{equation}
M_{ar_{2}}<X_{u}=\frac{N^{*}}{\frac{N^{*}-N_{r_{1}}}{N^{*}-N_{r_{1}}+N_{r_{2}}}+\frac{T(p_{2})}{E(T|D^{*})}\frac{N_{r_{2}}}{N^{*}-N_{r_{1}}+N_{r_{2}}}}\label{eq:Xu-condition}
\end{equation}

Since from%
{} Proposition \ref{prop:The-functions-Mr2-Mar2} $M_{r_{2}}<M_{ar_{2}}$,
when $0<\kappa<1$ and $M(D')\leq N(D')$, for there to be a viable
$\Delta_{S}>0$, this inequality chain would need to hold. 

\begin{equation}
1\leq X_{l}<M_{r_{2}}<M_{ar_{2}}<X_{u}\label{eq:viable-inequality-kappa-lt-1}
\end{equation}
If $X_{u}<X_{l}$ or $X_{u}<M_{ar_{2}}$ or $X_{l}>M_{r_{2}}$ the
inequality chain is broken, and there will be no viable $\Delta_{S}>0$
as it implies that the conditions \ref{eq:Xl-condition} or \ref{eq:Xu-condition}
are violated.%

Alternatively, when $\kappa\geq1$ and $M(D')\geq N(D')$ the inequality
chain for viability is:

\begin{equation}
1+N_{r_{2}}<M_{r_{2}}<M_{ar_{2}}<X_{u}\label{eq:eq:viable-inequality-kappa-gt-1}
\end{equation}

Here if $X_{u}<M_{ar_{2}}$ the inequality will be broken and $r_{2}$
will not be a viable addition to the sequence after $D^{*}$. 

The equality $X_{u}<X_{l}$ provides a $S.C.$ for non viability when
$\kappa<1$ and is analyzed first using the next 2 Facts. 
\begin{fact}
\label{fact:UnderB3-Xl-lt-Xu}Under the assumptions of Sequence Scenario
B3 where $\frac{T(p_{1})}{E(T|D_{a})}=1$, $X_{l}<X_{u}$ for $0\leq N_{r_{2}}\leq1=N^{*}$.

\[
1\leq X_{l}=\frac{N^{*}}{\frac{N^{*}}{N^{*}+N_{r_{2}}}+\frac{T(p_{2})}{E(T|D^{*})}\frac{N_{r_{2}}}{N^{*}+N_{r_{2}}}}<\frac{N^{*}}{\frac{N^{*}-N_{r_{1}}}{N^{*}-N_{r_{1}}+N_{r_{2}}}+\frac{T(p_{2})}{E(T|D^{*})}\frac{N_{r_{2}}}{N^{*}-N_{r_{1}}+N_{r_{2}}}}=X_{u}
\]
\end{fact}
\begin{proof}
From Proposition \ref{prop:For-a-point-added-after-D*} $\frac{T(p_{2})}{E(T|D^{*})}<1$,
so $X_{l}>1$, thus,%

\[
\frac{N^{*}-N_{r_{1}}}{N^{*}-N_{r_{1}}+N_{r_{2}}}+\frac{T(p_{2})}{E(T|D^{*})}\frac{N_{r_{2}}}{N^{*}-N_{r_{1}}+N_{r_{2}}}<\frac{N^{*}}{N^{*}+N_{r_{2}}}+\frac{T(p_{2})}{E(T|D^{*})}\frac{N_{r_{2}}}{N^{*}+N_{r_{2}}}
\]

\[
\frac{1-N_{r_{2}}}{1-N_{r_{1}}+N_{r_{2}}}\left[1-N_{r_{1}}+\frac{T(p_{2})}{E(T|D^{*})}N_{r_{2}}\right]<1+\frac{T(p_{2})}{E(T|D^{*})}N_{r_{2}}
\]

because $\frac{N_{r_{2}}}{N^{*}+N_{r_{2}}}<\frac{N_{r_{2}}}{N^{*}-N_{r_{1}}+N_{r_{2}}}$
the $L.H.S.$ weighs $\frac{T(p_{2})}{E(T|D^{*})}$ more. Thus proving
the inequality.%

\end{proof}
\begin{fact}
\label{fact:Xl-gt-Xu-B4} Under the assumptions of Sequence Scenario
B4 where $\frac{T(p_{1})}{E(T|D_{a})}<1$, and for $0\leq N_{r_{2}}\leq1=N^{*}$
and $\frac{T(p_{2})}{E(T|D^{*})}<1,$ there is a threshold $\tau$
on $\frac{T(p_{1})}{E(T|D_{a})}$ where for $\frac{T(p_{1})}{E(T|D_{a})}<\tau<1$
it will be the case that $X_{l}>X_{u}$ and where 

\[
\tau=\frac{\left(1-N_{r_{1}}\right)+\left(2+N_{r_{2}}-N_{r_{1}}\right)\frac{T(p_{2})}{E(T|D^{*})}N_{r_{2}}}{\left(1-N_{r_{1}}+N_{r_{2}}\right)\left(1+\frac{T(p_{2})}{E(T|D^{*})}N_{r_{2}}\right)}
\]

and further, under B4, $X_{u}<X_{l}$ for $N_{r_{2}}$ at and near
0. 
\end{fact}
\begin{proof}
 Set $N^{*}=N_{a}+N_{r_{1}}=1$, and show by algebraic derivation.
The formula for $X_{u}$ under B4, comes from equation \ref{eq:OrderCond-r1-r2}.
Substituting into the inequality \ref{eq:viable-inequality-kappa-lt-1}
and checking for non viability, yields,

\[
X_{l}=\frac{1}{\frac{1}{1+N_{r_{2}}}+\frac{T(p_{2})}{E(T|D^{*})}\frac{N_{r_{2}}}{1+N_{r_{2}}}}>\frac{1-N_{r_{1}}+\frac{T(p_{1})}{E(T|D_{a})}N_{r_{1}}}{\frac{1-N_{r_{1}}}{1-N_{r_{1}}+N_{r_{2}}}+\frac{T(p_{2})}{E(T|D^{*})}\frac{N_{r_{2}}}{1-N_{r_{1}}+N_{r_{2}}}}=X_{u}
\]

Evaluating this at $N_{r_{2}}=0$ confirms that $X_{u}<X_{l}$ for
$N_{r_{2}}$ at and near 0. 

Rearranging the inequality and reducing leaves,%

\[
\frac{1+N_{r_{2}}}{1-N_{r_{1}}+N_{r_{2}}}\frac{1-N_{r_{1}}+\frac{T(p_{2})}{E(T|D^{*})}N_{r_{2}}}{1+\frac{T(p_{2})}{E(T|D^{*})}N_{r_{2}}}>1-N_{r_{1}}+\frac{T(p_{1})}{E(T|D_{a})}N_{r_{1}}
\]

Further rearranging yields,

\[
\frac{-N_{r_{1}}N_{r_{2}}+\left(1+N_{r_{2}}\right)\frac{T(p_{2})}{E(T|D^{*})}N_{r_{2}}-\left(1-N_{r_{1}}+N_{r_{2}}\right)\frac{T(p_{2})}{E(T|D^{*})}N_{r_{2}}}{\left(1-N_{r_{1}}+N_{r_{2}}\right)\left(1+\frac{T(p_{2})}{E(T|D^{*})}N_{r_{2}}\right)}>N_{r_{1}}\left(\frac{T(p_{1})}{E(T|D_{a})}-1\right)
\]

and finally, 

\begin{equation}
\tau=\frac{\left(1-N_{r_{1}}\right)+\left(2-N_{r_{1}}+N_{r_{2}}\right)\frac{T(p_{2})}{E(T|D^{*})}N_{r_{2}}}{1-N_{r_{1}}+N_{r_{2}}+\left(1-N_{r_{1}}+N_{r_{2}}\right)\frac{T(p_{2})}{E(T|D^{*})}N_{r_{2}}}>\frac{T(p_{1})}{E(T|D_{a})}\label{eq:Tp1-Threshold}
\end{equation}

Thus the threshold $\tau$ provides the upper bounds on the range
of $\frac{T(p_{1})}{E(T|D_{a})}$ where $X_{l}>X_{u}$

It can be further shown that provided $\frac{T(p_{2})}{E(T|D^{*})}<1$, 

\[
\tau=\frac{\left(1-N_{r_{1}}\right)+\left(2+N_{r_{2}}-N_{r_{1}}\right)\frac{T(p_{2})}{E(T|D^{*})}N_{r_{2}}}{\left(1-N_{r_{1}}+N_{r_{2}}\right)\left(1+\frac{T(p_{2})}{E(T|D^{*})}N_{r_{2}}\right)}<1
\]

Expanding,

\[
\left(1-N_{r_{1}}\right)+\left(2+N_{r_{2}}-N_{r_{1}}\right)\frac{T(p_{2})}{E(T|D^{*})}N_{r_{2}}<1-N_{r_{1}}+N_{r_{2}}+\left(1-N_{r_{1}}+N_{r_{2}}\right)\frac{T(p_{2})}{E(T|D^{*})}N_{r_{2}}
\]

Reducing, 

\[
\frac{T(p_{2})}{E(T|D^{*})}N_{r_{2}}<N_{r_{2}}
\]

Yields,

\[
\frac{T(p_{2})}{E(T|D^{*})}<1
\]

Completing the proof.
\end{proof}
Fact \ref{fact:UnderB3-Xl-lt-Xu} shows that under the assumptions
of B3 the threshold for ordering $r_{2}$ in the sequence prior to
$D^{*}$ is above the threshold for $\Delta_{S}>0$. Although small,
there is always a range of values where $r_{2}$ is viable and has
a $\Delta_{S}>0$.

On the other hand Fact \ref{fact:Xl-gt-Xu-B4} shows that there is
possibility that $X_{u}<X_{l}$ for all $N_{r_{2}},$meaning that
there are no viable points where $\Delta_{S}>0$, and $X_{u}(N_{r_{2}})<X_{l}(N_{r_{2}})$
always holds for relatively small values of $N_{r_{2}}$. Fact \ref{fact:Xl-gt-Xu-B4}
provides a $S.C.$ for non-viability when $\kappa<1$.

---

The rest of appendix assumes $N^{*}=N_{a}+N_{r_{1}}=1$ unless otherwise
indicated. 

The conditions for $\kappa<1$ are analyzed using the following 2
facts. 
\begin{fact}
\label{fact:Xu(0)=00003D1 FD-lt-1}Under the assumptions of Sequence
Scenario B3 where $\frac{T(p_{1})}{E(T|D_{a})}=1$ and with $0<\frac{T(p_{2})}{E(T|D^{*})}<1$,
$X_{u}$ will have these properties: 

(i) $1\leq X_{u}(N_{r_{2}})=\frac{1-N_{r_{1}}+N_{r_{2}}}{1-N_{r_{1}}+\frac{T(p_{2})}{E(T|D^{*})}N_{r_{2}}}$
and $X_{u}(0)=1$, and 

(ii) $\text{if }N_{r_{1}}\leq\frac{T(p_{2})}{E(T|D^{*})},$ $\text{ }\frac{dX_{u}(N_{r_{2}})}{dN_{r_{2}}}<1\text{ for all }N_{r_{2}}>0$,
and

(iii) $\frac{d^{2}X_{u}(N_{r_{2}})}{dN_{r_{2}}^{2}}<0$.
\end{fact}
\begin{proof}
Show $X_{u}(0)=1$ by substitution.

\[
X_{u}(0)=\frac{1-N_{r_{1}}}{1-N_{r_{1}}}=1
\]

Next, derive and evaluate the first derivative, 

\[
\frac{dX_{u}(N_{r_{2}})}{dN_{r_{2}}}=\frac{1}{\left(1-N_{r_{1}}+\frac{T(p_{2})}{E(T|D^{*})}N_{r_{2}}\right)}-\frac{1}{\left(1-N_{r_{1}}+\frac{T(p_{2})}{E(T|D^{*})}N_{r_{2}}\right)^{2}}\left(1-N_{r_{1}}+N_{r_{2}}\right)\frac{T(p_{2})}{E(T|D^{*})}
\]

Rearranging and reducing yields the first derivative as,

\[
\frac{dX_{u}(N_{r_{2}})}{dN_{r_{2}}}=\frac{\left(1-N_{r_{1}}\right)\left(1-\frac{T(p_{2})}{E(T|D^{*})}\right)}{\left(1-N_{r_{1}}+\frac{T(p_{2})}{E(T|D^{*})}N_{r_{2}}\right)^{2}}>0
\]

And the second derivative as,

\begin{alignat*}{1}
\frac{d^{2}X_{u}(N_{r_{2}})}{dN_{r_{2}}^{2}} & =-2\frac{\left(1-N_{r_{1}}\right)\left(1-\frac{T(p_{2})}{E(T|D^{*})}\right)\frac{T(p_{2})}{E(T|D^{*})}}{\left(1-N_{r_{1}}+\frac{T(p_{2})}{E(T|D^{*})}N_{r_{2}}\right)^{3}}<0
\end{alignat*}

Thus $X_{u}$ is increasing at a decreasing rate. 

Evaluating the first derivative at $N_{r_{2}}=0$,%

\[
\frac{dX_{u}(0)}{dN_{r_{2}}}=\frac{\left(1-\frac{T(p_{2})}{E(T|D^{*})}\right)}{\left(1-N_{r_{1}}\right)}
\]

Evaluating when $\frac{dX_{u}(0)}{dN_{r_{2}}}<1$ yields the condition,%

\begin{equation}
N_{r_{1}}<\frac{T(p_{2})}{E(T|D^{*})}\label{eq:Cond-Nr1-Tp2}
\end{equation}

Thus since the first derivative of $X_{u}$ at 0 is $<1$ and $X_{u}(N_{r_{2}})$
is increasing at a decreasing rate, it has been shown that $\frac{dX_{u}(N_{r_{2}})}{dN_{r_{2}}}<1$
for all $N_{r_{2}}>0$ when $N_{r_{1}}<\frac{T(p_{2})}{E(T|D^{*})}$
.

\end{proof}
Thus Fact \ref{fact:Xu(0)=00003D1 FD-lt-1} (ii) provides a $S.C.$
to determine if $\frac{dX_{u}(N_{r_{2}})}{dN_{r_{2}}}<1$ at $N_{r_{2}}$.
Also, because $X_{u}(N_{r_{2}})$ is increasing at a decreasing rate,
even if condition \ref{eq:Cond-Nr1-Tp2} does not hold it is still
possible that there is a upper range of values for $N_{r_{2}}$ where
$\frac{dX_{u}(N_{r_{2}})}{dN_{r_{2}}}<1$ and $X_{u}(N_{r_{2}})<N_{r_{2}}<1$.
\begin{fact}
\label{fact:B4-Xu-FD-lt-1}Under the assumptions of Sequence Scenario
B4 where $\frac{T(p_{1})}{E(T|D_{a})}\leq1$ and with $0<\frac{T(p_{2})}{E(T|D^{*})}<1$,
$X_{u}$ will have these properties: 

(i) $1\leq X_{u}(N_{r_{2}})=\frac{1-N_{r_{1}}+\frac{T(p_{1})}{E(T|D_{a})}N_{r_{1}}}{\frac{1-N_{r_{1}}}{1-N_{r_{1}}+N_{r_{2}}}+\frac{T(p_{2})}{E(T|D^{*})}\frac{N_{r_{2}}}{1-N_{r_{1}}+N_{r_{2}}}}$
and $X_{u}(0)=1+N_{r_{1}}\left(\frac{T(p_{1})}{E(T|D_{a})}-1\right)$, 

(ii) $\text{\ensuremath{\frac{dX_{u}(N_{r_{2}})}{dN_{r_{2}}}>0} and \ensuremath{\frac{d^{2}X_{u}(N_{r_{2}})}{dN_{r_{2}}^{2}}<0}, and if }\frac{N_{r_{1}}\frac{T(p_{1})}{E(T|D_{a})}}{\left(1-N_{r_{1}}+N_{r_{1}}\frac{T(p_{1})}{E(T|D_{a})}\right)}<\frac{T(p_{2})}{E(T|D^{*})},$
$\text{ }\frac{dX_{u}(N_{r_{2}})}{dN_{r_{2}}}<1\text{ for all }N_{r_{2}}>0$%
, and 

(iii) $\frac{dX_{u}(N_{r_{2}})}{dN_{r_{2}}}$ in the B4 case is less
than $\frac{dX_{u}(N_{r_{2}})}{dN_{r_{2}}}$ in the B3 case by a factor
of $1+N_{r_{1}}\left(\frac{T(p_{1})}{E(T|D_{a})}-1\right)$, and the
limitations on $N_{r_{1}}$ and $\frac{T(p_{2})}{E(T|D^{*})}$ are
less restrictive in the B4 case. 
\end{fact}
\begin{proof}
The formula for $X_{u}(N_{r_{2}})$ under B4 is derived from equation
\ref{eq:OrderCond-r1-r2} and $X_{u}(0)$ can be determined by substitution.
Thus proving (i). For $N_{r_{2}}>0$ it is clear from the formula
that $X_{u}$ decreases as $\frac{T(p_{1})}{E(T|D_{a})}$ decreases
below $1$.

The first derivative is:%

\[
\frac{dX_{u}(N_{r_{2}})}{dN_{r_{2}}}=\frac{\left(1-N_{r_{1}}\right)\left(1-\frac{T(p_{2})}{E(T|D^{*})}\right)}{\left(1-N_{r_{1}}+\frac{T(p_{2})}{E(T|D^{*})}N_{r_{2}}\right)^{2}}\left[1+N_{r_{1}}\left(\frac{T(p_{1})}{E(T|D_{a})}-1\right)\right]>0
\]

The sign is evident given the restrictions on $N_{r_{1}}$ and $N_{r_{2}}.$
Evaluating $\frac{dX_{u}(N_{r_{2}})}{dN_{r_{2}}}$ for the B4 case
at at $N_{r_{2}}=0$ and checking when $<1$ yields, 

\[
\frac{\left(1-\frac{T(p_{2})}{E(T|D^{*})}\right)}{\left(1-N_{r_{1}}\right)}\left[1-N_{r_{1}}+\frac{T(p_{1})}{E(T|D_{a})}N_{r_{1}}\right]<1
\]

Rearranging yields,

\[
N_{r_{1}}<\frac{\frac{T(p_{2})}{E(T|D^{*})}}{\frac{T(p_{2})}{E(T|D^{*})}+\frac{T(p_{1})}{E(T|D_{a})}\left(1-\frac{T(p_{2})}{E(T|D^{*})}\right)}
\]

Proving (ii) and confirming that the condition for B4 is more restrictive
than the condition for B3. 

\[
\frac{\frac{T(p_{2})}{E(T|D^{*})}}{\frac{T(p_{2})}{E(T|D^{*})}+\frac{T(p_{1})}{E(T|D_{a})}\left(1-\frac{T(p_{2})}{E(T|D^{*})}\right)}>\frac{T(p_{2})}{E(T|D^{*})}
\]

Rearranging,%

\[
1>\frac{T(p_{2})}{E(T|D^{*})}
\]
'

The full restrictions on $N_{r_{1}}$ 

\[
N_{r_{1}}<\frac{\frac{T(p_{2})}{E(T|D^{*})}}{\frac{T(p_{2})}{E(T|D^{*})}+\frac{T(p_{1})}{E(T|D_{a})}-\frac{T(p_{2})}{E(T|D^{*})}\frac{T(p_{1})}{E(T|D_{a})}}<1
\]

or alternatively on $\frac{T(p_{2})}{E(T|D^{*})}$,%

\[
\frac{N_{r_{1}}\frac{T(p_{1})}{E(T|D_{a})}}{\left(1-N_{r_{1}}+N_{r_{1}}\frac{T(p_{1})}{E(T|D_{a})}\right)}<\frac{T(p_{2})}{E(T|D^{*})}
\]

Thus (iii) is obvious from comparing $\frac{dX_{u}(N_{r_{2}})}{dN_{r_{2}}}$
in the B3 and B4 cases.
\end{proof}
Fact \ref{fact:B4-Xu-FD-lt-1}, like Fact \ref{fact:Xu(0)=00003D1 FD-lt-1},
provides a $S.C.$ to determine if $\frac{dX_{u}(N_{r_{2}})}{dN_{r_{2}}}<1$
at $N_{r_{2}}$. However, Fact \ref{fact:B4-Xu-FD-lt-1} also shows
that under B4 there is a less restrictive range of values for $N_{r_{2}}$
where $\frac{dX_{u}(N_{r_{2}})}{dN_{r_{2}}}<1$ compared to the B3
conditions. See the figure below. 

\label{fig:Xu-Xl-comparison}\includegraphics[scale=0.75]{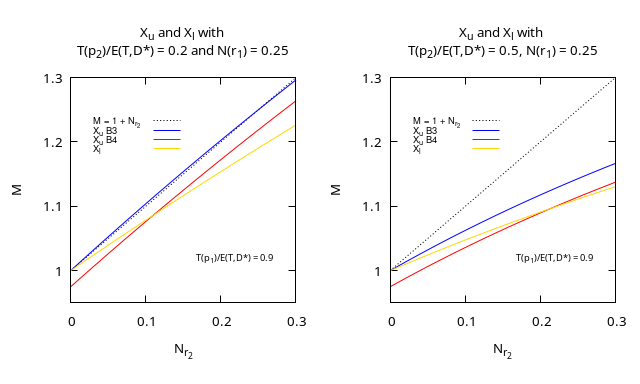}

\subsubsection*{Expressing in Terms of $\kappa$ }

The condition for $\Delta_{S}\leq0$ when $\kappa_{r_{2}}<1$, implying
$M_{r_{2}}<X_{l}$ , can be expressed as, \ref{eq:Xl-condition} and
\ref{eq:Delta-S-Inequality-with-kappa}

\begin{equation}
M_{r_{2}}=1+\kappa_{r_{2}}N_{r_{2}}<\frac{1+N_{r_{2}}}{1+\frac{T(p_{2})}{E(T|D^{*})}N_{r_{2}}}=X_{l}\label{eq:Delta-S-Inequality-with-kappa}
\end{equation}

where $N^{*}=1$ as above. Rearranging yields, 

\begin{equation}
\kappa_{r_{2}}<\frac{1-\frac{T(p_{2})}{E(T|D^{*})}}{1+\frac{T(p_{2})}{E(T|D^{*})}N_{r_{2}}}=X_{l\kappa}<1\label{eq:limit-on-kappa-r2}
\end{equation}

Similarly the order condition for non viability, $X_{u}<M_{ar_{2}}$,
can be expressed as.

\begin{equation}
X_{u}=\frac{1-N_{r_{1}}+N_{r_{2}}}{1-N_{r_{1}}+\frac{T(p_{2})}{E(T|D^{*})}N_{r_{2}}}<1+\kappa_{ar_{2}}N_{r_{2}}=M_{ar_{2}}\label{eq:Non-Viable-Inquality-with-kappa}
\end{equation}

Rearranging,%

\begin{equation}
X_{u\kappa}=\frac{1-\frac{T(p_{2})}{E(T|D^{*})}}{1-N_{r_{1}}+\frac{T(p_{2})}{E(T|D^{*})}N_{r_{2}}}<\kappa_{ar_{2}}\label{eq:limit-on-kappa-ar2}
\end{equation}

Note that $X_{u\kappa}<1$ %
when $N_{r_{1}}<\frac{T(p_{2})}{E(T|D^{*})}\left(N_{r_{2}}+1\right)$
thus $\kappa_{ar_{2}}>1$ assures non viability in that case. Intuitively
this is saying that the points with a $\kappa_{ar_{2}}>1$ are too
good to not have been in the sequence prior to $D^{*}$.

Likewise, $X_{u\kappa}>1$ when $N_{r_{1}}>\frac{T(p_{2})}{E(T|D^{*})}\left(N_{r_{2}}+1\right)$
thus assuring viability when $\kappa_{ar_{2}}<1$ .

Using alternative $X_{u}$ under B4 assumptions, 

\[
X_{u}=\frac{\left(1-N_{r_{1}}+N_{r_{2}}\right)\left(1-N_{r_{1}}+\frac{T(p_{1})}{E(T|D_{a})}N_{r_{1}}\right)}{1-N_{r_{1}}+\frac{T(p_{2})}{E(T|D^{*})}N_{r_{2}}}<1+\kappa_{ar_{2}}N_{r_{2}}=M_{ar_{2}}
\]

Rearranging yields an alternative threshold viability,%

\begin{equation}
X_{u\kappa}=\frac{\left(1-N_{r_{1}}+N_{r_{2}}\right)\left(\frac{T(p_{1})}{E(T|D_{a})}-1\right)\frac{N_{r_{1}}}{N_{r_{2}}}+\left(1-\frac{T(p_{2})}{E(T|D^{*})}\right)}{1-N_{r_{1}}+\frac{T(p_{2})}{E(T|D^{*})}N_{r_{2}}}<\kappa_{ar_{2}}
\end{equation}

The $L.H.S$. is $<1$ when this holds

\[
\left(1-N_{r_{1}}+N_{r_{2}}\right)\left(\frac{T(p_{1})}{E(T|D_{a})}-1\right)\frac{N_{r_{1}}}{N_{r_{2}}}+\left(1-\frac{T(p_{2})}{E(T|D^{*})}\right)<1-N_{r_{1}}+\frac{T(p_{2})}{E(T|D^{*})}N_{r_{2}}
\]

Rearranging,%

\[
\frac{T(p_{1})}{E(T|D_{a})}<\frac{-N_{r_{1}}+\frac{T(p_{2})}{E(T|D^{*})}\left(1+N_{r_{2}}\right)}{\left(1-N_{r_{1}}+N_{r_{2}}\right)}\left(\frac{N_{r_{2}}}{N_{r_{1}}}\right)+1
\]

This is always true given $\frac{T(p_{1})}{E(T|D_{a})}<1$ when $\frac{T(p_{2})}{E(T|D^{*})}\left(1+N_{r_{2}}\right)>N_{r_{1}}$.
Thus another confirmation of that relation. 

The functions $X_{l\kappa}(N_{r_{2}})$ and $X_{u\kappa}(N_{r_{2}})$
have been defined in equations \ref{eq:limit-on-kappa-r2} and \ref{eq:limit-on-kappa-ar2}
and will be used in further analysis below.
\begin{fact}
\label{fact:kappa-lt-1-only}Under the assumptions of B3 and when
$N_{r_{1}}<\frac{T(p_{2})}{E(T|D^{*})}\leq1$, or under B4 when $\frac{N_{r_{1}}\frac{T(p_{1})}{E(T|D_{a})}}{\left(1-N_{r_{1}}+N_{r_{1}}\frac{T(p_{1})}{E(T|D_{a})}\right)}<\frac{T(p_{2})}{E(T|D^{*})}$
, the inequality $\kappa_{r_{2}}<\kappa_{ar_{2}}<1$ will hold and
when $N_{r_{2}}>0$ it is a $N.C.$ for $\kappa_{ar_{2}}$ to be viable.
Further, $N_{r_{1}}<\frac{T(p_{2})}{E(T|D^{*})}\left(N_{r_{2}}+1\right)$
and $\kappa_{ar_{2}}>1$ implies non viability, while $N_{r_{1}}>\frac{T(p_{2})}{E(T|D^{*})}\left(N_{r_{2}}+1\right)$
and $\kappa_{ar_{2}}<1$ implies viability.
\end{fact}
\begin{proof}
Facts \ref{fact:Xu(0)=00003D1 FD-lt-1} and \ref{fact:B4-Xu-FD-lt-1}
provide conditions where $\frac{dX_{u}(0)}{dN_{r_{2}}}<1$. Under
those conditions since $\frac{d^{2}X_{u}(N_{r_{2}})}{dN_{r_{2}}^{2}}<0$
it is not possible for $1<\frac{dX_{u}(N_{r_{2}})}{dN_{r_{2}}}.$
Thus $X_{u}<1+N_{r_{2}}$ because $X_{u}(0)\leq1$ and$\frac{dX_{u}(N_{r_{2}})}{dN_{r_{2}}}<1$.
Thus, from Equation \ref{eq:Non-Viable-Inquality-with-kappa}, $1<\kappa_{ar_{2}}$
implies that $X_{u}(N_{r_{2}})<1+N_{r_{2}}<1+\kappa_{ar_{2}}N_{r_{2}}=M_{ar_{2}}$
making $\kappa_{ar_{2}}$ not viable. Further from Proposition \ref{prop:The-functions-Mr2-Mar2}
$M_{r_{2}}<M_{ar_{2}}$ in this case and this implies $\kappa_{r_{2}}<\kappa_{ar_{2}}$.
The further statements are shown to hold by inequality \ref{eq:limit-on-kappa-ar2}.
\end{proof}

Consistent with Fact \ref{fact:kappa-lt-1-only} if $X_{u\kappa}<\kappa_{ar_{2}}$
the point is not viable. So only need to consider $\kappa_{ar_{2}}\leq X_{u\kappa}$.

Further, if $\kappa_{r_{2}}<1$ when $\kappa_{ar_{2}}\leq X_{u\kappa}$,
the fact that $X_{u\kappa}>1$ is not important, as $\kappa_{r_{2}}<1$
leads away from a possible second equilibrium. Also when $\kappa_{r_{2}}<1$
the value of $x_{l\kappa}$ is relevant in determining $\Delta_{S}$.
So need to check if $\kappa_{r_{2}}<X_{l\kappa}$ as in inequality
\ref{eq:limit-on-kappa-r2}. %
So it will be useful to evaluate $\kappa_{ar_{2}}-\kappa_{r_{2}}$. 

Combining equations \ref{eq:Delta-S-Inequality-with-kappa} and \ref{eq:Non-Viable-Inquality-with-kappa}
yields the difference in kappas as,

\[
\kappa_{ar_{2}}-\kappa_{r_{2}}=\frac{1}{N_{r2}}\left(M_{ar_{2}}-M_{r_{2}}\right)
\]

Define $m$ as,

\[
\frac{M_{ar_{2}}}{M_{r_{2}}}=m=\left(\left(\frac{1+N_{r_{2}}}{1-N_{r_{1}}+N_{r_{2}}}\right)\frac{1-N_{r_{1}}+N_{r_{2}}\frac{c_{r_{2}}}{Q(D_{a})}}{\left(1-N_{r_{1}}+\frac{c_{r_{1}}}{Q(D_{a})}N_{r_{1}}\right)+N_{r_{2}}\frac{c_{r_{2}}}{Q(D_{a})}}\right)^{\alpha}>1
\]
Substituting $m$ in the first formula provides for an alternative
formulation for the difference in kappas,

\[
N_{r_{2}}\left(\kappa_{ar_{2}}-\kappa_{r_{2}}\right)=M_{ar_{2}}\left(1-\frac{1}{m}\right)
\]

\begin{equation}
N_{r_{2}}\left(\kappa_{ar_{2}}-\kappa_{r_{2}}\right)=\left(1+\kappa_{ar_{2}}N_{r_{2}}\right)\left(1-\frac{1}{m}\right)\label{eq:InterimResult-DiffKappas}
\end{equation}

Thus the difference in the kappas is a function of $\kappa_{ar_{2}}\text{, }N_{r_{2}}\text{, and }m$.

\begin{equation}
\kappa_{ar_{2}}-\kappa_{r_{2}}=\left(1-\frac{1}{m}\right)\left(\kappa_{ar_{2}}+\frac{1}{N_{r_{2}}}\right)\label{eq:diff-between-kappas}
\end{equation}

substituting for $\kappa_{ar_{2}}=X_{u\kappa}$ the comparison needed
for viable $r_{2}$ with $\kappa_{r_{2}}<1$ is,%

\[
X_{u\kappa}-1<\kappa_{ar_{2}}-\kappa_{r_{2}}=\left(1-\frac{1}{m}\right)\left(X_{u\kappa}+\frac{1}{N_{r_{2}}}\right)
\]

Rearranging,%

\[
X_{u\kappa}\frac{1}{m}<\left(1-\frac{1}{m}\right)\frac{1}{N_{r_{2}}}+1
\]

\[
X_{u\kappa}<\left(m-1\right)\frac{1}{N_{r_{2}}}+m
\]

\[
N_{r_{2}}X_{u\kappa}<m\left(1+N_{r_{2}}\right)-1
\]

Substituting for $X_{u\kappa}$ yields,

\[
N_{r_{2}}\frac{1-\frac{T(p_{2})}{E(T|D^{*})}}{1-N_{r_{1}}+\frac{T(p_{2})}{E(T|D^{*})}N_{r_{2}}}<m\left(1+N_{r_{2}}\right)-1
\]

Rearranging results in,

\begin{equation}
\frac{1-N_{r_{1}}+N_{r_{2}}}{\left(1+N_{r_{2}}\right)\left(1-N_{r_{1}}+\frac{T(p_{2})}{E(T|D^{*})}N_{r_{2}}\right)}<m\label{eq:diff-kappas-lt-Xu-1}
\end{equation}

\begin{fact}
\label{fact:SC-4-kappa-r2-lt-1}Inequality \ref{eq:diff-kappas-lt-Xu-1}
provides a $S.C.$ for $\kappa_{r_{2}}(N_{r_{2}})<1$.
\end{fact}
\begin{proof}
See derivation of equation \ref{eq:diff-kappas-lt-Xu-1} above.
\end{proof}
Note the values for $N_{r_{2}}$, $N_{r_{1}}$, and $\frac{T(p_{2})}{E(T|D^{*})}$
must be consistent with $X_{u\kappa}>1$. %
Thus when $N_{r_{2}}<\frac{N_{r_{1}}}{\frac{T(p_{2})}{E(T|D^{*})}}-1$
and combining with $0<N_{r_{2}}$ it will be the case that $X_{u\kappa}>1$
is only possible when $\frac{N_{r_{1}}}{\frac{T(p_{2})}{E(T|D^{*})}}>1$.
Consistent with Fact \ref{fact:Xu(0)=00003D1 FD-lt-1}.

Inequality \ref{eq:diff-kappas-lt-Xu-1} implies $\kappa_{r_{2}}<1$.
The $R.H.S.>1$ by Proposition \ref{prop:The-functions-Mr2-Mar2}
and it can be shown by algebraic inspection that the $L.H.S.<1$ when
$\frac{N_{r_{1}}}{\frac{T(p_{2})}{E(T|D^{*})}}-1<N_{r_{2}}$. That
would be for all $N_{r_{2}}\geq0$ when $N_{r_{1}}\leq\frac{T(p_{2})}{E(T|D^{*})}$.
However that is not the case here where $N_{r_{1}}>\frac{T(p_{2})}{E(T|D^{*})}$.
Nevertheless for $\frac{N_{r_{1}}}{\frac{T(p_{2})}{E(T|D^{*})}}>1$
but still close to $1$ there are realistic parameters that make $m$
sufficiently large enough to satisfy inequality \ref{eq:diff-kappas-lt-Xu-1}.
Also note that the $L.H.S.=1$ at $N_{r_{2}}=0$. Thus likely that
the inequality holds when $N_{r_{2}}$ is close to zero.%

See graph below for different $m(N_{r_{2}})$ and threshold functions
($L.H.S.$). The graph shows $m$ when $c_{2}$ is equal to $1.25$
and $2.0$. For $N_{r_{1}}=0.2$ , when $T(p_{2})$ is less but closer
to $N_{r_{1}}$ the threshold is below both lines for $m$, and when
$T(p_{2})$ is much less than $N_{r_{1}}$, the threshold is more
likely to be above $m$.
\noindent \begin{center}
\includegraphics[scale=0.75]{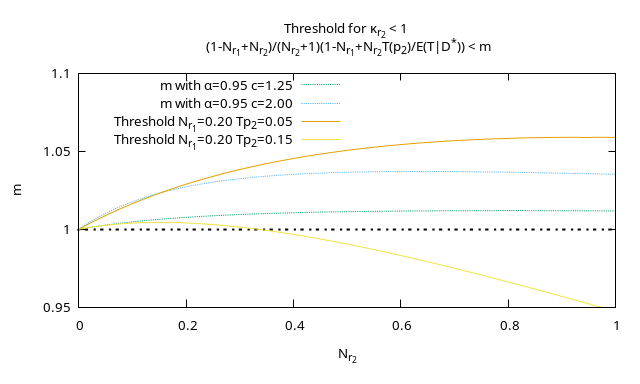}
\par\end{center}
\begin{fact}
\label{fact:KappaDif-gt-Xuk-Xlk}For $\kappa_{r_{2}}\leq1$ and $0<N_{r2}$,
if $\kappa_{ar_{2}}-\kappa_{r_{2}}>X_{u\kappa}(N_{r_{2}})-X_{l\kappa}(N_{r_{2}})$
or equivalently $M_{ar_{2}}-M_{r_{2}}>X_{u}(N_{r_{2}})-X_{l}(N_{r_{2}})$
a $\Delta_{S}>0$ is not viable. 
\end{fact}
\begin{proof}
The proof follows from rearranging $M_{ar_{2}}-M_{r_{2}}>X_{u}(N_{r_{2}})-X_{l}(N_{r_{2}})$
to get $M_{ar_{2}}-X_{u}(N_{r_{2}})>M_{r_{2}}-X_{l}(N_{r_{2}})$ which
indicates that either $M_{ar_{2}}>X_{u}(N_{r_{2}})$ or $M_{r_{2}}<X_{l}(N_{r_{2}})$
implying that one of the conditions \ref{eq:Xl-condition} or \ref{eq:Xu-condition}
are violated.
\end{proof}
Applying Fact \ref{fact:KappaDif-gt-Xuk-Xlk} and utilizing equation
\ref{eq:diff-between-kappas} yields,

\[
X_{u\kappa}(N_{r_{2}})-X_{l\kappa}(N_{r_{2}})<\kappa_{ar_{2}}-\kappa_{r_{2}}=\left(1-\frac{1}{m}\right)\left(\kappa_{ar_{2}}+\frac{1}{N_{r_{2}}}\right)
\]

Letting $X_{u\kappa}=\kappa_{ar_{2}}$ substituting and rearranging%

\label{eq:Threshold-Xu=0003BA---Xl=0003BA-lt-=0003BAar---=0003BAr-1}
\[
X_{u\kappa}(N_{r_{2}})-X_{l\kappa}(N_{r_{2}})m<\left(m-1\right)\left(\frac{1}{N_{r_{2}}}\right)
\]

Further rearranging yields,%

\[
\frac{N_{r_{2}}X_{u\kappa}(N_{r_{2}})+1}{N_{r_{2}}X_{l\kappa}(N_{r_{2}})+1}<m
\]

Substituting for $X_{u\kappa}(N_{r_{2}})$ and $X_{l\kappa}(N_{r_{2}})$,%

\begin{equation}
\frac{\left(1-N_{r_{1}}+N_{r_{2}}\right)\left(1+\frac{T(p_{2})}{E(T|D^{*})}N_{r_{2}}\right)}{\left(1+N_{r_{2}}\right)\left(1-N_{r_{1}}+\frac{T(p_{2})}{E(T|D^{*})}N_{r_{2}}\right)}<m\label{eq:=0003BAar-=0003BAr-gt-Xu=0003BA-Xl=0003BA-1}
\end{equation}

The $L.H.S$. represents the threshold on $m$ and it can be shown
{} that the $L.H.S.>1$. This is similar to \ref{eq:diff-kappas-lt-Xu-1}
but more restrictive as the inequality indicates a condition for a
non viable $r_{2}$.

An alternative use of the difference in kappas is express $\kappa_{ar_{2}}-\kappa_{r_{2}}$
as a fraction of $X_{u\kappa}(N_{r_{2}})-X_{l\kappa}(N_{r_{2}})$
to show the reduction in viable values of $\kappa_{r_{2}}$ where
$\Delta_{S}>0$ and the point $r_{2}$ is viable. Let $RVV$ be the
Reduction in Viable $\kappa_{r_{2}}$ Values where $\Delta_{S}>0$.
$RVV$ is relative to $X_{u\kappa}-X_{l\kappa}$. In other terms $RVV$
is the upper fraction of interval between $X_{u\kappa}(N_{r_{2}})-X_{l\kappa}(N_{r_{2}})$
that cannot be a value for $\kappa_{r_{2}}$ or the upper fraction
of that interval that cannot be a value for $\kappa_{ar_{2}}$.

\[
RVV=\frac{\kappa_{ar_{2}}-\kappa_{r_{2}}}{X_{u\kappa}(N_{r_{2}})-X_{l\kappa}(N_{r_{2}})}=\frac{m-1}{m}\frac{\left(X_{u\kappa}(N_{r_{2}})+\frac{1}{N_{r_{2}}}\right)}{X_{u\kappa}(N_{r_{2}})-X_{l\kappa}(N_{r_{2}})}
\]

Simplifying and rearranging yields, 

\begin{equation}
RVV=\left(\frac{m-1}{m}\right)\frac{\left(1+\frac{T(p_{2})}{E(T|D^{*})}N_{r_{2}}\right)\left(1-N_{r_{1}}+N_{r_{2}}\right)}{N_{r_{2}}N_{r_{1}}\left(1-\frac{T(p_{2})}{E(T|D^{*})}\right)}\label{eq:RVV}
\end{equation}

As previously noted, by Proposition \ref{prop:The-functions-Mr2-Mar2},
$m>1$ thus $\frac{m-1}{m}<1$. It can be shown%
{} that for all $N_{r_{2}}>0$,

\[
\frac{\left(1+\frac{T(p_{2})}{E(T|D^{*})}N_{r_{2}}\right)\left(1-N_{r_{1}}+N_{r_{2}}\right)}{N_{r_{2}}N_{r_{1}}\left(1-\frac{T(p_{2})}{E(T|D^{*})}\right)}>1
\]

Intuitively, it appears that for sufficiently large values of $m$
or relatively high values of $\frac{T(p_{2})}{E(T|D^{*})}$ and low
values of $N_{r_{1}}$, $RVV$ will be close to 1, and possibly $>1$
implying that there is no viable point $r_{2}$ that can be added
after $D^{*}$ as would be consistent with Fact \ref{fact:KappaDif-gt-Xuk-Xlk}. 

Explicitly checking when $RVV>1$ by combining the equation for $RVV$
and the inequality $RVV>1$ yields, 

\[
\left(1-N_{r_{1}}+N_{r_{2}}\right)+\frac{T(p_{2})}{E(T|D^{*})}N_{r_{2}}\left(1-N_{r_{1}}+N_{r_{2}}\right)>\frac{m}{m-1}\left(N_{r_{2}}N_{r_{1}}-N_{r_{2}}N_{r_{1}}\frac{T(p_{2})}{E(T|D^{*})}\right)
\]

Simplifying yields,

\[
\frac{1}{N_{r_{1}}}\left(1+N_{r_{2}}\right)\left(1+\frac{T(p_{2})}{E(T|D^{*})}N_{r_{2}}\right)>1+N_{r_{2}}\left(\frac{m-\frac{T(p_{2})}{E(T|D^{*})}}{m-1}\right)
\]

and alternatively,%

\[
1+\frac{T(p_{2})}{E(T|D^{*})}N_{r_{2}}+N_{r_{2}}+\frac{T(p_{2})}{E(T|D^{*})}\left(N_{r_{2}}\right)^{2}-N_{r_{1}}>N_{r_{1}}N_{r_{2}}\left(\frac{m-\frac{T(p_{2})}{E(T|D^{*})}}{m-1}\right)
\]

Solving for $m$ results in,%

\begin{equation}
m>\frac{1+\frac{T(p_{2})}{E(T|D^{*})}N_{r_{2}}(1-N_{r_{1}})+N_{r_{2}}+\frac{T(p_{2})}{E(T|D^{*})}\left(N_{r_{2}}\right)^{2}-N_{r_{1}}}{1+\frac{T(p_{2})}{E(T|D^{*})}N_{r_{2}}+N_{r_{2}}+\frac{T(p_{2})}{E(T|D^{*})}\left(N_{r_{2}}\right)^{2}-N_{r_{1}}\left(1+N_{r_{2}}\right)}=\frac{\left(1-N_{r_{1}}+N_{r_{2}}\right)\left(1+\frac{T(p_{2})}{E(T|D^{*})}N_{r_{2}}\right)}{\left(1+N_{r_{2}}\right)\left(1-N_{r_{1}}+\frac{T(p_{2})}{E(T|D^{*})}N_{r_{2}}\right)}\label{eq:m-threshold-for-RVV-gt-0}
\end{equation}

Comparing \ref{eq:m-threshold-for-RVV-gt-0} with the check of $\kappa_{r_{2}}<1$
condition \ref{eq:diff-kappas-lt-Xu-1} %
the threshold here is $\left(1+\frac{T(p_{2})}{E(T|D^{*})}N_{r_{2}}\right)$
times more restrictive. This will be lower, for lower values of $N_{r_{2}}$
and $\frac{T(p_{2})}{E(T|D^{*})}$.

\begin{fact}
\label{fact:Xuk-Xlk-properties}$\frac{\partial(X_{u\kappa}-X_{l\kappa})}{\partial T(p_{2})}<0$
and $\frac{\partial^{2}(X_{u\kappa}-X_{l\kappa})}{\partial T(p_{2})^{2}}>0$
\end{fact}
\begin{proof}
$X_{u\kappa}-X_{l\kappa}=\frac{1-\frac{T(p_{2})}{E(T|D^{*})}}{1-N_{r_{1}}+\frac{T(p_{2})}{E(T|D^{*})}N_{r_{2}}}-\frac{1-\frac{T(p_{2})}{E(T|D^{*})}}{1+\frac{T(p_{2})}{E(T|D^{*})}N_{r_{2}}}=\left(1-\frac{T(p_{2})}{E(T|D^{*})}\right)\frac{1+\frac{T(p_{2})}{E(T|D^{*})}N_{r_{2}}-\left(1-N_{r_{1}}+\frac{T(p_{2})}{E(T|D^{*})}N_{r_{2}}\right)}{\left(1+\frac{T(p_{2})}{E(T|D^{*})}N_{r_{2}}\right)\left(1-N_{r_{1}}+\frac{T(p_{2})}{E(T|D^{*})}N_{r_{2}}\right)}=\left(1-\frac{T(p_{2})}{E(T|D^{*})}\right)\frac{N_{r_{1}}}{\left(1+\frac{T(p_{2})}{E(T|D^{*})}N_{r_{2}}\right)\left(1-N_{r_{1}}+\frac{T(p_{2})}{E(T|D^{*})}N_{r_{2}}\right)}$

$\frac{\partial(X_{u\kappa}-X_{l\kappa})}{\partial T(p_{2})}=\frac{-N_{r_{1}}}{\left(1+\frac{T(p_{2})}{E(T|D^{*})}N_{r_{2}}\right)\left(1-N_{r_{1}}+\frac{T(p_{2})}{E(T|D^{*})}N_{r_{2}}\right)}-\frac{N_{r_{2}}\left(1-\frac{T(p_{2})}{E(T|D^{*})}\right)\left[1+\frac{T(p_{2})}{E(T|D^{*})}N_{r_{2}}+1-N_{r_{1}}+\frac{T(p_{2})}{E(T|D^{*})}N_{r_{2}}\right]}{\left(1+\frac{T(p_{2})}{E(T|D^{*})}N_{r_{2}}\right)^{2}\left(1-N_{r_{1}}+\frac{T(p_{2})}{E(T|D^{*})}N_{r_{2}}\right)^{2}}<0$

$\frac{\partial^{2}(X_{u\kappa}-X_{l\kappa})}{\partial T(p_{2})^{2}}=\frac{N_{r_{2}}\left(1+\frac{T(p_{2})}{E(T|D^{*})}N_{r_{2}}\right)+N_{r_{2}}\left(1-N_{r_{1}}+\frac{T(p_{2})}{E(T|D^{*})}N_{r_{2}}\right)}{\left(1+\frac{T(p_{2})}{E(T|D^{*})}N_{r_{2}}\right)^{2}\left(1-N_{r_{1}}+\frac{T(p_{2})}{E(T|D^{*})}N_{r_{2}}\right)^{2}}+\frac{N_{r_{2}}\left[1+\frac{T(p_{2})}{E(T|D^{*})}N_{r_{2}}+1-N_{r_{1}}+\frac{T(p_{2})}{E(T|D^{*})}N_{r_{2}}\right]-N_{r_{2}}^{2}\left(1-\frac{T(p_{2})}{E(T|D^{*})}\right)}{\left(1+\frac{T(p_{2})}{E(T|D^{*})}N_{r_{2}}\right)^{2}\left(1-N_{r_{1}}+\frac{T(p_{2})}{E(T|D^{*})}N_{r_{2}}\right)^{2}}+$

$+\frac{2N_{r_{2}}\left(1-\frac{T(p_{2})}{E(T|D^{*})}\right)\left[1+\frac{T(p_{2})}{E(T|D^{*})}N_{r_{2}}+1-N_{r_{1}}+\frac{T(p_{2})}{E(T|D^{*})}N_{r_{2}}\right]^{2}}{\left(1+\frac{T(p_{2})}{E(T|D^{*})}N_{r_{2}}\right)^{4}\left(1-N_{r_{1}}+\frac{T(p_{2})}{E(T|D^{*})}N_{r_{2}}\right)^{4}}>0$

$\frac{\partial^{2}(X_{u\kappa}-X_{l\kappa})}{\partial T(p_{2})^{2}}>0$
\end{proof}
\textbf{Summary of the the case where $M(D')>M(D^{*})$}

Facts \ref{fact:UnderB3-Xl-lt-Xu} and \ref{fact:Xl-gt-Xu-B4} provide
a $S.C.$ for a range of $N_{r_{2}}$ where a $\Delta_{S}>0$ is not
possible because $X_{u}<X_{l}$ as in inequality chain \ref{eq:viable-inequality-kappa-lt-1}
and \ref{eq:eq:viable-inequality-kappa-gt-1}. With assumptions under
B4 as long as $T(p_{1})<E(T|D^{*})$ there is at least a small range
starting at $N_{r_{2}}=0$ where $\Delta_{S}>0$ is not possible.
The inequalities \ref{eq:viable-inequality-kappa-lt-1} and \ref{eq:eq:viable-inequality-kappa-gt-1}
also provide restrictions on $M_{ar_{2}}$ for point $r_{2}$ to be
viable, and restrictions on $M_{r_{2}}$ for $\Delta_{S}>0$. 

Facts \ref{fact:Xu(0)=00003D1 FD-lt-1}, \ref{fact:B4-Xu-FD-lt-1},
and \ref{fact:kappa-lt-1-only} provide conditions where $\kappa_{ar_{2}}<1$
for all $N_{r_{2}}$, and thus conditions where points with $\kappa_{ar_{2}}>1$
are not viable. 

When $\kappa_{r_{2}}<1$ it indicates that 1) an extended equilibrium
is not possible and 2) that the alternative function for $\Delta_{S}$
(when $\kappa_{r_{2}}>1$) will not apply. 

Fact \ref{fact:KappaDif-gt-Xuk-Xlk} shows that the narrow range between
$X_{u}$ and $X_{l}$ or between $X_{u\kappa}$ and $X_{l\kappa}$
is even narrower owing to the relation between $\kappa_{ar_{2}}$
and $\kappa_{r_{2}}$ as determined by equation \ref{eq:RVV}. 

\par\end{flushleft}
\typeout{get arXiv to do 4 passes: Label(s) may have changed. Rerun}
\end{document}